\documentclass[11pt,a4paper]{article}
\usepackage{amssymb}
\usepackage{amsmath}
\usepackage{amsfonts}
\usepackage{bbm}
\usepackage{amsthm}
\usepackage{mathrsfs}
\usepackage{hyperref}
\usepackage{color}
\usepackage[margin=2.41cm]{geometry}
\usepackage[all,cmtip]{xy}
\usepackage[utf8]{inputenc}
\usepackage{graphicx}
\usepackage{varwidth}
\usepackage{comment}
\usepackage{slashed}
\usepackage{upgreek}
\usepackage{rotating}
\usepackage{tikz-cd}

\usepackage[shortlabels]{enumitem}

\usepackage{tikz}
\usetikzlibrary{arrows.meta,bending,decorations.markings}
\tikzset{
    arc arrow/.style args={%
    to pos #1 with length #2}{
    decoration={
        markings,
         mark=at position 0 with {\pgfextra{%
         \pgfmathsetmacro{\tmpArrowTime}{#2/(\pgfdecoratedpathlength)}
         \xdef\tmpArrowTime{\tmpArrowTime}}},
        mark=at position {#1-\tmpArrowTime} with {\coordinate(@1);},
        mark=at position {#1-2*\tmpArrowTime/3} with {\coordinate(@2);},
        mark=at position {#1-\tmpArrowTime/3} with {\coordinate(@3);},
        mark=at position {#1} with {\coordinate(@4);
        \draw[-{Triangle[length=#2,bend]}]       
        (@1) .. controls (@2) and (@3) .. (@4);},
        },
     postaction=decorate,
     },
fermion arc arrow/.style={arc arrow=to pos #1 with length 2.5mm},
Vertex/.style={fill,circle,inner sep=1.5pt},
insert vertex/.style={decoration={
        markings,
         mark=at position #1 with {\node[Vertex]{};},
        },
     postaction=decorate}     
}

\definecolor{darkred}{rgb}{0.8,0.1,0.1}
\hypersetup{
     colorlinks=false,         
     linkcolor=darkred,
     citecolor=blue,
}

\theoremstyle{plain}
\newtheorem{theo}{Theorem}[section]

\newtheorem{cor}[theo]{Corollary}

\theoremstyle{definition}
\newtheorem{defi}[theo]{Definition}

\newtheorem{exx}[theo]{Example}
\newenvironment{ex}
  {\pushQED{\qed}\exx}
  {\popQED\endexx}

\newtheorem{remm}[theo]{Remark}
\newenvironment{rem}
  {\pushQED{\qed}\remm}
  {\popQED\endremm}

\numberwithin{equation}{section}

\def\nn{\nonumber}

\def\bbK{\mathbb{K}}
\def\bbR{\mathbb{R}}
\def\bbC{\mathbb{C}}

\def\bbZ{\mathbb{Z}}
\def\bbT{\mathbb{T}}
\def\bbS{\mathbb{S}}

\def\ii{{\,{\rm i}\,}}
\def\e{{\,\rm e}\,}

\def\Ch{\mathsf{Ch}}

\def\Hom{\mathrm{Hom}}
\def\hom{\underline{\mathrm{hom}}}

\def\Sym{\mathrm{Sym}}
\def\Obs{\mathrm{Obs}}

\def\cl{\mathrm{cl}}
\def\BV{\mathrm{BV}}

\def\id{\mathrm{id}}

\def\dd{\mathrm{d}}
\def\oone{\mathbbm{1}}

\def\Tr{\mathrm{Tr}}

\def\CAlg{\mathsf{CAlg}}
\def\Mod{\mathsf{Mod}}

\def\su{\mathfrak{su}}

\def\ra{\triangleright}

\def\sk{\vspace{1mm}}

\newcommand\und[1]{\underline{#1}}

\newcommand{\cyc}[2]{\langle\!\langle #1 , #2\rangle\!\rangle}

\makeatletter
\let\@fnsymbol\@alph
\makeatother

\title{
Batalin-Vilkovisky quantization of fuzzy field theories
}

\author{
Hans Nguyen$^{1,a}$, Alexander Schenkel$^{1,b}$ and Richard J.~Szabo$^{2,c}$\vspace{4mm}\\
{\small $^{1}$ School of Mathematical Sciences, University of Nottingham,}\\
{\small University Park, Nottingham NG7 2RD,  UK.}\vspace{2mm}\\
{\small $^{2}$ Department of Mathematics, Heriot-Watt University, }\\
{\small Colin Maclaurin Building, Riccarton, Edinburgh EH14 4AS, UK.}\\
{\small \& Maxwell Institute for Mathematical Sciences,  Edinburgh, UK.}\\
{\small \& Higgs Centre for Theoretical Physics, Edinburgh, UK.}\vspace{4mm}\\
{\small \begin{tabular}{ll}
Email: & ${}^a$~\texttt{hans.nguyen@nottingham.ac.uk}\\
& ${}^b$~\texttt{alexander.schenkel@nottingham.ac.uk}\\
& ${}^c$~\texttt{R.J.Szabo@hw.ac.uk}\vspace{2mm}
\end{tabular}
}
}

\date{November 2021}

\begin{document}

\maketitle


\begin{abstract}
\noindent We apply the modern Batalin-Vilkovisky quantization techniques of Costello and Gwilliam to noncommutative field theories in the finite-dimensional case of fuzzy spaces. We further develop a generalization of this framework to theories that are equivariant under a triangular Hopf algebra symmetry, which in particular leads to quantizations of finite-dimensional analogs of the field theories proposed recently through the notion of `braided $L_\infty$-algebras'. The techniques are illustrated by computing perturbative correlation functions for scalar and Chern-Simons theories on the fuzzy $2$-sphere, as well as for braided scalar field theories on the fuzzy $2$-torus.
\end{abstract}

\begin{flushright}
{\sf\small EMPG--21--09}
\end{flushright}


\renewcommand{\baselinestretch}{0.8}\normalsize
\tableofcontents
\renewcommand{\baselinestretch}{1.0}\normalsize

\newpage


\section{Introduction}
Noncommutative quantum field theories are well-known to exhibit many 
novel features not present in conventional quantum field theory, 
see e.g.~\cite{Sza03} for a review. In particular, despite many 
years of extensive investigation, the quantization of noncommutative 
gauge theories is not completely understood, see e.g.~\cite{BKSW10} 
for a review. In this paper we will offer a new perspective on this 
problem by applying modern incarnations of Batalin-Vilkovisky (BV) 
quantization to noncommutative field theories. We follow the approach of 
Costello and Gwilliam~\cite{CG16,Gwi12}. We treat only \emph{fuzzy} field 
theories, which are by definition finite-dimensional systems, i.e.\ 
matrix models, and so can be quantized in a completely rigorous way 
while avoiding the analytic issues involved when dealing with continuum 
field theories. These examples will serve to nicely illustrate our 
formalism while avoiding much technical clutter. We give a general 
review of these quantization techniques in Section~\ref{sec:prelim}. 
A finite-dimensional BV formalism for certain matrix models is also 
discussed in~\cite{Ise19a,Ise19b}, and related to the spectral triple 
formulation of noncommutative geometry in~\cite{IvS17}.
\sk

Our approach is inspired in part by recent analyses of noncommutative 
field theories in the framework of $L_\infty$-algebras. The classical 
$L_\infty$-algebra formulation of the standard noncommutative gauge 
theories was originally presented in~\cite{BBKL18}. A new notion of 
`braided $L_\infty$-algebra' was defined more recently in~\cite{DCGRS20,DCGRS21}, 
where it was used to construct `braided field theories' which are equivariant 
under the action of a triangular Hopf algebra and involve fields with braided 
noncommutativity. We would like to stress that the term `braided' in 
these papers, as well as in our present one, is used to refer to algebraic structures
and field theories that are defined in a {\em symmetric} braided monoidal category,
where however the triangular $R$-matrix is non-trivial in the sense that it is not the identity.
A generalization of \cite{DCGRS20,DCGRS21} and the results of the present paper
to the truly (i.e.\ non-symmetric) braided case is considerably more complicated, see Section 
\ref{sec:conclusion} for more comments.
To handle the theories in \cite{DCGRS20,DCGRS21}, we develop a braided version 
of the BV formalism in Section~\ref{sec:braidedBV}, which fully captures 
their perturbative quantization and explicitly computes their correlation 
functions. Our perspective circumvents the issues involved 
in constructing the classical (Maurer-Cartan) solution space in the equivariant 
and braided setting that were pointed out in~\cite{DCGRS21}, as it characterizes 
this space by its equivariant function dg-algebra. For scalar field theories, 
our approach agrees with Oeckl's algebraic approach to (symmetric) braided 
quantum field theory~\cite{Oec01,Oec00}, which is based on a braided 
generalization of Wick's theorem and Gaussian integration. Our framework 
has the advantage of being able to straightforwardly treat theories with 
gauge symmetries, which are not addressed in Oeckl's approach. 
\sk

In this paper we treat prototypical fuzzy versions of both types of 
noncommutative field theories within the BV formalism. These are defined 
respectively on the fuzzy $2$-sphere (Section~\ref{sec:fuzzysphere}) and 
on the fuzzy $2$-torus with a non-trivial $R$-matrix (Section~\ref{sec:fuzzytorus}). 
\sk

On the fuzzy sphere we illustrate the finite-dimensional BV formalism 
for both scalar and gauge field theories. We study in detail the example 
of $\Phi^4$-theory where we reproduce the known $2$-point function at 
$1$-loop order obtained through more traditional techniques~\cite{CMS01}. 
In particular, this nicely illustrates how BV quantization captures the 
known distinction between planar and non-planar loop corrections in the 
standard noncommutative field theories, see e.g.~\cite{Sza03} for a review. 
While it is possible to treat Yang-Mills theory on the fuzzy sphere using 
our techniques, for illustration we consider the simpler example of 
Chern-Simons gauge theory, which was introduced in~\cite{ARS00,GMS01}. 
The quantization of this theory has so far only been briefly mentioned 
in~\cite{GMS02}, without any detailed analysis. In this paper we provide 
a complete framework in which the perturbative correlation functions of 
Chern-Simons theory on the fuzzy $2$-sphere can be computed. Our results
share many similarities with the Chern-Simons model on finite-dimensional 
commutative dg Frobenius algebras studied in~\cite{CM10}.
\sk

On the fuzzy torus we apply our finite-dimensional braided BV formalism 
to scalar field theories, which serve to illustrate a host of novelties 
compared to the standard noncommutative field theories. In particular, our 
approach produces fuzzy versions of symmetric braided quantum field theories 
in the continuum~\cite{Oec01}. We observe the absence of the 
notion of non-planar loop corrections as a consequence of the braided symmetry, 
as pointed out for twist deformed field theories in~\cite{Oec00} through 
algebraic means, and later by~\cite{Bal+07} through more heuristic methods. 
We stress that, in our case, this does not imply that there are no non-trivial 
braiding effects in the correlation functions; we illustrate this through 
explicit examples. Nevertheless, we expect the situation to be much different 
for gauge theories (as also suggested by~\cite{Bal+07}), as in this case even the 
classical braided field theories generally follow a different pattern from the 
conventional noncommutative gauge theories~\cite{DCGRS21}. Unfortunately, unlike 
the fuzzy sphere, we are not aware of any construction of a differential calculus 
on the fuzzy torus that could be used to define versions of the standard gauge 
theories, in contrast to its continuum version, i.e.\ the noncommutative torus; 
for a rigorous discussion of this point, see~\cite{LLS01}.


\section{\label{sec:prelim}Batalin-Vilkovisky quantization}
In this section we review the necessary background on the Batalin-Vilkovisky 
(BV) formalism and related tools for the computation of correlation 
functions of quantum field theories.

\subsection{\label{subsec:cochain}Cochain complexes}
We briefly recall some preliminary facts about cochain complexes. 
Let us fix a field $\bbK$ of characteristic $0$. 
We denote by $\Ch_\bbK$ the category of 
cochain complexes of $\bbK$-vector spaces,
i.e.\ we work with cohomological degree 
conventions in which the differential has degree $+1$. 
\sk

Recall that $\Ch_\bbK$ is a closed symmetric monoidal category.
The tensor product $V\otimes W\in\Ch_\bbK$ of two cochain complexes
$V,W\in\Ch_\bbK$ is defined by the graded vector space
\begin{subequations}\label{eqn:tensorproductChK}
\begin{flalign}
(V\otimes W)^n \,:= \, \bigoplus_{m\in \bbZ}\, V^m\otimes W^{n-m}\quad,
\end{flalign}
for all $n\in\bbZ$, and the differential
\begin{flalign}
\dd_{V\otimes W}(v\otimes w) \,:=\, (\dd_V v)\otimes w + (-1)^{\vert v\vert}\,v\otimes(\dd_W w)\quad,
\end{flalign}
\end{subequations}
where $\vert v\vert \in\bbZ$ denotes the degree of a homogeneous element $v\in V$.
The monoidal unit is $\bbK\in\Ch_\bbK$, regarded as a cochain complex in degree $0$,
and the symmetric braiding is given by the Koszul sign rule
\begin{flalign}\label{eqn:braidingChK}
\tau \,:\,V\otimes W~\longrightarrow~W\otimes V~,~~v\otimes w~\longmapsto~(-1)^{\vert v\vert\,\vert w\vert}\, w\otimes v\quad.
\end{flalign}
The mapping complex (or internal hom) $\hom(V,W)\in\Ch_\bbK$
between two cochain complexes $V,W\in\Ch_\bbK$ is defined by
the graded vector space
\begin{subequations}\label{eqn:mappingChK}
\begin{flalign}
\hom(V,W)^n \,:=\, \prod_{m\in\bbZ}\, \Hom_\bbK\big(V^m,W^{n+m}\big)\quad,
\end{flalign}
for all $n\in\bbZ$,  and the `adjoint' differential
\begin{flalign}
\partial(\zeta)\,:=\,\dd_W\circ \zeta - (-1)^{\vert \zeta \vert}\, \zeta\circ \dd_V\quad.
\end{flalign}
\end{subequations}
The $0$-cocycles of this complex, i.e.\ $f\in\hom(V,W)^0$ with $\partial(f)=0$, 
are precisely the cochain maps $f: V\to W$.  A \emph{cochain homotopy} between two cochain maps
$f,g : V\to W$ is a $(-1)$-cochain $h\in \hom(V,W)^{-1}$ satisfying $f-g=\partial(h)$.
\sk

Recall also that, given any $V\in\Ch_\bbK$ and $k\in\bbZ$,
the $k$-shifted cochain complex $V[k]\in\Ch_\bbK$ is defined
by $V[k]^n := V^{n+k}$, for all $n\in\bbZ$,  and the differential
$\dd_{V[k]} := (-1)^k\,\dd_V$.  Observe that $V[k]\cong\bbK[k]\otimes V$.
Given any cochain map $f : V\to W$, the $k$-shifted cochain map
$f[k]: V[k]\to W[k]$ is defined by $f[k]:=f$.

\subsection{\label{subsec:BV}Finite-dimensional BV formalism}
In the following we provide an elementary and 
self-contained review of the BV quantization techniques
developed by Costello and Gwilliam \cite{CG16,Gwi12}. 
Since the focus of this paper is on matrix models,
which are finite-dimensional systems,  we can work
in a purely algebraic setting and thereby avoid 
the functional analytic subtleties for continuum field theories
addressed in~\cite{CG16,Gwi12}. We also refer to~\cite{GJF18} for a
useful earlier exposition of BV quantization in finite dimensions.
\sk

Let us start by recalling the definition of a classical 
free BV theory from~\cite{CG16,Gwi12}.
\begin{defi}\label{def:linearBV}
A {\em free BV theory} is a cochain complex $E\in\Ch_\bbK$, with differential
denoted by $\dd_E = -Q$,\footnote{We have included a minus sign in the definition of
$Q$ in order to avoid unpleasant sign factors in the dual differential on the observables,
which we shall use more frequently in the present paper.} together with a cochain map 
$\langle\,\cdot\,,\,\cdot\,\rangle : E\otimes E\to\bbK[-1]$ that is non-degenerate and
antisymmetric, i.e.\ $\langle\,\cdot\,, \,\cdot\,\rangle =- \langle\,\cdot\,,\,\cdot\,\rangle\circ \tau $ 
where $\tau$ denotes the symmetric braiding on $\Ch_\bbK$.
\end{defi}
\begin{rem}\label{rem:dersolution}
The complex $E=(E,-Q)$ should be interpreted as a {\em derived} solution space.
The elements in degree $0$ are the fields of the theory, while the negative degrees
encode the ghost fields and the positive degrees encode the antifields.
See Section \ref{sec:fuzzysphere} for some explicit examples.
The pairing $\langle\,\cdot\,,\,\cdot\,\rangle$ plays the role of a $(-1)$-shifted symplectic structure. 
\sk

Since we are interested mainly in matrix models, which are finite-dimensional systems,
we shall assume implicitly throughout the whole paper that each homogeneous component 
$E^n$ of the complex $E$ is a finite-dimensional vector space and also that
the complex itself is bounded from above and below, i.e.\ there exists 
some positive integer $N\in\bbZ_{>0}$ such that $E^n =0$ for both $n>N$ and $n<-N$.
As a consequence, $E$ is a perfect complex and thus dualizable.
\end{rem}

To every free BV theory $(E,-Q,\langle\,\cdot\,, \,\cdot\,\rangle)$ one can assign
a commutative dg-algebra of polynomial observables. This is defined
as the symmetric algebra $\Sym\,E^\ast \in\CAlg(\Ch_\bbK)$, where $E^\ast$ denotes 
the dual cochain complex of $E=(E,-Q)$. Making use of the non-degenerate pairing
$\langle\,\cdot\,, \,\cdot\,\rangle$, we observe that the dual complex
$E^\ast\cong E[1]$ is isomorphic to the $1$-shifted complex. In fact, the duality pairing
between $E[1]$ and $E$ is given by
\begin{flalign}
\xymatrix@C=3.5em{
E[1]\otimes E \,\cong \,\bbK[1]\otimes E\otimes E \ar[r]^-{\id\otimes \langle\,\cdot\,, \,\cdot\,\rangle}
~&~ \bbK[1]\otimes \bbK[-1]\,\cong\,\bbK
}\quad.
\end{flalign}
Let us also recall that the differential on 
$E[1]$ acquires an additional minus sign 
due to the shifting conventions from Section \ref{subsec:cochain}, hence
\begin{flalign}
\dd_{E[1]} \,=\, - \dd_{E} \,=\, Q\quad.
\end{flalign}
With the usual abuse of notation, we denote the differential
on the dg-algebra of polynomial observables
$\Sym\,E^\ast\cong \Sym\,E[1]$ by the same symbol $Q$ as the differential on $E[1]$.
As we explain in more detail below,  the pairing $\langle\,\cdot\,, \,\cdot\,\rangle$
induces a shifted Poisson bracket on $\Sym \,E[1]$ in the form of a $P_0$-algebra structure.
We briefly recall this crucial concept and refer to \cite{Saf17} for more details
on shifted Poisson structures.
\begin{defi}\label{def:P0algebraTrad}
A {\em $P_0$-algebra} is a commutative dg-algebra $A\in  \CAlg(\Ch_\bbK)$
together with a Lie bracket $[\,\cdot\, , \,\cdot\,] : A[-1]\otimes A[-1]\to A[-1]$ on the shifted
cochain complex $A[-1]$ such that $[a,\,\cdot\,]$ defines a derivation
on $A$ for each $a\in A$.
\end{defi}
\begin{rem}\label{rem:P0algebraTrad}
Observe that the datum of a 
Lie bracket $[\,\cdot\, , \,\cdot\,]$ on $A[-1]\cong \bbK[-1]\otimes A$
is equivalent to the datum of the cochain map 
$\{\, \cdot\, , \,\cdot\,\} : A\otimes A\to A[1]$ defined by 
\begin{flalign}
\xymatrix@C=3em{
\ar[d]_-{\cong}A\otimes A \ar[r]^-{\{\,\cdot\, , \,\cdot\, \}} ~&~ A[1] \\
(A[-1]\otimes A[-1])[2] \ar[r]_-{[\,\cdot\, , \,\cdot\, ][2]}~&~A[-1][2]\ar[u]_-{\cong}
}
\end{flalign}
where the left vertical isomorphism uses the symmetric braiding $\tau$
and is given explicitly by $a\otimes b \mapsto (-1)^{\vert a\vert}\,a \otimes b$.
When expressed in terms of the bracket $\{\,\cdot\,,\,\cdot\,\}$, the $P_0$-algebra axioms
take the following form:
\begin{itemize}
\item[(i)] Compatibility with the differential: For all $a,b\in A$,
\begin{flalign}
-\dd\{a,b\}\,=\,\{\dd a, b\} + (-1)^{\vert a\vert } \,\{a,\dd b\}\quad,
\end{flalign}
where $\dd$ denotes the differential on the unshifted cochain complex $A$
and the minus sign on the left-hand side is due to our shifting conventions
from Section \ref{subsec:cochain}.

\item[(ii)] Symmetry: For all $a,b\in A$,
\begin{flalign}
\{ a, b \} \, =\,  (-1)^{\vert a\vert\,\vert b\vert}\, \{b,a\}\quad,
\end{flalign} 
i.e.\ antisymmetry of $[\,\cdot\, , \,\cdot\,]$ on $A[-1]\otimes A[-1]$ implies symmetry of $\{\,\cdot\, ,\,\cdot\,\}$ on $A\otimes A$.

\item[(iii)] Jacobi identity: For all $a,b,c\in A$,
\begin{flalign}
(-1)^{\vert a\vert \,\vert c\vert + \vert b\vert + \vert c\vert }~\{a,\{b,c\}\} 
& +(-1)^{\vert b\vert \,\vert a\vert  + \vert c\vert + \vert a\vert }~ \{b,\{c,a\}\} \nonumber \\
& \hspace{3cm} +(-1)^{\vert c\vert \,\vert b\vert + \vert a\vert + \vert b\vert }~ \{c,\{a,b\}\} \,=\,0\quad.
\end{flalign}

\item[(iv)] Derivation property: For all $a,b,c\in A$,
\begin{flalign}
\{a,b\,c\} \,=\, \{a,b\} \,c + (-1)^{\vert b\vert \,(\vert a\vert +1)}\, b\,\{a,c\}\quad.
\end{flalign}
\end{itemize}
In what follows,  we will always describe $P_0$-algebras
in this more explicit form.
\end{rem}

Let us now explain in more detail how the pairing of 
a free BV theory $(E,-Q,\langle\,\cdot\,, \,\cdot\,\rangle)$ 
induces a $P_0$-algebra structure on
the symmetric algebra $\Sym \,E[1]\in \CAlg(\Ch_\bbK)$.
Analogously to Remark~\ref{rem:P0algebraTrad}, the cochain
map $\langle \,\cdot\, , \,\cdot\,\rangle : E\otimes E\to\bbK[-1]$
defines a pairing on the shifted complex $E[1]$ via
\begin{flalign}\label{eqn:dualpairing}
\xymatrix@C=3em{
\ar[d]_-{\cong}E[1]\otimes E[1] \ar[r]^-{(\,\cdot\, , \,\cdot\, )} ~&~ \bbK[1] \\
(E\otimes E)[2] \ar[r]_-{\langle\,\cdot\, , \,\cdot\, \rangle[2]}~&~\bbK[-1][2]\ar[u]_-{\cong}
}
\end{flalign}
where the left vertical isomorphism is given by $\varphi \otimes \psi 
\mapsto (-1)^{\vert \varphi\vert+1}\,\varphi \otimes \psi$,
for all $\varphi,\psi\in E[1]$. (To understand this sign factor,
note that for $\varphi\in E[1]$ with $E[1]$-degree $\vert\varphi\vert$, the
$E$-degree is $\vert\varphi\vert +1$.) One easily checks that antisymmetry
of $\langle\,\cdot\,,\,\cdot\,\rangle$ implies that
the cochain map $(\,\cdot\,,\,\cdot\,) : E[1]\otimes E[1]\to \bbK[1]$
is symmetric, i.e.\ $(\,\cdot\,, \,\cdot\,) = (\,\cdot\,,\,\cdot\,)\circ \tau $ 
where $\tau$ denotes the symmetric braiding on $\Ch_\bbK$.
Using now the properties (ii) and (iv) from Remark \ref{rem:P0algebraTrad},
we observe that the shifted Poisson 
bracket $\{\, \cdot\, , \,\cdot\,\} : \Sym \,E[1] \otimes \Sym \,E[1] \to (\Sym \,E[1])[1]$ 
is completely determined by its value on the generators, for which we set
\begin{flalign}\label{eqn:P0frompairing}
\{\varphi,\psi \} \,:=\, ( \varphi,\psi) \,\oone \quad,
\end{flalign}
for all $\varphi,\psi\in E[1]$, where $\oone\in\Sym\, E[1]$ 
denotes the unit element. The Jacobi identity (iii) holds trivially because we are considering
a constant shifted Poisson structure, and
property (i) follows from the fact that $(\,\cdot\, ,\, \cdot\,)$ is 
a cochain map. Altogether, this construction defines a $P_0$-algebra
\begin{flalign}
\Obs^\cl\,:=\,\big(\Sym\, E[1],Q,\{\,\cdot\, ,\,\cdot\,\}\big)\quad,
\end{flalign}
which is interpreted as the classical observables of the free BV theory. 
\sk

Interactions and quantization are both
described by certain types of deformations of $\Obs^\cl$, which we shall now briefly review.
We start with the interactions.  Let $\lambda$ be a formal parameter,
interpreted as a coupling constant, and consider the formal power series
extension of $\Obs^\cl$, which we denote with the usual abuse of notation by the same symbol.
Given any $0$-cochain $I\in (\Sym\,E[1])^0$, interpreted as an interaction term
for the classical BV action,  we would like to define a deformed differential
on $\Obs^\cl$ by the formula
\begin{flalign}\label{eqn:Qint}
Q^{\mathrm{int}}\,:=\,Q+\{\lambda\,I,\,\cdot\,\}\quad.
\end{flalign}
Using the axioms for $P_0$-algebras from Remark \ref{rem:P0algebraTrad}, as well as non-degeneracy
of the pairing $\langle\,\cdot\, ,\,\cdot\,\rangle$, one easily checks that the 
nilpotency condition $(Q^{\mathrm{int}})^2 =0$
of the deformed differential is equivalent to the \emph{classical master equation}
\begin{flalign}\label{eqn:CME}
Q(\lambda\,I) + \tfrac{1}{2}\,\{\lambda\,I,\lambda\, I\}\,=\,0\quad.
\end{flalign}
The resulting $P_0$-algebra
\begin{flalign}
\Obs^{\cl,\mathrm{int}}\,:=\, \big(\Sym\,E[1],Q^{\mathrm{int}},\{\,\cdot\, ,\,\cdot\,\}\big)
\end{flalign}
is  interpreted as the classical observables of the interacting BV theory corresponding to  the
interaction term $I\in (\Sym\,E[1])^0$, which must satisfy the classical 
master equation \eqref{eqn:CME}.
\sk

For quantization of the free BV theory,  let $\hbar$ be another formal parameter, interpreted as 
Planck's constant, and consider the formal power series extension of $\Obs^\cl$,  denoted again 
by the same symbol.  The deformation of the differential on $\Obs^\cl$
that encodes quantization is given by 
\begin{flalign}\label{eqn:Qhbar}
Q^{\hbar}\,:=\, Q + \hbar\,\Delta_{\BV}\quad,
\end{flalign}
where the BV Laplacian is the cochain map $\Delta_{\BV} : \Sym\,E[1]\to(\Sym\,E[1])[1]$ defined
as follows: For symmetric powers $0$,  $1$ and $2$, we set
\begin{subequations}\label{eqn:BVLaplacian}
\begin{flalign}
\Delta_{\BV}(\oone) \,:=\,0~~,\quad
\Delta_{\BV}(\varphi)\,:=\, 0~~,\quad
\Delta_{\BV}(\varphi\,\psi)\,:=\,\{\varphi,\psi\} = (\varphi,\psi)\,\oone\quad,
\end{flalign}
for all $\varphi,\psi\in E[1]$. This is extended to all of $\Sym\,E[1]$ by demanding
\begin{flalign}\label{eqn:BVLaplacianDiffOp}
\Delta_{\BV}(a\, b)\,=\,\Delta_{\BV} (a)\,b + (-1)^{\vert a\vert}\, 
a\,\Delta_{\BV} ( b) +  \{a,b\}\quad,
\end{flalign}
\end{subequations}
for all $a,b\in\Sym\,E[1]$.  By iterating the latter property and using
the $P_0$-algebra axioms from Remark \ref{rem:P0algebraTrad},
one finds the explicit expression
\begin{flalign}\label{eqn:BVexplicit}
\Delta_{\BV}\big(\varphi_1\cdots \varphi_n \big) \,=\, \sum_{i<j}\, (-1)^{\sum_{k=1}^{i-1}\,\vert \varphi_k\vert + \vert\varphi_j\vert\,  \sum_{k=i+1}^{j-1}\, \vert \varphi_k\vert}~( \varphi_i,\varphi_j ) ~\varphi_1\cdots \widehat{\varphi}_i\cdots\widehat{\varphi}_j\cdots\varphi_n
\end{flalign}
for the BV Laplacian, for all $\varphi_1,\dots,\varphi_n\in E[1]$ with $n\geq 2$, 
where the hat means to omit the corresponding factor.\footnote{The sign
factors in \eqref{eqn:BVexplicit} can be understood as follows:
Since the pairing $(\,\cdot\,,\,\cdot\,)$ is of degree $1$, we obtain 
the first sign factor  $(-1)^{\sum_{k=1}^{i-1}\,\vert \varphi_k\vert}$ 
when moving it across $\varphi_1\cdots\varphi_{i-1}$.
Permuting $\varphi_j$ to the $i+1$-st position yields the second sign factor
$(-1)^{\vert\varphi_j\vert\,  \sum_{k=i+1}^{j-1}\, \vert \varphi_k\vert}$.
The evaluated quantity $(\varphi_i,\varphi_j)$ is a scalar, hence it commutes 
with all $\varphi_k$ without introducing any further signs.}
By construction, the BV Laplacian satisfies the two properties
\begin{flalign}
Q\,\Delta_\BV + \Delta_\BV\,Q \,=\,0~~,\quad
(\Delta_\BV)^2\,=\,0\quad,
\end{flalign}
which imply that \eqref{eqn:Qhbar} defines a differential, i.e.\ $(Q^\hbar)^2=0$.
We denote by
\begin{flalign}\label{eqn:Obshbar}
\Obs^\hbar \,:=\,\big(\Sym\,E[1],Q^\hbar\big)
\end{flalign}
the resulting deformed cochain complex, which is
interpreted as the quantum observables for the free BV theory.
It is important to emphasize that,  as a consequence of \eqref{eqn:BVLaplacianDiffOp},
the deformed differential $Q^\hbar$ does {\em not} respect the multiplication on $\Sym\,E[1]$, 
i.e.\ $\Obs^\hbar$ is {\em not} a dg-algebra.  The algebraic
structure of $\Obs^\hbar$ is that of an \emph{$E_0$-algebra},
i.e.\ a cochain complex with a distinguished $0$-cocycle,
which in the present case is the unit element $\oone\in\Sym\,E[1]$.
\sk

The two types of deformations corresponding to interactions and quantization
can be combined,  which leads to interacting quantum BV theories.
Starting from the quantum observables for the free theory \eqref{eqn:Obshbar},
let us choose again a $0$-cochain $I\in(\Sym\,E[1])^0$ playing the role of an interaction term.
We would like to define a deformed differential on $\Obs^\hbar$ by the formula
\begin{flalign}\label{eqn:Qhbarint}
Q^{\hbar, \mathrm{int}}\,:=\,Q^\hbar + \{\lambda\,I,\,\cdot\,\} \,=\, Q + \hbar\,\Delta_\BV + \{\lambda\,I,\,\cdot\,\}\quad.
\end{flalign}
The nilpotency condition $(Q^{\hbar, \mathrm{int}})^2=0$ for this differential
is equivalent to the \emph{quantum master equation}
\begin{flalign}\label{eqn:QME}
Q(\lambda\,I) +\hbar \,\Delta_\BV(\lambda\,I) + \tfrac{1}{2}\,\{\lambda\,I,\lambda\, I\}\,=\,0\quad.
\end{flalign}
To verify the last statement, it is helpful to note the identity
\begin{flalign}
-\Delta_{\BV}\big(\{a,b\}\big) \,=\, \{\Delta_\BV(a),b\} + (-1)^{\vert a\vert}\,\{a,\Delta_\BV(b)\}\quad,
\end{flalign}
for all $a,b\in \Sym\, E[1]$, which may be derived by applying $\Delta_\BV$ on both sides of 
\eqref{eqn:BVLaplacianDiffOp}. The resulting $E_0$-algebra
\begin{flalign}
\Obs^{\hbar,\mathrm{int}}\,:=\, \big(\Sym\,E[1],Q^{\hbar,\mathrm{int}}\big)
\end{flalign}
is  interpreted as the quantum observables of the interacting BV theory corresponding to 
the interaction term $I\in (\Sym\,E[1])^0$, which must satisfy the quantum
master equation \eqref{eqn:QME}.

\subsection{\label{subsec:Linfty}Cyclic $L_\infty$-algebras}
Let us briefly recall the well-known and powerful
construction of interaction terms $I\in(\Sym\,E[1])^0$ satisfying the classical 
(and also the quantum) master equation from cyclic $L_\infty$-algebra structures.
See e.g.~\cite{JRSW19} for further details.
\begin{defi}
An {\em $L_\infty$-algebra} is a $\bbZ$-graded vector space $L$
together with a collection $\{ \ell_n : L^{\otimes n}\to L\}_{n\in \bbZ_{\geq 1}}$
of graded antisymmetric linear maps of degree $\vert \ell_n\vert = 2-n$
that satisfy the homotopy Jacobi identities
\begin{flalign}\label{eqn:homotopyJacobi}
\sum_{k=0}^{n-1}\, (-1)^{k}~
\ell_{k+1}\circ \big( \ell_{n-k}\otimes \id_{L^{\otimes k}}\big)\,\circ\, \sum_{\sigma\in\mathrm{Sh}(n-k;k)}\,
\mathrm{sgn}(\sigma) \,\tau^{\sigma}\,=\,0\quad,
\end{flalign}
for all $n\geq 1$, where $\mathrm{Sh}(n-k;k)\subset S_n$ denotes the 
set of $(n-k;k)$-shuffled permutations on $n$ letters and $\tau^\sigma : L^{\otimes n}\to L^{\otimes n}$
denotes the action of the permutation $\sigma$ via the symmetric 
braiding on the category of graded vector spaces.
\end{defi}
\begin{rem}\label{rem:Linfty}
As a consequence of the identity \eqref{eqn:homotopyJacobi} for $n=1$, the linear map
$\ell_1 : L\to L$ of degree $1$ is nilpotent: $(\ell_1)^2=0$. Hence
every $L_\infty$-algebra has an underlying cochain complex $(L,\dd_L := \ell_1)$. The 
identity \eqref{eqn:homotopyJacobi} for $n=2$ states that $\ell_2:L\otimes L\to L$ is 
a cochain map, while for $n=3$ it states that $\ell_2$ obeys the Jacobi identity up 
to the cochain homotopy $\ell_3:L\otimes L\otimes L\to L$. In particular, an 
$L_\infty$-algebra with $\ell_n=0$ for all $n\geq 3$
is simply a dg-Lie algebra with Lie bracket~$[\,\cdot\,,\,\cdot\,]:=\ell_2$.
\end{rem}
\begin{defi}
A {\em cyclic $L_\infty$-algebra} is an $L_\infty$-algebra $(L,\{\ell_n\})$ 
together with a non-degenerate symmetric cochain
map $\cyc{\,\cdot\,}{\,\cdot\,} : L\otimes L \to\bbK[-3]$ that satisfies
the cyclicity condition
\begin{flalign}
\cyc{v_0}{\ell_n(v_1,\dots,v_n)}\,=\,(-1)^{n\,(\vert v_0\vert + 1)} ~\cyc{v_n}{\ell_n(v_0,\dots,v_{n-1})}\quad,
\end{flalign}
for all $n\geq 1$ and all homogeneous elements $v_0,v_1,\dots,v_n\in L$.
\end{defi}

Given any free BV theory $(E,-Q,\langle\,\cdot\,,\,\cdot\,\rangle)$
as in Definition \ref{def:linearBV}, the shifted cochain complex
$E[-1]$ defines an Abelian cyclic $L_\infty$-algebra 
with $\ell_1 = \dd_{E[-1]}= Q$ and $\ell_n=0$ for all $n\geq 2$.
The cyclic structure
\begin{flalign}\label{eqn:cyclicCl}
\xymatrix@C=3.5em{
\cyc{\,\cdot\,}{\,\cdot\,} \,:\,E[-1]\otimes E[-1]\,\cong\, (E\otimes E)[-2]
\ar[r]^-{\langle\,\cdot\, , \,\cdot\,\rangle[-2]} ~&~ \bbK[-1][-2]\,\cong\,\bbK[-3]
}
\end{flalign}
is defined analogously to \eqref{eqn:dualpairing}. The problem of 
finding an interaction term $I\in (\Sym\,E[1])^0$ that satisfies
the classical master equation \eqref{eqn:CME}
is then equivalent to endowing the cochain complex 
$(E[-1],\ell_1 = Q)$ with higher brackets $\{\ell_n\}_{n\geq 2}$ that 
result in a cyclic $L_\infty$-algebra with respect to \eqref{eqn:cyclicCl}.
The relationship between the brackets $\{\ell_n\}_{n\geq 2}$ 
and the interaction term $I\in (\Sym\,E[1])^0$ is given by the {\em homotopy Maurer-Cartan action}, 
see e.g.\ \cite[Section 4.3]{JRSW19} for a detailed presentation. 
\sk

This is most easily written
down by using `contracted coordinate functions'.
Let us choose any basis $\{\varepsilon_\alpha \in E[-1]\}$ of the $L_\infty$-algebra
and denote by $\{\varrho^\alpha \in E[-1]^\ast \cong E[2]\}$ its dual with respect
to the cyclic structure, i.e.\ $\cyc{\varrho^\alpha}{\varepsilon_\beta} 
=\delta^\alpha_\beta$ for all $\alpha,\beta$. Then the contracted
coordinate functions are defined as the element
\begin{flalign}\label{eqn:CCF}
\mathsf{a} \,:=\, \sum_\alpha\, \varrho^\alpha\otimes \varepsilon_\alpha
\ \in \ \big((\Sym\,E[1])\otimes E[-1]\big)^1
\end{flalign} 
of degree $1$ in the tensor product of the dg-algebra of
polynomial observables and the $L_\infty$-algebra $E[-1]$.
The cyclic $L_\infty$-algebra structure on $E[-1]$ extends in the obvious way
to the tensor product $(\Sym\,E[1])\otimes E[-1]$, providing
extended brackets 
\begin{flalign}
\ell_n^{\mathrm{ext}} \,:\, \big((\Sym\,E[1])\otimes E[-1]\big)^{\otimes n}~\longrightarrow~
(\Sym\,E[1])\otimes E[-1]\quad,
\end{flalign} 
for all $n \geq 2$, and an extended pairing
\begin{flalign}
\cyc{\,\cdot\,}{\,\cdot\,}_{\mathrm{ext}} \,:\, \big((\Sym\,E[1])\otimes E[-1]\big)\otimes\big(
(\Sym\,E[1])\otimes E[-1]\big) ~\longrightarrow~ (\Sym\,E[1])[-3]\quad.
\end{flalign}
The explicit formulas can be found in \cite[Section 2.3]{JRSW19}. 
\sk

We can now finally write down the desired relationship
\begin{flalign}\label{eqn:MaurerCartanAction}
\lambda \, I \,=\, \sum_{n\geq 2}\,\frac{\lambda^{n-1}}{(n+1)!}\, \cyc{\mathsf{a}}{\ell_n^{\mathrm{ext}}(\mathsf{a},\dots,\mathsf{a})}_{\mathrm{ext}}\ \in\ (\Sym\,E[1])^0
\end{flalign}
between the higher brackets $\{\ell_n\}_{n\geq 2}$ and the interaction term $I\in (\Sym\,E[1])^0$,
where we recall that $\lambda$ is a formal parameter which is interpreted as a coupling 
constant.\footnote{It is easy to check that
given any $L_\infty$-algebra structure $\{\ell_n\}_{n\geq 1}$ and
any (formal or non-formal) parameter $\lambda$, the rescaled brackets 
$\{\lambda^{n-1}\,\ell_n\}_{n\geq 1}$ also satisfy the homotopy Jacobi identities 
\eqref{eqn:homotopyJacobi}. This explains the 
powers of $\lambda$ on the right-hand side of \eqref{eqn:MaurerCartanAction}.}
It can be shown (cf.~\cite[Section~4.3]{JRSW19}) that the interaction term 
\eqref{eqn:MaurerCartanAction} satisfies the classical master equation \eqref{eqn:CME} and
that it is annihilated by the BV Laplacian, i.e.\ $\Delta_{\BV} (\lambda\,I) = 0$.
As a consequence, it also satisfies the quantum master equation \eqref{eqn:QME}.

\subsection{\label{subsec:HPL}Homological perturbation theory}
The correlation functions of non-interacting and also interacting 
quantum BV theories can be computed by employing techniques
from homological perturbation theory, see e.g.\ \cite{CG16,Gwi12}.  
We will now briefly review the relevant constructions.
In the following definition we regard the cohomology $H^\bullet(V)$
of a cochain complex $V\in\Ch_\bbK$ as a cochain complex with trivial differential.
\begin{defi}\label{def:defret}
A {\em strong deformation retract} of a cochain complex $V\in \Ch_\bbK$ onto its cohomology
$H^\bullet(V)$ is given by the following data:
\begin{itemize}
\item[(i)] A cochain map $\iota :H^\bullet(V)\to V$;
\item[(ii)] A cochain map $\pi : V\to H^\bullet(V)$;
\item[(iii)] A $(-1)$-cochain $\gamma \in\hom(V,V)^{-1}$.
\end{itemize}
These data are required to satisfy the following conditions:
\begin{itemize}
\item[a)] $\pi \, \iota = \id_{H^\bullet(V)}$;
\item[b)] $\iota\, \pi - \id_{V} = \partial (\gamma) = \dd\,\gamma + \gamma\,\dd$;
\item[c)] $\gamma^2=0$,  $\gamma\, \iota=0$ and $\pi\, \gamma=0$.
\end{itemize}
A strong deformation retract may be visualized by
\begin{equation}\label{eqn:defretpicture}
\begin{tikzcd}
(H^\bullet(V),0) \ar[r,shift right=1ex,swap,"\iota"] & \ar[l,shift right=1ex,swap,"\pi"] (V,\dd)\ar[loop,out=25,in=-25,distance=28,"\gamma"]
\end{tikzcd}\quad,
\end{equation}
where we also explicitly display the differentials.
\end{defi}

The homological perturbation lemma (see e.g.\ \cite{Cra04}) states that small perturbations
$\dd + \delta$ of the differential $\dd$ on $V$ lead to perturbations
of strong deformation retracts.  By a \emph{small perturbation} one means that,
in addition to $(\dd + \delta)^2=0$, the map $\id_V - \delta\, \gamma$
is invertible. In particular, the formal deformations of Section~\ref{subsec:BV}
are always small perturbations in this sense.  The precise statement
of the homological perturbation lemma is as follows.
\begin{theo}\label{theo:HPL}
Consider any strong deformation retract as in \eqref{eqn:defretpicture}
and let $\delta \in\hom(V,V)^1$ be a small perturbation.  Then there exists
a strong deformation retract
\begin{subequations}\label{eqn:HPL}
\begin{equation}
\begin{tikzcd}
(H^\bullet(V),\widetilde{\delta} \, ) \ar[r,shift right=1ex,swap,"\widetilde{\iota}"] & \ar[l,shift right=1ex,swap,"\widetilde{\pi}"] (V,\dd+\delta) \quad \ar[loop,out=25,in=-25,distance=28,"\widetilde{\gamma}"]
\end{tikzcd}
\end{equation}
with
\begin{flalign}
\widetilde{\delta}\,&=\, \pi\, (\id_V - \delta\, \gamma)^{-1}\, \delta\, \iota\quad,\\[4pt]
\widetilde{\iota}\,&=\, \iota + \gamma\, (\id_V - \delta\, \gamma)^{-1}\,\delta\, \iota\quad,\\[4pt]
\widetilde{\pi}\,&=\, \pi + \pi\, (\id_V - \delta\, \gamma)^{-1}\,\delta\, \gamma\quad,\\[4pt]
\widetilde{\gamma} \,&=\, \gamma + \gamma\, (\id_V - \delta\, \gamma)^{-1}\, \delta\, \gamma\quad.
\end{flalign}
\end{subequations}
\end{theo}

Let us consider now a free BV theory $(E,-Q,\langle\,\cdot\,, \,\cdot\,\rangle)$  
in the sense of Definition \ref{def:linearBV} and choose
a strong deformation retract for its dual complex $E^\ast \cong E[1]$, i.e.\
\begin{equation}\label{eqn:defretL}
\begin{tikzcd}
(H^\bullet(E[1]),0) \ar[r,shift right=1ex,swap,"\iota"] & \ar[l,shift right=1ex,swap,"\pi"] (E[1],Q)\ar[loop,out=25,in=-25,distance=28,"\gamma"]
\end{tikzcd}\quad.
\end{equation}
In Section \ref{sec:fuzzysphere} we shall illustrate through concrete examples of matrix models
that such strong deformation retracts are related to Green operators.
It was shown in \cite[Proposition 2.5.5]{Gwi12} that there exists an associated
strong deformation retract
\begin{equation}\label{eqn:defretSym}
\begin{tikzcd}
\big(\Sym\,H^\bullet(E[1]),0\big) \ar[r,shift right=1ex,swap,"\Sym\,\iota"] & \ar[l,shift right=1ex,swap,"\Sym\,\pi"] \big(\Sym\,E[1],Q\big) \quad \ \ar[loop,out=25,in=-25,distance=28,"\Sym\,\gamma"]
\end{tikzcd}
\end{equation}
at the level of symmetric algebras. The cochain maps
$\Sym\,\iota$ and $\Sym\,\pi$ are given by extending $\iota$
and $\pi$ in the usual way as commutative dg-algebra morphisms, i.e.\
\begin{flalign}
\Sym\,\iota\big([\psi_1]\cdots[\psi_n]\big)\,:=\,\iota([\psi_1])\cdots\iota([\psi_n])~~,\quad
\Sym\,\pi\big(\varphi_1\cdots\varphi_n\big)\,:=\, \pi(\varphi_1)\cdots \pi(\varphi_n)\quad,
\end{flalign}
for all $[\psi_1],\dots,[\psi_n]\in H^\bullet(E[1])$ and $\varphi_1,\dots,\varphi_n\in E[1]$.
\sk

The cochain homotopy $\Sym\,\gamma$ is slightly more complicated to define.
We first note that $\iota\,\pi : E[1]\to E[1]$ defines a projector,  i.e.\ $(\iota\,\pi)^2=\iota\,\pi$.
Hence one obtains a decomposition
\begin{subequations}
\begin{flalign}
E[1]\,\cong\, E[1]^{\perp}\oplus H^\bullet(E[1])
\end{flalign}
and consequently
\begin{flalign}
\Sym\,E[1] \,\cong\, \Sym\,E[1]^\perp \otimes \Sym\,H^\bullet(E[1])\,\cong\, \bigoplus_{n\geq 0}\,  \Sym^n\,E[1]^\perp \otimes \Sym\,H^\bullet(E[1]) \quad,
\end{flalign}
\end{subequations}
where $\Sym^n$ denotes the $n$-th symmetric power.  The cochain homotopy $\Sym\,\gamma$
is then defined by setting
\begin{flalign}\label{eqn:symhomotopy}
\Sym\,\gamma\big(\varphi^\perp_1\cdots \varphi_n^\perp\otimes a \big)
\,=\,\frac{1}{n}\, \sum_{i=1}^n\,(-1)^{\sum_{j=1}^{i-1}\, \vert \varphi_j^\perp\vert }~
\varphi_1^\perp\cdots \varphi_{i-1}^\perp\,  \gamma(\varphi_i^\perp)\,\varphi_{i+1}^\perp \cdots\varphi_n^\perp\otimes a\quad,
\end{flalign}
for all homogeneous elements 
$\varphi^\perp_1\cdots \varphi_n^\perp \otimes a \in  
\Sym^n\,E[1]^\perp \otimes \Sym\,H^\bullet(E[1])$ in this decomposition.
For $n=0$,  this expression should be read as $\Sym\,\gamma(a) =0$,
for all $a\in \Sym\,H^\bullet(E[1])$.
\sk

The correlation functions of non-interacting and also interacting
quantum BV theories can then be determined by applying the homological perturbation
lemma from Theorem~\ref{theo:HPL} to the strong deformation retract \eqref{eqn:defretSym}
and the deformed differentials from Section~\ref{subsec:BV}.
Let us explain this in some more detail. 
As we have explained in Section~\ref{subsec:BV}, 
quantization and including interactions are described by
deforming the differential $Q$ of the right complex in the strong deformation retract \eqref{eqn:defretSym}.
Let us write generically $Q+\delta$ for the deformed differential, where 
$\delta$ stands either for an interaction term \eqref{eqn:Qint}, the BV Laplacian
\eqref{eqn:Qhbar} or the sum of both \eqref{eqn:Qhbarint}.
Applying Theorem~\ref{theo:HPL} we obtain a deformed strong deformation retract,
which we denote by
\begin{equation}\label{eqn:tildeSymretract}
\begin{tikzcd}
\big(\Sym\,H^\bullet(E[1]),\widetilde{\delta}\ \big) \ar[r,shift right=1ex,swap,"\widetilde{\Sym\,\iota}"] & \ar[l,shift right=1ex,swap,"\widetilde{\Sym\,\pi}"] \big(\Sym\,E[1],Q+\delta \big) \qquad \quad \ \ar[loop,out=25,in=-25,distance=40,"\widetilde{\Sym\,\gamma}"]
\end{tikzcd}\quad.
\end{equation}
\sk

The `smeared' $n$-point correlation functions are then given by applying
the map $\widetilde{\Sym\,\pi}$ on a product of the `test functions'
$\varphi_1,\dots,\varphi_n\in E[1]$, i.e.\
\begin{flalign}\label{eqn:npointfunction}
\widetilde{\Sym\,\pi}\big(\varphi_1\cdots\varphi_n\big)\ \in\ \Sym\,H^\bullet(E[1])\quad.
\end{flalign}
This can be computed perturbatively (as formal power series in $\lambda$ or $\hbar$, or both) 
by using the explicit formulas from Theorem~\ref{theo:HPL}.
Note that, in general, the correlation functions are not
simply numbers, rather they are elements of the symmetric algebra $\Sym\,H^\bullet(E[1])$.
The latter should be interpreted as the algebra of polynomial functions on
the space of vacua of the theory, which is the cohomology $H^\bullet(E)$ of
the derived solution complex $E$, cf.\ Remark \ref{rem:dersolution}.
Hence the $n$-point correlation functions in \eqref{eqn:npointfunction} 
are functions on the space of vacua which, when evaluated in a particular vacuum, 
give the usual numerical correlations of the perturbative field theory
around this vacuum. We will illustrate this through concrete examples
in Section \ref{sec:fuzzysphere}.


\section{\label{sec:fuzzysphere}Field theories on the fuzzy sphere}
We illustrate the formalism of Section \ref{sec:prelim}
by studying scalar field theories and also Chern-Simons theory
on the fuzzy $2$-sphere.
The examples presented in this section are over the field 
$\bbK=\bbC$ of complex numbers.

\subsection{\label{subsec:ScalarSphere}Scalar field theories}
We consider first the simplest case of scalar field theories, 
where we show how our formalism reproduces the known $1$-loop $2$-point 
function for $\Phi^4$-theory on the fuzzy sphere, see e.g.~\cite{CMS01}.
However, in contrast to the traditional approach of~\cite{CMS01}, our 
correlation functions are generally disconnected and $1$-particle reducible,
and involve unamputated external legs.

\paragraph{The fuzzy 2-sphere.}
To fix our notation and conventions, let us recall the definition of the 
fuzzy $2$-sphere following~\cite{CMS01}.
Let $N\in\bbZ_{>0}$ be a positive integer and let $V$
denote the irreducible spin $N/2$ representation of the $\su(2)$ Lie algebra.
The algebra of functions on the fuzzy sphere $\bbS^2_N$ is defined by
\begin{flalign}\label{eqn:sphereAlgebra}
A \,:=\, \underline{\mathrm{end}}(V)\,=\, V \otimes V^\ast \quad,
\end{flalign}
where $V^\ast$ denotes the dual representation.
Since the underlying vector space of $V$ 
is $(N+1)$-dimensional, it follows that $A\cong \mathrm{Mat}_{N+1}(\bbC)$
is isomorphic to the algebra of $(N+1)\times(N+1)$-matrices with complex entries.
The action of $\su(2)$ on $V$ is encoded by a Lie algebra homomorphism
\begin{flalign}
\rho\, : \,\su(2)~\longrightarrow~ A\quad,
\end{flalign}
where the Lie bracket on $A$ is the matrix commutator.
Let us choose a basis $\{e_i\in\su(2)\}_{i=1,2,3}$ of $\su(2)$
with the usual Lie bracket relations $[e_i,e_j]=\ii \epsilon_{ijk}\,e_k$, 
where $\epsilon_{ijk}$ is the Levi-Civita symbol and summation over 
repeated indices is always understood. We introduce the constant
\begin{flalign}
\lambda_N \,:=\,\frac{1}{\sqrt{\frac{N}{2}\,\big(\frac{N}{2}+1\big)}}\ \in\ \bbR\quad.
\end{flalign} 
Then the elements
\begin{subequations}
\begin{flalign}
X_i \,:=\, \lambda_N\,\rho(e_i)\ \in\ A
\end{flalign}
generate the algebra $A$ and satisfy the fuzzy unit sphere relations
\begin{flalign}
[X_i,X_j]\,=\,\ii\lambda_N\, \epsilon_{ijk}\,X_k~~,\quad
\delta_{ij}\,X_{i}\,X_j   \,=\,\oone ~~,\quad
X_i^\ast \,=\, X_{i}\quad,
\end{flalign}
\end{subequations}
where ${}^\ast$ denotes Hermitian conjugation and
$\delta_{ij}$ is the Kronecker delta-symbol.
\sk

Integration on the fuzzy sphere is given by the normalized trace map
\begin{flalign}
A~\longrightarrow~\bbC~,~~a~\longmapsto~ \frac{4\pi}{N+1}\,\Tr(a)\quad,
\end{flalign}
and the scalar Laplacian reads as
\begin{flalign}\label{eqn:scalarsphereLaplacian}
\Delta\,:\, A~\longrightarrow~A~,~~a~\longmapsto~\Delta(a)\,:=\, 
\frac{1}{\lambda_N^2}\,\delta_{ij}\,[X_i,[X_j,a]]\quad.
\end{flalign}
A vector space basis of $A$ is given by the fuzzy spherical harmonics
$Y^J_j\in A$, for $J=0,1,\dots, N$ and $-J\leq j\leq J$.
The fuzzy spherical harmonics are eigenfunctions of the scalar Laplacian satisfying the identities
\begin{flalign}
\Delta(Y^J_j) \,=\,J\,(J+1)\, Y^J_j~~,\quad
{Y^J_j}^\ast \,=\, (-1)^J\, Y^J_{-j}~~,\quad
\frac{4\pi}{N+1}\, \Tr\big({Y^J_j}^\ast\, Y^{J^\prime}_{j^\prime}\big)
\,=\,\delta_{JJ^\prime}\,\delta_{j j^\prime}\quad.
\end{flalign}
There is also an explicit `fusion formula' for the products $Y^I_i\, Y^J_j$
of fuzzy spherical harmonics in terms of Wigner's $3j$ and $6j$ symbols, see
e.g.\ \cite{CMS01}, which we however do not need in the present paper.

\paragraph{Free BV theory.}
We are now ready to describe a non-interacting scalar field theory
on the fuzzy sphere as a free BV theory in the sense of Definition \ref{def:linearBV}.
\begin{defi}\label{def:spherescalar}
The free BV theory associated to a scalar 
field with mass parameter $m^2\geq 0$ on the fuzzy sphere is 
given by the cochain complex
\begin{flalign}
E\,=\,\big(\xymatrix{
\stackrel{(0)}{A} \ar[r]^-{-Q} & \stackrel{(1)}{A}
}\big)\qquad\text{with}\qquad Q\,:=\,\Delta +m^2
\end{flalign}
concentrated in degrees $0$ and $1$, together with the pairing
\begin{flalign}
\langle\,\cdot\,,\,\cdot\,\rangle\,:\,
E\otimes E~\longrightarrow~\bbC[-1]~,~~
\varphi\otimes\psi~\longmapsto~\langle\varphi,\psi\rangle\,:=\, (-1)^{\vert \varphi\vert}\,\frac{4\pi}{N+1}\,\Tr\big( \varphi\,\psi\big)\quad.
\end{flalign}
\end{defi}

Following the general approach outlined in Section \ref{subsec:BV},
we can construct from this input a $P_0$-algebra
\begin{flalign}
\Obs^\cl \,:=\,\big(\Sym\,E[1] , Q ,\{\,\cdot\,,\,\cdot\,\}\big)
\end{flalign}
of classical observables of the non-interacting theory,
where we note that the complex
\begin{flalign}\label{eqn:scalarsphereLinObs}
E[1] \,=\, \big(\xymatrix{
\stackrel{(-1)}{A} \ar[r]^-{Q} & \stackrel{(0)}{A}
}\big)
\end{flalign}
is concentrated in degrees $-1$ and $0$.

\paragraph{Interactions.}
In the present case of a scalar field,  
the dg-algebra $\Sym\,E[1]$ is concentrated in non-positive
degrees because the generators $E[1]$ are of degrees $-1$ and $0$.
Hence, simply for degree reasons, every $0$-cochain $I\in (\Sym\, E[1])^0$ satisfies 
\begin{flalign}
Q(I)\,=\, \{I,I\} \,=\, \Delta_{\BV}(I) \,=\,0\quad,
\end{flalign}
and consequently it automatically satisfies both the classical
and quantum master equations  \eqref{eqn:CME} and \eqref{eqn:QME}, respectively.
This means that every $0$-cochain $I\in (\Sym\,E[1])^0$
provides a well-defined interaction term for a scalar field
in both the classical and quantum cases. 
\sk

Let us nevertheless use the cyclic $L_\infty$-algebra formalism 
from Section \ref{subsec:Linfty} to introduce concrete examples of interaction terms,
focusing on the typical $m+1$-point interactions.
The Abelian cyclic $L_\infty$-algebra corresponding
to the scalar field from Definition \ref{def:spherescalar} is given by
the cochain complex
\begin{flalign}
E[-1]\,=\, \big(\xymatrix{
\stackrel{(1)}{A} \ar[r]^-{Q} & \stackrel{(2)}{A}
}\big)
\end{flalign}
and the cyclic structure
\begin{flalign}
\cyc{\,\cdot\,}{\,\cdot\,} \,:\, E[-1]\otimes E[-1]~\longrightarrow~\bbC[-3]~,~~
\varphi\otimes\psi~\longmapsto~\cyc{\varphi}{\psi}\,=\,\frac{4\pi}{N+1}\, \Tr\big(\varphi\,\psi\big)\quad.
\end{flalign}
Choosing any $m\geq 2$, we can endow this with the compatible $m$-bracket
\begin{flalign}\label{eqn:ellmscalar}
\ell_m\,:\,E[-1]^{\otimes m}~\longrightarrow~E[-1]~,~~
\varphi_1\otimes\cdots\otimes\varphi_m~\longmapsto~
\frac{1}{m!}\,\sum_{\sigma\in S_m}\,\varphi_{\sigma(1)}\cdots\varphi_{\sigma(m)}\quad,
\end{flalign}
which for degree reasons is only non-vanishing if each $\varphi_i\in E[-1]$ 
is of degree $1$ in $E[-1]$. The symmetrization of the matrix multiplications
in the definition of $\ell_m$ then implies that $\ell_m$
is, as required, graded antisymmetric.
\sk

The contracted coordinate functions \eqref{eqn:CCF} 
in the present case can be described in terms of the fuzzy spherical harmonics $Y^J_j$.
We shall write $Y^J_j\in E[-1]^1=A$ when we regard the fuzzy spherical harmonics
as elements of degree $1$ in $E[-1]$ and \smash{$\widetilde{Y}^J_j\in E[-1]^2=A$} when we regard them
as elements of degree $2$ in $E[-1]$. With these notational conventions,
the contracted coordinate functions read as
\begin{flalign}\label{eqn:fuzzyS2sfa}
\mathsf{a}\,=\, \sum_{J,j}\, {Y^J_j}^\ast\otimes Y^J_j 
+ \sum_{J,j}\, {\widetilde{Y}^J_j}{}^\ast\otimes \widetilde{Y}^J_j
\ \in\ \big((\Sym \,E[1])\otimes E[-1]\big)^1\quad,
\end{flalign}
where ${Y^J_j}^\ast\in E[1]^0=A$ denotes elements
of degree $0$ in $E[1]$ and \smash{${\widetilde{Y}^J_j}{}^\ast\in E[1]^{-1}=A$} denotes elements
of degree $-1$ in $E[1]$. 
\sk

The interaction term \eqref{eqn:MaurerCartanAction}
corresponding to an $m+1$-point interaction 
then reads concretely as
\begin{subequations}\label{eqn:mplus1point}
\begin{flalign}
\lambda \,I\,&=\,\frac{\lambda^{m-1}}{(m+1)!} \,
\cyc{\mathsf{a}}{\ell^{\mathrm{ext}}_m(\mathsf{a},\dots,\mathsf{a})}_{\mathrm{ext}}\\[4pt]
\nn \,&=\, \frac{\lambda^{m-1}}{(m+1)!} \,\sum_{J_0,j_0,\dots,J_m,j_m}\,{Y^{J_0}_{j_0}}^\ast\, {Y^{J_1}_{j_1}}^\ast\,
\cdots\,{Y^{J_m}_{j_m}}^\ast~\cyc{Y_{j_0}^{J_0}}{\ell_m(Y^{J_1}_{j_1},\dots,Y^{J_m}_{j_m})} 
\ \in\ (\Sym\,E[1])^0\quad,
\end{flalign}
where we stress that the products of the ${Y^{J_i}_{j_i}}^\ast\in E[1]^0$ 
in the second line are {\em not} given by matrix multiplication 
but rather by the product in the symmetric algebra
$\Sym\,E[1]$. The constants
\begin{flalign}
I_{j_0j_1\cdots j_m}^{J_0J_1\cdots J_m}\,
:=\,\cyc{Y_{j_0}^{J_0}}{\ell_m(Y^{J_1}_{j_1},\dots,Y^{J_m}_{j_m})} \ \in\ \bbC
\end{flalign}
\end{subequations}
can in principle be worked out explicitly in terms of the
Wigner $3j$ and $6j$ symbols as in \cite{CMS01}, but this level of detail 
is not needed in the present paper. Because
of the underlying cyclic $L_\infty$-algebra structure, the constants $I_{j_0j_1\cdots j_m}^{J_0J_1\cdots J_m}$
are symmetric under the exchange of any neighboring pairs of indices, i.e.\
\begin{flalign}\label{eqn:Isymmetry}
I_{j_0j_1\cdots j_i j_{i+1}\cdots j_m}^{J_0J_1\cdots J_i J_{i+1}\cdots J_m}
\,=\,I_{j_0j_1\cdots j_{i+1} j_{i}\cdots j_m}^{J_0J_1\cdots J_{i+1} J_{i}\cdots J_m} \quad.
\end{flalign}

\paragraph{Strong deformation retract.}
The cohomology of the complex \eqref{eqn:scalarsphereLinObs}
depends on whether one considers a massive or a massless scalar field.
Since the spectrum of the scalar Laplacian 
\eqref{eqn:scalarsphereLaplacian} is $\{J\,(J+1)\,:\, J=0,1,\dots, N\}$,
we obtain  
\begin{flalign}
H^\bullet(E[1]) \,\cong\,\begin{cases}
0 ~&~ \text{for } m^2 >0\quad,\\
\bbC[1]\oplus \bbC ~&~\text{for } m^2=0\quad.
\end{cases}
\end{flalign}
In the massive case $m^2 > 0$, the strong deformation retract
is given by
\begin{equation}\label{eqn:massiveDefRet}
\begin{tikzcd}
(0,0) \ar[r,shift right=1ex,swap,"\iota=0"] & \ar[l,shift right=1ex,swap,"\pi=0"] (E[1],Q)\ar[loop,out=25,in=-25,distance=28,"\gamma=-G"]
\end{tikzcd}\qquad\text{for }m^2>0\quad,
\end{equation}
where $G$ is the inverse of $Q = \Delta +m^2$, i.e.\ the Green operator,
and $\gamma=-G$ is defined to act as a degree $-1$ map on $E[1]$ 
(cf.\ Definition \ref{def:defret} (iii)). 
\sk

The massless case $m^2=0$ is slightly more complicated
because the scalar Laplacian $Q=\Delta$ has a non-trivial kernel, which is
given by complex multiples of the unit $\oone\in A$. 
The linear map \smash{$\eta \, \frac{1}{N+1}\, \Tr : A\to A\,,~a\mapsto 
\frac{1}{N+1}\,\Tr(a)\,\oone$}, 
obtained by composing the normalized trace and the unit map $\eta:\bbC\to A$,
defines a projector onto the kernel of $\Delta$, which can be used to decompose
\begin{flalign}\label{eqn:Adecomposition}
A\,\cong\, A^\perp\oplus \bbC\quad.
\end{flalign}
By the rank-nullity theorem of linear algebra,
the scalar Laplacian restricts to an isomorphism 
$\Delta^\perp : A^\perp \to A^\perp$
and we denote its inverse by $G^\perp :  A^\perp \to A^\perp$.
Extending $G^\perp$ by $0$ to all of $A$ we obtain the linear map
\begin{flalign}
G_{0}\,:=\, G^\perp \, \big(\id_A-\eta \, \tfrac{1}{N+1}\,\Tr \big)\,:\, A~\longrightarrow~A\quad,
\end{flalign}
from which the strong deformation retract in the massless case is given by
\begin{equation}\label{eqn:masslessDefRet}
\begin{tikzcd}
(\bbC[1]\oplus\bbC,0) \ar[rr,shift right=1ex,swap,"\iota=\eta"] && \ar[ll,shift right=1ex,swap,"\pi=\frac{1}{N+1}\,\Tr"] (E[1],Q)\ar[loop,out=25,in=-25,distance=28,"\gamma=-G_{0}"]
\end{tikzcd}\qquad\text{for }m^2=0\quad.
\end{equation}
The strong deformation retracts for the massive and massless cases 
\eqref{eqn:massiveDefRet} and \eqref{eqn:masslessDefRet}, respectively, 
extend to the symmetric algebras via the construction outlined below \eqref{eqn:defretSym}.

\paragraph{Correlation functions for $\boldsymbol{m^2>0}$.} We shall now explain in some more detail
how correlation functions may be computed and provide some explicit examples.
We focus here on the massive case $m^2>0$ and comment briefly on the massless case
later on. 
\sk

Recall that the strong deformation retract \eqref{eqn:massiveDefRet}
extends to the symmetric algebras. Given any small perturbation $\delta$ 
of the differential $Q$ on $\Sym\,E[1]$, we obtain the deformed strong deformation retract
\begin{equation}
\begin{tikzcd}
\big(\Sym\,0\cong \bbC ,0 \big) \ar[r,shift right=1ex,swap,"\widetilde{\Sym\,\iota}"] & \ar[l,shift right=1ex,swap,"\widetilde{\Sym\,\pi}"] \big(\Sym\,E[1],Q+\delta \big) \qquad \quad \ \ar[loop,out=25,in=-25,distance=40,"\widetilde{\Sym\,\gamma}"]
\end{tikzcd}\quad,
\end{equation}
where the tilded quantities are computed through the homological perturbation lemma, cf.\
Theorem \ref{theo:HPL}. To compute the correlation functions \eqref{eqn:npointfunction}, 
we have to consider the cochain map 
\begin{flalign}
\nn \widetilde{\Pi} \,:=&\ \widetilde{\Sym\,\pi}
\,=\, \Sym\,\pi + \Sym\,\pi \,\big(\id - \delta~\Sym\,\gamma\big)^{-1}\,
\delta\,(\Sym\,\gamma) \\[4pt]
\,=&\ \big( \Sym\,\pi\big)\,\circ ~\sum_{k=0}^{\infty}\, \big(\delta~\Sym\,\gamma\big)^k
\,=\, \Pi \,\circ~ \sum_{k=0}^{\infty}\, \big(\delta~\Gamma\big)^k\quad,\label{eqn:widetildePi}
\end{flalign}
where we have simplified the notation by denoting the extensions of maps 
to symmetric algebras by capital symbols. Recall that the relevant perturbations $\delta$
are of the form
\begin{flalign}
\delta\,=\, \hbar\,\Delta_{\BV} + \{ \lambda\, I,\,\cdot\,\}\quad,
\end{flalign}
where $\Delta_{\BV}$ is the BV Laplacian \eqref{eqn:BVexplicit} and 
$\lambda\, I\in (\Sym\,E[1])^0$ denotes the $m+1$-point interaction term \eqref{eqn:mplus1point}
for some $m\geq 2$.
\sk

We are particularly interested in the correlation functions $\widetilde{\Pi}(\varphi_1\cdots\varphi_n)$
for test functions $\varphi_1,\dots,\varphi_n\in E[1]^0$ of degree zero; 
these describe the correlators of the physical field, in contrast to correlators involving antifields.
To work out the perturbative expansion \eqref{eqn:widetildePi}
of such correlators, we have to understand how the maps $\Pi$ and $\delta\,\Gamma$
act on elements $\varphi_1\cdots\varphi_n\in \Sym\,E[1]$ with all $\varphi_i\in E[1]^0$ of degree zero.
Because $\pi=0$ in the present case (cf.\ \eqref{eqn:massiveDefRet}), we have
\begin{flalign}\label{eqn:Pimapscalar}
\Pi(\oone) \,=\, 1~~,\quad \Pi(\varphi_1\cdots\varphi_n) \, =\,0\quad,
\end{flalign}
for all $n\geq 1$. To describe the map $\delta\,\Gamma = \hbar\,\Delta_{\BV}\,\Gamma + \{\lambda\,I, \,\cdot\,\}\,\Gamma$,
it is convenient to consider the two summands individually. For the first term, we use 
the definition \eqref{eqn:symhomotopy} of $\Gamma=\Sym\,\gamma$ and the explicit description 
of the BV Laplacian \eqref{eqn:BVexplicit}, resulting in 
\begin{flalign}\label{eqn:DeltaBVhomotopy}
\hbar\,\Delta_{\BV}\,\Gamma\big(\varphi_1\cdots \varphi_n\big)\,=\,
-\frac{2\,\hbar}{n}\,\sum_{i<j}\, \big(\varphi_i,G(\varphi_j)\big)~\varphi_1\cdots\widehat{\varphi}_i\cdots\widehat{\varphi}_j\cdots\varphi_n\quad,
\end{flalign}
where we recall that $G=-\gamma$ is the Green operator for $Q=\Delta + m^2$.
For the second term, we use the axioms of $P_0$-algebras (cf.\ Remark \ref{rem:P0algebraTrad})
and the explicit expression \eqref{eqn:mplus1point} for the $m+1$-point interaction term (together with
its symmetry property \eqref{eqn:Isymmetry}), resulting in
\begin{multline}\label{eqn:Ihomotopy}
\big\{\lambda\,I,\Gamma(\varphi_1\cdots\varphi_n)\big\} \,=\,\\[4pt]
-\frac{\lambda^{m-1}}{m!\,n}\, \sum_{i=1}^n \ \sum_{J_0,j_0,\dots,J_m,j_m}\,
\varphi_1\cdots\varphi_{i-1}~\,I_{j_0 j_1\cdots j_m}^{J_0 J_1\cdots J_m}\,
\big({Y^{J_0}_{j_0}}^\ast , G(\varphi_i)\big)~ {Y^{J_1}_{j_1}}^\ast\,
\cdots\,{Y^{J_m}_{j_m}}^\ast\,\varphi_{i+1}\cdots \varphi_n\quad.
\end{multline}

The two expressions in \eqref{eqn:DeltaBVhomotopy} and \eqref{eqn:Ihomotopy} 
admit a convenient graphical description. 
Depicting the element $\varphi_1\cdots \varphi_n$ by $n$ vertical lines,
the map in \eqref{eqn:DeltaBVhomotopy} may be depicted as
\begin{flalign}
\hbar\,\Delta_{\BV}\,\Gamma\Big(\,\vcenter{\hbox{\begin{tikzpicture}[scale=0.5]
\draw[thick] (0,0) -- (0,1);
\draw[thick] (0.5,0) -- (0.5,1);
\draw[thick] (1,0) -- (1,1);
\node (dottt) at (1.75,0.5) {$\cdots$};
\draw[thick] (2.5,0) -- (2.5,1);
\draw[thick] (3,0) -- (3,1);
\end{tikzpicture}}}\,\Big)
\, =\, -\frac{2\,\hbar}{n}\,\Big(\,
\vcenter{\hbox{\begin{tikzpicture}[scale=0.5]
\draw[thick] (0,0) -- (0,1);
\draw[thick] (0.5,0) -- (0.5,1);
\draw[thick] (0,1) to[out=90,in=90] (0.5,1);
\draw[thick] (1,0) -- (1,1);
\node (dottt) at (1.75,0.5) {$\cdots$};
\draw[thick] (2.5,0) -- (2.5,1);
\draw[thick] (3,0) -- (3,1);
\end{tikzpicture}}}
~+~
\vcenter{\hbox{\begin{tikzpicture}[scale=0.5]
\draw[thick] (0,0) -- (0,1);
\draw[thick] (0.5,0) -- (0.5,1);
\draw[thick] (1,0) -- (1,1);
\draw[thick] (0,1) to[out=90,in=90] (1,1);
\node (dottt) at (1.75,0.5) {$\cdots$};
\draw[thick] (2.5,0) -- (2.5,1);
\draw[thick] (3,0) -- (3,1);
\end{tikzpicture}}}
~+~\cdots~+~
\vcenter{\hbox{\begin{tikzpicture}[scale=0.5]
\draw[thick] (0,0) -- (0,1);
\draw[thick] (0.5,0) -- (0.5,1);
\draw[thick] (1,0) -- (1,1);
\node (dottt) at (1.75,0.5) {$\cdots$};
\draw[thick] (2.5,0) -- (2.5,1);
\draw[thick] (3,0) -- (3,1);
\draw[thick] (2.5,1) to[out=90,in=90] (3,1);
\end{tikzpicture}}}
\,\Big)\quad,
\end{flalign}
where the cap indicates a contraction of two elements
with respect to $(\,\cdot\, ,G(\cdot))$.
The map in \eqref{eqn:Ihomotopy} may be depicted as
\begin{flalign}
\Big\{ \lambda\,I\,,\, \Gamma\Big(\,\vcenter{\hbox{\begin{tikzpicture}[scale=0.5]
\draw[thick] (0,0) -- (0,1);
\draw[thick] (0.5,0) -- (0.5,1);
\draw[thick] (1,0) -- (1,1);
\node (dottt) at (1.75,0.5) {$\cdots$};
\draw[thick] (2.5,0) -- (2.5,1);
\draw[thick] (3,0) -- (3,1);
\end{tikzpicture}}}\,\Big)\Big\}
\, =\, -\frac{\lambda^{m-1}}{m!\,n}\,\bigg(
\vcenter{\hbox{\begin{tikzpicture}[scale=0.5]
\draw[thick] (-0.5,0) -- (-0.5,0.5);
\draw[thick] (-0.5,0.5) -- (0,1);
\draw[thick] (-0.5,0.5) -- (-0.25,1);
\draw[thick] (-0.5,0.5) -- (-0.75,1);
\draw[thick] (-0.5,0.5) -- (-1,1);
\node (ov) at (-0.5,1.5)  {\footnotesize{$m\text{ legs}$}};
\node (ol) at (-0.5,-0.75)  {\footnotesize{~~}};
\draw[thick] (0.5,0) -- (0.5,1);
\draw[thick] (1,0) -- (1,1);
\node (dottt) at (1.75,0.6) {$\cdots$};
\draw[thick] (2.5,0) -- (2.5,1);
\draw[thick] (3,0) -- (3,1);
\end{tikzpicture}}}
~+~\cdots~+~
\vcenter{\hbox{\begin{tikzpicture}[scale=0.5]
\draw[thick] (0,0) -- (0,1);
\draw[thick] (0.5,0) -- (0.5,1);
\draw[thick] (1,0) -- (1,1);
\node (dottt) at (1.75,0.5) {$\cdots$};
\draw[thick] (2.5,0) -- (2.5,1);
\draw[thick] (3.5,0) -- (3.5,0.5);
\draw[thick] (3.5,0.5) -- (3,1);
\draw[thick] (3.5,0.5) -- (3.25,1);
\draw[thick] (3.5,0.5) -- (3.75,1);
\draw[thick] (3.5,0.5) -- (4,1);
\node (ov) at (3.5,1.5)  {\footnotesize{$m\text{ legs}$}};
\node (ol) at (3.5,-0.75)  {\footnotesize{~~}};
\end{tikzpicture}}}
\bigg)\quad,
\end{flalign}
where the vertex acts on an element as
$\sum_{J_0,j_0,\dots,J_m,j_m}\,I_{j_0 j_1\cdots j_m}^{J_0 J_1\cdots J_m}\,
\big({Y^{J_0}_{j_0}}^\ast , G(\cdot)\big)~ {Y^{J_1}_{j_1}}^\ast\,
\cdots\,{Y^{J_m}_{j_m}}^\ast$, i.e.\ it turns a single vertical line
into $m$ legs.

\begin{ex}\label{ex:fuzzysphere2pt}
Let us set $m=3$ and compute the $2$-point function
\begin{flalign}\label{eqn:2ptfunction}
\widetilde{\Pi}\big(\varphi_1\,\varphi_2\big)
\,=\,\sum_{k=0}^\infty\, \Pi\,\big((\delta\,\Gamma)^k(\varphi_1\,\varphi_2)\big)
\end{flalign}
of $\Phi^4$-theory
to the lowest non-trivial order in the coupling constant. 
(Due to our conventions in \eqref{eqn:mplus1point},
the $4$-point interaction vertex has coupling constant $\lambda^2$.)
Using our graphical description, we compute
\begin{flalign}\label{eqn:deltaGamma1}
\delta\,\Gamma(\varphi_1\,\varphi_2)\,&=\,
-\hbar ~~\vcenter{\hbox{\begin{tikzpicture}[scale=0.5]
\draw[thick] (0,0) -- (0,1);
\draw[thick] (0.5,0) -- (0.5,1);
\draw[thick] (0,1) to[out=90,in=90] (0.5,1);
\end{tikzpicture}}}
~
-\frac{\lambda^2}{3!\,2}\,
\Big(\,
\vcenter{\hbox{\begin{tikzpicture}[scale=0.5]
\draw[thick] (-0.5,0) -- (-0.5,0.5);
\draw[thick] (-0.5,0.5) -- (0,1);
\draw[thick] (-0.5,0.5) -- (-0.5,1);
\draw[thick] (-0.5,0.5) -- (-1,1);
\draw[thick] (0.5,0) -- (0.5,1);
\end{tikzpicture}}}
~+~
\vcenter{\hbox{\begin{tikzpicture}[scale=0.5]
\draw[thick] (-1.5,0) -- (-1.5,1);
\draw[thick] (-0.5,0) -- (-0.5,0.5);
\draw[thick] (-0.5,0.5) -- (0,1);
\draw[thick] (-0.5,0.5) -- (-0.5,1);
\draw[thick] (-0.5,0.5) -- (-1,1);
\end{tikzpicture}}}
\,\Big)\quad.
\end{flalign}
The $2$-fold application of $\delta\,\Gamma$ is then given by
\begin{flalign}
\nn (\delta\,\Gamma)^2(\varphi_1\,\varphi_2)\,&=\,
\frac{\lambda^2\,\hbar}{3!\,4}\,\Big(\,
\vcenter{\hbox{\begin{tikzpicture}[scale=0.5]
\draw[thick] (-0.5,0) -- (-0.5,0.5);
\draw[thick] (-0.5,0.5) -- (0,1);
\draw[thick] (-0.5,0.5) -- (-0.5,1);
\draw[thick] (-0.5,0.5) -- (-1,1);
\draw[thick] (0.5,0) -- (0.5,1);
\draw[thick] (-1,1) to[out=90,in=90] (-0.5,1);
\end{tikzpicture}}}
~+~
\vcenter{\hbox{\begin{tikzpicture}[scale=0.5]
\draw[thick] (-0.5,0) -- (-0.5,0.5);
\draw[thick] (-0.5,0.5) -- (0,1);
\draw[thick] (-0.5,0.5) -- (-0.5,1);
\draw[thick] (-0.5,0.5) -- (-1,1);
\draw[thick] (0.5,0) -- (0.5,1);
\draw[thick] (-1,1) to[out=90,in=90] (0,1);
\end{tikzpicture}}}
~+~
\vcenter{\hbox{\begin{tikzpicture}[scale=0.5]
\draw[thick] (-0.5,0) -- (-0.5,0.5);
\draw[thick] (-0.5,0.5) -- (0,1);
\draw[thick] (-0.5,0.5) -- (-0.5,1);
\draw[thick] (-0.5,0.5) -- (-1,1);
\draw[thick] (0.5,0) -- (0.5,1);
\draw[thick] (-1,1) to[out=90,in=90] (0.5,1);
\end{tikzpicture}}}
~+~
\vcenter{\hbox{\begin{tikzpicture}[scale=0.5]
\draw[thick] (-0.5,0) -- (-0.5,0.5);
\draw[thick] (-0.5,0.5) -- (0,1);
\draw[thick] (-0.5,0.5) -- (-0.5,1);
\draw[thick] (-0.5,0.5) -- (-1,1);
\draw[thick] (0.5,0) -- (0.5,1);
\draw[thick] (-0.5,1) to[out=90,in=90] (0,1);
\end{tikzpicture}}}
~+~
\vcenter{\hbox{\begin{tikzpicture}[scale=0.5]
\draw[thick] (-0.5,0) -- (-0.5,0.5);
\draw[thick] (-0.5,0.5) -- (0,1);
\draw[thick] (-0.5,0.5) -- (-0.5,1);
\draw[thick] (-0.5,0.5) -- (-1,1);
\draw[thick] (0.5,0) -- (0.5,1);
\draw[thick] (-0.5,1) to[out=90,in=90] (0.5,1);
\end{tikzpicture}}}
~+~
\vcenter{\hbox{\begin{tikzpicture}[scale=0.5]
\draw[thick] (-0.5,0) -- (-0.5,0.5);
\draw[thick] (-0.5,0.5) -- (0,1);
\draw[thick] (-0.5,0.5) -- (-0.5,1);
\draw[thick] (-0.5,0.5) -- (-1,1);
\draw[thick] (0.5,0) -- (0.5,1);
\draw[thick] (0,1) to[out=90,in=90] (0.5,1);
\end{tikzpicture}}}\\
\nn &\quad \hspace{1cm}~+~
\vcenter{\hbox{\begin{tikzpicture}[scale=0.5]
\draw[thick] (-1.5,0) -- (-1.5,1);
\draw[thick] (-0.5,0) -- (-0.5,0.5);
\draw[thick] (-0.5,0.5) -- (0,1);
\draw[thick] (-0.5,0.5) -- (-0.5,1);
\draw[thick] (-0.5,0.5) -- (-1,1);
\draw[thick] (-1.5,1) to[out=90,in=90] (-1,1);
\end{tikzpicture}}}
~+~
\vcenter{\hbox{\begin{tikzpicture}[scale=0.5]
\draw[thick] (-1.5,0) -- (-1.5,1);
\draw[thick] (-0.5,0) -- (-0.5,0.5);
\draw[thick] (-0.5,0.5) -- (0,1);
\draw[thick] (-0.5,0.5) -- (-0.5,1);
\draw[thick] (-0.5,0.5) -- (-1,1);
\draw[thick] (-1.5,1) to[out=90,in=90] (-0.5,1);
\end{tikzpicture}}}
~+~
\vcenter{\hbox{\begin{tikzpicture}[scale=0.5]
\draw[thick] (-1.5,0) -- (-1.5,1);
\draw[thick] (-0.5,0) -- (-0.5,0.5);
\draw[thick] (-0.5,0.5) -- (0,1);
\draw[thick] (-0.5,0.5) -- (-0.5,1);
\draw[thick] (-0.5,0.5) -- (-1,1);
\draw[thick] (-1.5,1) to[out=90,in=90] (0,1);
\end{tikzpicture}}}
~+~
\vcenter{\hbox{\begin{tikzpicture}[scale=0.5]
\draw[thick] (-1.5,0) -- (-1.5,1);
\draw[thick] (-0.5,0) -- (-0.5,0.5);
\draw[thick] (-0.5,0.5) -- (0,1);
\draw[thick] (-0.5,0.5) -- (-0.5,1);
\draw[thick] (-0.5,0.5) -- (-1,1);
\draw[thick] (-1,1) to[out=90,in=90] (-0.5,1);
\end{tikzpicture}}}
~+~
\vcenter{\hbox{\begin{tikzpicture}[scale=0.5]
\draw[thick] (-1.5,0) -- (-1.5,1);
\draw[thick] (-0.5,0) -- (-0.5,0.5);
\draw[thick] (-0.5,0.5) -- (0,1);
\draw[thick] (-0.5,0.5) -- (-0.5,1);
\draw[thick] (-0.5,0.5) -- (-1,1);
\draw[thick] (-1,1) to[out=90,in=90] (0,1);
\end{tikzpicture}}}
~+~
\vcenter{\hbox{\begin{tikzpicture}[scale=0.5]
\draw[thick] (-1.5,0) -- (-1.5,1);
\draw[thick] (-0.5,0) -- (-0.5,0.5);
\draw[thick] (-0.5,0.5) -- (0,1);
\draw[thick] (-0.5,0.5) -- (-0.5,1);
\draw[thick] (-0.5,0.5) -- (-1,1);
\draw[thick] (-0.5,1) to[out=90,in=90] (-0,1);
\end{tikzpicture}}}
\,\Big)+\mathcal{O}(\lambda^4)
\\[4pt]
\,&=\,
\frac{\lambda^2\,\hbar}{8}\,\Big(\,
\vcenter{\hbox{\begin{tikzpicture}[scale=0.5]
\draw[thick] (-0.5,0) -- (-0.5,0.5);
\draw[thick] (-0.5,0.5) -- (0,1);
\draw[thick] (-0.5,0.5) -- (-0.5,1);
\draw[thick] (-0.5,0.5) -- (-1,1);
\draw[thick] (0.5,0) -- (0.5,1);
\draw[thick] (-1,1) to[out=90,in=90] (-0.5,1);
\end{tikzpicture}}}
~+~2~
\vcenter{\hbox{\begin{tikzpicture}[scale=0.5]
\draw[thick] (-0.5,0) -- (-0.5,0.5);
\draw[thick] (-0.5,0.5) -- (0,1);
\draw[thick] (-0.5,0.5) -- (-0.5,1);
\draw[thick] (-0.5,0.5) -- (-1,1);
\draw[thick] (0.5,0) -- (0.5,1);
\draw[thick] (0,1) to[out=90,in=90] (0.5,1);
\end{tikzpicture}}}
~+~
\vcenter{\hbox{\begin{tikzpicture}[scale=0.5]
\draw[thick] (-1.5,0) -- (-1.5,1);
\draw[thick] (-0.5,0) -- (-0.5,0.5);
\draw[thick] (-0.5,0.5) -- (0,1);
\draw[thick] (-0.5,0.5) -- (-0.5,1);
\draw[thick] (-0.5,0.5) -- (-1,1);
\draw[thick] (-0.5,1) to[out=90,in=90] (-0,1);
\end{tikzpicture}}}
\,\Big)~+~\mathcal{O}(\lambda^4)\quad, \label{eqn:deltaGamma2}
\end{flalign}
where the simplification in the second equality uses the symmetry
property of the interaction term \eqref{eqn:Isymmetry}.
The $3$-fold application of $\delta\,\Gamma$ 
is given by
\begin{flalign}
(\delta\,\Gamma)^3(\varphi_1\,\varphi_2)\, =\,
-\frac{\lambda^2\,\hbar^2}{2}~~\vcenter{\hbox{\begin{tikzpicture}[scale=0.5]
\draw[thick] (-0.5,0) -- (-0.5,0.5);
\draw[thick] (-0.5,0.5) -- (0,1);
\draw[thick] (-0.5,0.5) -- (-0.5,1);
\draw[thick] (-0.5,0.5) -- (-1,1);
\draw[thick] (0.5,0) -- (0.5,1);
\draw[thick] (-1,1) to[out=90,in=90] (-0.5,1);
\draw[thick] (0,1) to[out=90,in=90] (0.5,1);
\end{tikzpicture}}} ~+~ \mathcal{O}(\lambda^4) \quad. \label{eqn:deltaGamma3}
\end{flalign}

From this we can compute the $2$-point function \eqref{eqn:2ptfunction}
to leading order in the coupling constant as
\begin{flalign}
\widetilde{\Pi}(\varphi_1\,\varphi_2)\,&=\,
-\hbar ~~\vcenter{\hbox{\begin{tikzpicture}[scale=0.5]
\draw[thick] (0,0) -- (0,1);
\draw[thick] (0.5,0) -- (0.5,1);
\draw[thick] (0,1) to[out=90,in=90] (0.5,1);
\end{tikzpicture}}}
~-\frac{\lambda^2\,\hbar^2}{2}~~\vcenter{\hbox{\begin{tikzpicture}[scale=0.5]
\draw[thick] (-0.5,0) -- (-0.5,0.5);
\draw[thick] (-0.5,0.5) -- (0,1);
\draw[thick] (-0.5,0.5) -- (-0.5,1);
\draw[thick] (-0.5,0.5) -- (-1,1);
\draw[thick] (0.5,0) -- (0.5,1);
\draw[thick] (-1,1) to[out=90,in=90] (-0.5,1);
\draw[thick] (0,1) to[out=90,in=90] (0.5,1);
\end{tikzpicture}}} ~+~ \mathcal{O}(\lambda^4)\\[4pt]
\nn \,&=\, -\hbar ~\big(\varphi_1,G(\varphi_2)\big) \\
\nn & \quad \, \, - \frac{\lambda^2\,\hbar^2}2\!\!
\sum_{J_0,j_0,\dots,J_3,j_3}\!\! I_{j_0j_1j_2j_3}^{J_0J_1J_2J_3}\,
\big({Y^{J_0}_{j_0}}^\ast , G(\varphi_1)\big)\,
\big({Y^{J_1}_{j_1}}^\ast , G({Y^{J_2}_{j_2}}^\ast)\big)\,\big({Y^{J_3}_{j_3}}^\ast , G(\varphi_2)\big)
+\mathcal{O}(\lambda^4) \ .
\end{flalign}
Note that the $2$-point function at order $\lambda^2$ (and higher) receives both planar
and non-planar contributions, analogously to the computation of 
\cite{CMS01} by traditional perturbative techniques, 
even though this is not directly apparent in our graphical presentation.
The origin of these two kinds of contributions
lies in the (graded anti-)symmetrization 
of the higher $L_\infty$-algebra bracket \eqref{eqn:ellmscalar} which enters
the definition of the constants $I_{j_0j_1j_2j_3}^{J_0J_1J_2J_3} $ in \eqref{eqn:mplus1point}.
\end{ex}

\paragraph{Correlation functions for $\boldsymbol{m^2=0}$.}
Let us briefly comment on the correlation functions 
in the massless case $m^2=0$. The relevant cochain map $\pi= \frac{1}{N+1}\,\Tr$ 
in the massless strong deformation retract \eqref{eqn:masslessDefRet} is not
simply the zero map, but the normalized trace. Hence, in contrast to \eqref{eqn:Pimapscalar},
the extension of $\pi$ to symmetric algebras is given in the massless case by
\begin{flalign}\label{eqn:Pimapscalarmassless}
\Pi(\oone) \,=\, 1 \ \in \ \Sym\,\bbC~~,\quad \Pi(\varphi_1\cdots\varphi_n) \, =\,\pi(\varphi_1) \odot \cdots \odot\pi(\varphi_n) \ \in \ \Sym\,\bbC\quad,
\end{flalign}
for all $\varphi_1,\dots,\varphi_n\in E[1]^0$ in degree $0$,
where we use the symbol $\odot$ to denote the product
of the symmetric algebra $\Sym\,\bbC$ to distinguish
it from the multiplication of complex numbers. Each 
$\pi(\varphi_i)\in\Sym\,\bbC$ is regarded as a linear function
on the space of vacua $\ker (\Delta : A\to A)\cong \bbC$ via
\begin{flalign}
\pi(\varphi_i)\,:\, \ker (\Delta : A\to A) ~\longrightarrow~\bbC~,~~
\underline{\Phi} \,=\, \Phi_0\,\oone ~\longmapsto~\frac{1}{N+1}\,\Tr\big(\varphi_i \,\underline{\Phi}\big)
\,=\, \pi(\varphi_i)\,\Phi_0\quad,
\end{flalign}
where the product (denoted by juxtaposition) in the last step is the usual multiplication of complex numbers.
\sk

It is convenient to denote \eqref{eqn:Pimapscalarmassless} graphically 
by attaching vertices on top of the vertical lines
\begin{flalign}
\Pi\Big(\,\vcenter{\hbox{\begin{tikzpicture}[scale=0.5]
\draw[thick] (0,0) -- (0,1);
\draw[thick] (0.5,0) -- (0.5,1);
\draw[thick] (1,0) -- (1,1);
\node (dottt) at (1.75,0.5) {$\cdots$};
\draw[thick] (2.5,0) -- (2.5,1);
\draw[thick] (3,0) -- (3,1);
\end{tikzpicture}}}\,\Big)
~=~\vcenter{\hbox{\begin{tikzpicture}[scale=0.5]
\draw[thick] (0,0) -- (0,1); 
\draw[black,fill=black] (0,1) circle (.75ex);
\draw[thick] (0.5,0) -- (0.5,1);
\draw[black,fill=black] (0.5,1) circle (.75ex);
\draw[thick] (1,0) -- (1,1);
\draw[black,fill=black] (1,1) circle (.75ex);
\node (dottt) at (1.75,0.5) {$\cdots$};
\draw[thick] (2.5,0) -- (2.5,1);
\draw[black,fill=black] (2.5,1) circle (.75ex);
\draw[thick] (3,0) -- (3,1);
\draw[black,fill=black] (3,1) circle (.75ex);
\end{tikzpicture}}}\quad,
\end{flalign}
which depict empty slots that 
can be evaluated on classical vacua $\underline{\Phi}\in \ker (\Delta : A\to A)\cong \bbC$.
These purely classical contributions to the correlation functions
are completely analogous to those one would obtain
in traditional approaches to quantum field theory by expanding
the field operator $\widehat{\Phi} + \underline{\Phi}$ around a generic classical
solution $\underline{\Phi}$. 

\begin{ex}
For the $2$-point function of massless $\Phi^4$-theory, 
using \eqref{eqn:deltaGamma1}--\eqref{eqn:deltaGamma3} we obtain
\begin{flalign}
\nn \widetilde{\Pi}(\varphi_1\,\varphi_2)~&=~
\vcenter{\hbox{\begin{tikzpicture}[scale=0.5]
\draw[thick] (0,0) -- (0,1); 
\draw[black,fill=black] (0,1) circle (.75ex);
\draw[thick] (0.5,0) -- (0.5,1);
\draw[black,fill=black] (0.5,1) circle (.75ex);
\end{tikzpicture}}}
~ -\hbar ~~\vcenter{\hbox{\begin{tikzpicture}[scale=0.5]
\draw[thick] (0,0) -- (0,1);
\draw[thick] (0.5,0) -- (0.5,1);
\draw[thick] (0,1) to[out=90,in=90] (0.5,1);
\end{tikzpicture}}}
~
-\frac{\lambda^2}{3!\,2}~
\Big(\,
\vcenter{\hbox{\begin{tikzpicture}[scale=0.5]
\draw[thick] (-0.5,0) -- (-0.5,0.5);
\draw[thick] (-0.5,0.5) -- (0,1);
\draw[black,fill=black] (0,1) circle (.75ex);
\draw[thick] (-0.5,0.5) -- (-0.5,1);
\draw[black,fill=black] (-0.5,1) circle (.75ex);
\draw[thick] (-0.5,0.5) -- (-1,1);
\draw[black,fill=black] (-1,1) circle (.75ex);
\draw[thick] (0.5,0) -- (0.5,1);
\draw[black,fill=black] (0.5,1) circle (.75ex);
\end{tikzpicture}}}
~+~
\vcenter{\hbox{\begin{tikzpicture}[scale=0.5]
\draw[thick] (-1.5,0) -- (-1.5,1);
\draw[black,fill=black] (-1.5,1) circle (.75ex);
\draw[thick] (-0.5,0) -- (-0.5,0.5);
\draw[thick] (-0.5,0.5) -- (0,1);
\draw[black,fill=black] (0,1) circle (.75ex);
\draw[thick] (-0.5,0.5) -- (-0.5,1);
\draw[black,fill=black] (-0.5,1) circle (.75ex);
\draw[thick] (-0.5,0.5) -- (-1,1);
\draw[black,fill=black] (-1,1) circle (.75ex);
\end{tikzpicture}}}
\,\Big)\\
 &\quad \, ~ +\frac{\lambda^2\,\hbar}{8}\,\Big(\,
\vcenter{\hbox{\begin{tikzpicture}[scale=0.5]
\draw[thick] (-0.5,0) -- (-0.5,0.5);
\draw[thick] (-0.5,0.5) -- (0,1);
\draw[black,fill=black] (0,1) circle (.75ex);
\draw[thick] (-0.5,0.5) -- (-0.5,1);
\draw[thick] (-0.5,0.5) -- (-1,1);
\draw[thick] (0.5,0) -- (0.5,1);
\draw[black,fill=black] (0.5,1) circle (.75ex);
\draw[thick] (-1,1) to[out=90,in=90] (-0.5,1);
\end{tikzpicture}}}
~+~2~
\vcenter{\hbox{\begin{tikzpicture}[scale=0.5]
\draw[thick] (-0.5,0) -- (-0.5,0.5);
\draw[thick] (-0.5,0.5) -- (0,1);
\draw[thick] (-0.5,0.5) -- (-0.5,1);
\draw[black,fill=black] (-0.5,1) circle (.75ex);
\draw[thick] (-0.5,0.5) -- (-1,1);
\draw[black,fill=black] (-1,1) circle (.75ex);
\draw[thick] (0.5,0) -- (0.5,1);
\draw[thick] (0,1) to[out=90,in=90] (0.5,1);
\end{tikzpicture}}}
~+~
\vcenter{\hbox{\begin{tikzpicture}[scale=0.5]
\draw[thick] (-1.5,0) -- (-1.5,1);
\draw[black,fill=black] (-1.5,1) circle (.75ex);
\draw[thick] (-0.5,0) -- (-0.5,0.5);
\draw[thick] (-0.5,0.5) -- (0,1);
\draw[thick] (-0.5,0.5) -- (-0.5,1);
\draw[thick] (-0.5,0.5) -- (-1,1);
\draw[black,fill=black] (-1,1) circle (.75ex);
\draw[thick] (-0.5,1) to[out=90,in=90] (-0,1);
\end{tikzpicture}}}
\,\Big) ~ -\frac{\lambda^2\,\hbar^2}{2}~~\vcenter{\hbox{\begin{tikzpicture}[scale=0.5]
\draw[thick] (-0.5,0) -- (-0.5,0.5);
\draw[thick] (-0.5,0.5) -- (0,1);
\draw[thick] (-0.5,0.5) -- (-0.5,1);
\draw[thick] (-0.5,0.5) -- (-1,1);
\draw[thick] (0.5,0) -- (0.5,1);
\draw[thick] (-1,1) to[out=90,in=90] (-0.5,1);
\draw[thick] (0,1) to[out=90,in=90] (0.5,1);
\end{tikzpicture}}} ~+~ \mathcal{O}(\lambda^4) \quad,
\end{flalign}
as an element in $\Sym\,\bbC$.
\end{ex}

\subsection{\label{subsec:CS}Chern-Simons theory}
The fuzzy $2$-sphere has a well-known $3$-dimensional differential 
calculus which allows for the definition of a Chern-Simons term 
on $\bbS^2_N$, see e.g.~\cite{ARS00,GMS01}. Similarly to~\cite{GMS01}, 
we shall focus on the Abelian Chern-Simons theory on $\bbS^2_N$
which, due to the noncommutativity of the differential calculus on $\bbS_N^2$,
includes a ternary interaction term; the extension to non-Abelian 
Chern-Simons theory with matrix gauge algebra such as $\mathfrak{gl}(n)$ 
or $\mathfrak{u}(n)$ is straightforward, as in~\cite{ARS00}, and presents 
no essential novelties. This is the simplest 
example which serves as the prototype for the BV formalism applied 
to field theories with gauge symmetries. On $\bbS_N^2$ it can be 
regarded as a fuzzy version of a BF-type theory on the classical 
$2$-sphere $\bbS^2$~\cite{GMS01}.

\paragraph{Differential calculus on the fuzzy 2-sphere.} 
In order to set up Chern-Simons gauge theory within the framework outlined in Section 
\ref{sec:prelim}, we recall some basic facts about differential forms
on the fuzzy $2$-sphere. The usual $\su(2)$-equivariant differential calculus
on the fuzzy sphere algebra \eqref{eqn:sphereAlgebra} is given by
the Chevalley-Eilenberg dg-algebra
\begin{flalign}
\Omega^\bullet(A) \,:=\, \mathrm{CE}^\bullet(\su(2),A)\, =\, A\otimes 
\scalebox{1.2}{$\wedge$}^\bullet \su(2)^\ast\quad.
\end{flalign}
The dual of the Lie algebra basis $\{e_i\in \su(2)\}_{i=1,2,3}$ 
defines a basis $\{ \theta^i \in \Omega^1(A)\}_{i=1,2,3}$
for the $A$-module of $1$-forms, which generates the whole differential calculus $\Omega^\bullet(A)$.
This basis is central, i.e.\ $a \, \theta^i = \theta^i\,a$ for all $a\in A=\Omega^0(A)$,
and $\theta^i\wedge \theta^j = -\theta^j\wedge \theta^i$, for all $i,j=1,2,3$.
The de Rham differential is specified by
\begin{flalign}\label{eqn:deRhamdifferential}
\dd a \,=\, \frac{1}{\lambda_N}\, [X_i,a]\,\theta^i\quad,\qquad \dd \theta^i \,=\, -\frac{\ii}{2}\,\epsilon^{ijk}\,\theta^j\wedge \theta^k\quad,
\end{flalign}
for all $a\in A=\Omega^0(A)$ and $i=1,2,3$,
together with the graded Leibniz rule
\begin{flalign}
\dd(\omega\wedge \zeta) \,= \,
(\dd \omega)\wedge \zeta + (-1)^{p}\,\omega\wedge(\dd\zeta)\quad,
\end{flalign} 
for all $\omega\in \Omega^p(A)$ and $\zeta\in\Omega^\bullet(A)$.
Note that the differential calculus $\Omega^\bullet(A)$ on the fuzzy $2$-sphere
is $3$-dimensional, in contrast to the $2$-dimensional calculus on the commutative $2$-sphere $\bbS^2$.
Higher-dimensional (covariant) calculi are a common feature in noncommutative geometry
which arise in a broad range of examples, reaching from the fuzzy sphere to quantum groups.
\sk

On the top-degree forms $\Omega^3(A)$ we can define an integration map via
\begin{flalign}\label{eqn:integrationforms}
\int \,:\,\Omega^3(A)~\longrightarrow~\bbC~,~~\omega \,=\, a\,\theta^1\wedge 
\theta^2\wedge \theta^3 ~\longmapsto~\int\,\omega \,:=\,\frac{4\pi}{N+1}\, \Tr(a)\quad.
\end{flalign}
Using \eqref{eqn:deRhamdifferential} and the graded Leibniz rule, 
one easily checks that this integration map satisfies the Stokes theorem
\begin{flalign}
\int\,\dd\zeta \,=\, 0\quad,
\end{flalign}
for all $2$-forms $\zeta=\tfrac{1}{2}\,\zeta_{ij}\,\theta^i\wedge \theta^j \in \Omega^2(A)$.
\sk

We shall also need the Hodge operator $\ast \,:\, \Omega^\bullet(A)~\to~\Omega^{3-\bullet}(A)$ 
on the fuzzy sphere, which is defined on the $A$-module
basis of $\Omega^\bullet(A)$ by
\begin{flalign}
\nn \ast(\oone) \,:=\,\tfrac{1}{3!} \,\epsilon^{ijk}\,\theta^i\wedge \theta^j\wedge \theta^k~~,&\quad
\ast(\theta^i)\,:=\, \tfrac{1}{2!}\, \epsilon^{ijk}\,\theta^j\wedge \theta^k~~,\quad\\[4pt]
\ast(\theta^i\wedge\theta^j)\,:=\, \epsilon^{ijk}\,\theta^k~~,&\quad
\ast(\theta^i\wedge\theta^j\wedge \theta^k)\,:=\, \epsilon^{ijk}\,\oone\quad.
\end{flalign}
Note that $\ast\ast(\omega) = \omega$, for all $\omega \in\Omega^p(A)$.
From this we can define the codifferential
\begin{flalign}
\delta\,:=\, (-1)^p\,\ast\dd\,\ast \,:\, \Omega^p(A)~\longrightarrow~\Omega^{p-1}(A)\quad
\end{flalign}
and the Hodge-de Rham Laplacian
\begin{flalign}\label{eqn:HdRLaplacian}
\Delta\,:=\, -\big(\delta\,\dd +\dd\,\delta\big)\,:\, \Omega^p(A)~\longrightarrow~\Omega^p(A)\quad,
\end{flalign}
for all $p=0,1,2,3$. A quick calculation shows that on $0$-forms the 
Hodge-de Rham Laplacian coincides with the scalar Laplacian \eqref{eqn:scalarsphereLaplacian}.

\paragraph{Free BV theory.}
We can now describe 
the non-interacting part of Abelian Chern-Simons theory
on the fuzzy sphere as a free BV theory in the sense of 
Definition~\ref{def:linearBV}. 
\begin{defi}\label{def:CSsphere}
The free BV theory associated to Abelian Chern-Simons theory
is given by the cochain complex
\begin{flalign}
E\,=\,\Omega^\bullet(A)[1]\,=\,\Big(\xymatrix@C=1.5em{
\stackrel{(-1)}{\Omega^0(A)} \ar[r]^-{-\dd} & 
\stackrel{(0)}{\Omega^1(A)} \ar[r]^-{-\dd} &
\stackrel{(1)}{\Omega^2(A)} \ar[r]^-{-\dd} 
& \stackrel{(2)}{\Omega^3(A)}
}\Big)\quad,
\end{flalign}
i.e.\ $Q:=\dd$ is the de Rham differential, together with the pairing
\begin{flalign}
\langle\,\cdot\, , \,\cdot\,\rangle\,:\, E\otimes E~\longrightarrow~\bbC[-1]~,\quad
\alpha\otimes\beta~\longmapsto~(-1)^{\vert \alpha\vert}\,\int\,\alpha\wedge\beta\quad,
\end{flalign}
where $\vert\alpha\vert$ denotes the cohomological degree of $\alpha\in E$.
(Note that the latter differs from the de Rham degree as
$\vert \alpha\vert_{\mathrm{dR}} = \vert \alpha\vert +1$.)
\end{defi}

Following the general approach outlined in Section \ref{subsec:BV},
we can construct from this input a $P_0$-algebra
\begin{flalign}
\Obs^\cl \,:=\,\big(\Sym\,E[1] , Q ,\{\,\cdot\,,\,\cdot\,\}\big)
\end{flalign}
of classical observables of the non-interacting theory,
where we note that the complex
\begin{flalign}\label{eqn:CSsphereLinObs}
E[1] \,=\, \Omega^\bullet(A)[2] \,=\, \Big(\xymatrix@C=1.5em{
\stackrel{(-2)}{\Omega^0(A)} \ar[r]^-{\dd} & 
\stackrel{(-1)}{\Omega^1(A)} \ar[r]^-{\dd} &
\stackrel{(0)}{\Omega^2(A)} \ar[r]^-{\dd} 
& \stackrel{(1)}{\Omega^3(A)}
}\Big)
\end{flalign}
is concentrated in degrees $-2$, $-1$, $0$ and $1$.

\paragraph{Interactions.} Using the
cyclic $L_\infty$-algebra formalism from Section \ref{subsec:Linfty}, 
we will now introduce an interaction term for the free Chern-Simons theory
from Definition \ref{def:CSsphere}. The Abelian cyclic $L_\infty$-algebra
associated with the free theory is given by the cochain complex
\begin{flalign}
E[-1] \,=\,\Omega^\bullet(A) \,=\, \Big(\xymatrix@C=1.5em{
\stackrel{(0)}{\Omega^0(A)} \ar[r]^-{\dd} & 
\stackrel{(1)}{\Omega^1(A)} \ar[r]^-{\dd} &
\stackrel{(2)}{\Omega^2(A)} \ar[r]^-{\dd} 
& \stackrel{(3)}{\Omega^3(A)}
}\Big)
\end{flalign}
and the cyclic structure
\begin{flalign}
\cyc{\,\cdot\,}{\,\cdot\,} \,:\, E[-1]\otimes E[-1]~\longrightarrow~\bbC[-3]~,~~
\alpha\otimes\beta ~\longmapsto~\cyc{\alpha}{\beta}\,=\,\int\, \alpha\wedge\beta\quad.
\end{flalign}
This can be endowed with the compatible $2$-bracket
\begin{flalign}
\ell_2\,:\,\Omega^\bullet(A)\otimes \Omega^\bullet(A)~\longrightarrow~\Omega^\bullet(A)~,~~
\alpha\otimes\beta~\longmapsto~[\alpha,\beta]\,:=\,\alpha\wedge \beta - (-1)^{\vert \alpha\vert \,\vert\beta\vert}\,\beta\wedge \alpha
\end{flalign}
given by the graded commutator in the differential calculus $\Omega^\bullet(A)$.
Note that, in contrast to commutative Chern-Simons theory, the bracket 
$\ell_2$ is {\em not} zero because the differential calculus on 
the fuzzy sphere $\bbS^2_N$ is noncommutative.
\sk

In order to write down the contracted coordinate functions corresponding to 
this non-Abelian cyclic dg-Lie algebra, we pick a basis of $E[-1]$ 
that we denote by 
$c_{a} \in E[-1]^{0} = \Omega^{0}(A)$, 
$A_{b} \in E[-1]^{1} = \Omega^{1}(A)$, 
$A^{+}_{c} \in E[-1]^{2} = \Omega^{2}(A)$ and 
$c^{+}_{d} \in E[-1]^{3} = \Omega^{3}(A)$. 
The contracted coordinate functions then take the form
\begin{flalign}
\mathsf{a}\,=\,\sum_{a}\, c^{\ast}_{a} \otimes c_{a}  + \sum_{b}\, A^{\ast}_{b} \otimes A_{b} 
+ \sum_{c}\, A^{+\ast}_{c} \otimes A^{+}_{c} + \sum_{d}\, c^{+\ast}_{d} \otimes c^{+}_{d}\,\in\,
\big((\Sym\,E[1])\otimes E[-1]\big)^1\quad,
\end{flalign}
where the dual basis with respect to the cyclic structure $\cyc{\,\cdot\,}{\,\cdot\,}$ 
is denoted by $c^{\ast}_{a} \in E[1]^{1} = \Omega^{3}(A)$, 
$A^{\ast}_{b} \in E[1]^{0} = \Omega^{2}(A)$,
$A^{+\ast}_{c} \in E[1]^{-1} = \Omega^{1}(A)$ 
and $c^{+\ast}_{d} \in E[1]^{-2} = \Omega^{0}(A)$.
The Chern-Simons interaction term thus reads as
\begin{flalign}
\nn \lambda \,I\,&=\,\frac{\lambda}{3!} \,
\cyc{\mathsf{a}}{\ell^{\mathrm{ext}}_2(\mathsf{a},\mathsf{a})}_{\mathrm{ext}}\\[4pt]
\,&=\, \frac{\lambda}{3!}\, \sum_{b,b',b''}\, A^{\ast}_{b}\, A^{\ast}_{b'}\, A^{\ast}_{b''}~
\cyc{A_{b}}{[A_{b'},A_{b''}]}
\nn \,-\, \lambda\,\sum_{a,b,c}\, c^{\ast}_{a}\, A^{\ast}_{b}\, A^{+\ast}_{c}~\cyc{c_{a}}{[A_{b},A^{+}_{c}]}\\
&\qquad \,+\, \frac{\lambda}{2}\, \sum_{d,a,a'}\, c^{+\ast}_{d}\, c^{\ast}_{a}\, c^{\ast}_{a'}~
\cyc{c^{+}_{d}}{[c_{a},c_{a'}]}\
\in\ (\Sym\,E[1])^0\quad,\label{eqn:CSinteraction}
\end{flalign}
where we stress again that the products of the dual basis elements 
are given by the product of the symmetric algebra $\Sym\,E[1]$.

\paragraph{Strong deformation retract.} The cohomology of the complex
\eqref{eqn:CSsphereLinObs} can be computed explicitly by
using the Whitehead lemma, see e.g. \cite[Theorem 7.8.9]{Wei94}. For this, we recall
that $\Omega^\bullet(A) = \mathrm{CE}^\bullet(\su(2),A)$ is by definition
the Chevalley-Eilenberg cochain complex of $\su(2)$ with coefficients
in the fuzzy sphere algebra \eqref{eqn:sphereAlgebra}. As an 
$\su(2)$-representation, the latter decomposes as $A\cong \bigoplus_{J=0}^{N}\,(J)$, where
$(J)$ denotes the irreducible spin $J$ representation, leading to
\begin{flalign}
\Omega^\bullet(A) \,\cong\, \bigoplus_{J=0}^N\, \mathrm{CE}^\bullet(\su(2),(J)) \quad.
\end{flalign}
By the Whitehead lemma, the cohomology of $\mathrm{CE}^\bullet(\su(2),(J))$ is trivial
for all $J>0$. This allows us to compute that
\begin{flalign}
H^\bullet\big(\Omega^\bullet(A)\big) \,\cong\, H^\bullet\big(\mathrm{CE}^\bullet(\su(2),(0))\big)
\,\cong\, \bbC\oplus \bbC[-3]
\end{flalign}
is concentrated in differential form degrees $0$ and $3$. From this it follows that
the cohomology of the complex \eqref{eqn:CSsphereLinObs} is given by
\begin{flalign}
H^\bullet(E[1])\,=\,H^\bullet\big(\Omega^\bullet(A)[2]\big)\,\cong\, \bbC[2] \oplus \bbC[-1]\quad.
\end{flalign}

To set up a strong deformation retract as in \eqref{eqn:defretL},
let us first define the cochain maps $\iota$ and $\pi$ for the present example.
We define the cochain map
\begin{flalign}
\parbox{0.5cm}{\xymatrix{
H^\bullet(E[1])\ar[d]_-{\iota}\\
E[1]
}
}~~=~~~~ \left(\parbox{2cm}{\xymatrix@C=2em{
\ar[d]_-{\eta}\bbC \ar[r]^-{0}~&~ \ar[d]_-{0}0 \ar[r]^-{0}~&  \ar[d]_-{0} 0 \ar[r]^-{0}~&~ \ar[d]_-{\ast\eta}\bbC\\
\Omega^0(A) \ar[r]_-{\dd}~&~ \Omega^1(A) \ar[r]_-{\dd}~&~ \Omega^2(A)\ar[r]_-{\dd}~&~\Omega^3(A)
}}\right)
\end{flalign}
via the unit $\eta(1)=\oone\in \Omega^0(A)=A$ and its Hodge dual
$\ast\eta(1) = \ast(\oone)\in \Omega^3(A)$. The cochain map
\begin{flalign}
\parbox{0.5cm}{\xymatrix{
E[1]\ar[d]_-{\pi}\\
H^\bullet(E[1])
}
}~~=~~~~ \left(\parbox{2cm}{\xymatrix@C=2em{
\ar[d]_-{\frac{1}{4\pi}\,\int\ast}\Omega^0(A) \ar[r]^-{\dd}~&~ \ar[d]_-{0}\Omega^1(A) \ar[r]^-{\dd}~&~ \ar[d]_-{0}\Omega^2(A)\ar[r]^-{\dd}~&~\ar[d]_-{\frac{1}{4\pi}\,\int}\Omega^3(A)\\
\bbC \ar[r]_-{0}~&~ 0 \ar[r]_-{0}~&~   0 \ar[r]_-{0}~&~ \bbC\\
}}\right)
\end{flalign}
is given by integration of top-forms \eqref{eqn:integrationforms}, where the normalization factor
$\frac{1}{4\pi}$ is chosen so that $\pi \, \iota = \id_{H^\bullet(E[1])}$. 
\sk

The composite cochain map
$\iota\, \pi : E[1]\to E[1]$ defines a projector onto the harmonic
forms, which can be used to decompose
\begin{flalign}
E[1]\,\cong\, E[1]^\perp \oplus H^\bullet(E[1])\quad.
\end{flalign}
By the rank-nullity theorem of linear algebra, the Hodge-de Rham Laplacian
restricts to an isomorphism $\Delta^\perp : E[1]^\perp\to E[1]^\perp$
and we denote its inverse (i.e.\ the Green operator) by $G^\perp: E[1]^\perp\to E[1]^\perp$.
Using also the codifferential $\delta$, we define the cochain homotopy
\begin{flalign}
\gamma\,:=\,\delta \, G^\perp \, \big(\id_{E[1]}- \iota \, \pi\big)\ \in\ \hom(E[1],E[1])^{-1}\quad,
\end{flalign}
which yields the desired strong deformation retract
\begin{equation}\label{eqn:CSDefRet}
\begin{tikzcd}
\big(\bbC[2]\oplus\bbC[-1],0\big) \ar[r,shift right=1ex,swap,"\iota"] & \ar[l,shift right=1ex,swap,"\pi"] \big(\Omega^\bullet(A)[2],\dd\big)\qquad\ar[loop,out=25,in=-25,distance=35,"\gamma"]
\end{tikzcd}
\end{equation}
for Chern-Simons theory. The relevant properties from Definition \ref{def:defret}
can be checked via simple routine calculations. For example, one checks that
\begin{flalign}
\nn \partial(\gamma) &\,=\, \dd \, \gamma + \gamma\, \dd \,=\,  
\dd \,\delta \, G^\perp \, \big(\id_{E[1]}- \iota \, \pi\big) + 
\delta \, G^\perp \, \big(\id_{E[1]}- \iota \, \pi\big)\,\dd\\[4pt]
\nn &\,=\,\dd \,\delta \, G^\perp \, \big(\id_{E[1]}- \iota \, \pi\big) + 
\delta \,\dd\,  G^\perp \, \big(\id_{E[1]}- \iota \, \pi\big)\\[4pt]
&\,=\,-\Delta^\perp\, G^\perp \, \big(\id_{E[1]}- \iota \, \pi\big)
\,=\, \iota\,\pi - \id_{E[1]}\quad.
\end{flalign}
In the second line we used the property that the projector $\iota\, \pi$ commutes
with $\dd$ because it is a cochain map, and also that the Green operator
commutes with $\dd$ because $\dd\, \Delta^\perp = \Delta^\perp\,\dd$.
In the last line we used the definition of the Hodge-de Rham Laplacian
\eqref{eqn:HdRLaplacian} and the definition of $G^\perp$ as the inverse
of $\Delta^\perp$.

\paragraph{Correlation functions.} The correlation
functions of Chern-Simons theory can be computed 
by starting from the strong deformation retract \eqref{eqn:CSDefRet} 
and considering the deformation (cf.\ Theorem~\ref{theo:HPL}) 
\begin{equation}
\begin{tikzcd}
\big(\Sym\, \big(\bbC[2]\oplus\bbC[-1]\big) ,\widetilde{\delta}\ \big) \ar[rr,shift right=1ex,swap,"\widetilde{\Sym\,\iota}"] && \ar[ll,shift right=1ex,swap,"\widetilde{\Pi}:=\widetilde{\Sym\,\pi}"] \big(\Sym\,E[1],Q+\delta \big) \qquad \quad \ \ar[loop,out=25,in=-25,distance=40,"\widetilde{\Gamma}:=\widetilde{\Sym\,\gamma}"]
\end{tikzcd}
\end{equation}
of its extension to symmetric algebras. The perturbation $\delta$
is given by
\begin{flalign}\label{eqn:deltaCS}
\delta\,=\, \hbar\,\Delta_{\BV} + \{ \lambda\, I,\,\cdot\,\}\quad,
\end{flalign}
where $\lambda\,I$ is the Chern-Simons interaction term \eqref{eqn:CSinteraction}.
The perturbative expansion
\begin{flalign}
\widetilde{\Pi}(\varphi_1\cdots\varphi_n)\,=\,\sum_{k=0}^\infty \,
\Pi \big((\delta\,\Gamma)^k(\varphi_1\cdots\varphi_n)\big)
\end{flalign}
of the $n$-point correlation functions can then be computed 
by using the algebraic properties of $\Gamma = \Sym\,\gamma$ 
(cf.\ \eqref{eqn:symhomotopy}), the BV Laplacian $\Delta_{\BV}$ (cf.\ \eqref{eqn:BVexplicit})
and the graded derivation $\{\lambda \, I,\,\cdot\,\}$ (cf.~the $P_0$-algebra axioms
in Remark \ref{rem:P0algebraTrad}). 
These computations are very similar to those
in the case of scalar field theories (cf.\ Section \ref{subsec:ScalarSphere}), and
hence we will not spell out any explicit examples of correlation functions
for Chern-Simons theory.


\section{\label{sec:braidedBV}BV quantization of braided field theories}
The definitions and constructions in Section \ref{sec:prelim}
can be generalized in a rather straightforward way
to the case of theories with a {\em triangular} Hopf algebra
symmetry. In the following we will spell out the details.
The quasi-triangular case is considerably more complicated
because it obstructs the formulation of symmetry properties
and the Jacobi identity, and we will not treat this more general case 
in the present paper. Nevertheless, following the terminology 
of~\cite{DCGRS20,DCGRS21}, we use the adjective `braided' (in contrast to
the categorically more accurate `symmetric braided') 
to emphasize the role of a non-identity triangular 
$R$-matrix, in addition to equivariance, in our treatment below.
\sk

\subsection{\label{subsec:Hopf}Triangular Hopf algebras and their representations}
We start by recalling some basic concepts and terminology from the theory of Hopf algebras
that are required in this paper. More details can be found 
in e.g.\ \cite{Maj95,BM20}.
\begin{defi}
A {\em Hopf algebra} is an associative unital algebra $H$ over $\bbK$
together with two algebra homomorphisms $\Delta : H\to H\otimes H$ ({\em coproduct}) and
$\epsilon : H\to \bbK$ ({\em counit}), as well as an algebra antihomomorphism $S : H\to H$ ({\em antipode})
satisfying
\begin{subequations}\label{eqn:Hopfaxioms}
\begin{flalign}
\label{eqn:coassociativity}(\Delta\otimes\id_H)\,\Delta\,&=\, (\id_H\otimes \Delta)\,\Delta\quad,\\[4pt]
(\epsilon\otimes \id_H) \, \Delta \,&=\, \id_H\,=\, (\id_H\otimes \epsilon)\,\Delta\quad,\\[4pt]
\mu\,(S\otimes \id_H) \,\Delta\,&=\, \eta\,\epsilon \,=\,\mu\,(\id_H\otimes S)\,\Delta\quad,
\end{flalign}
\end{subequations}
where $\mu : H\otimes H\to H$ denotes the product and $\eta : \bbK\to H$
denotes the unit of the algebra $H$.
\end{defi}
\begin{rem}
We shall often use the Sweedler notation
\begin{subequations}
\begin{flalign}
\Delta(h) \,=\, h_{\und{1}}\otimes h_{\und{2}}\qquad\text{(summation understood)}
\end{flalign}
for the coproduct of $h\in H$, and more generally
\begin{flalign}
\Delta^n(h) \,=\, h_{\und{1}}\otimes\cdots\otimes h_{\und{n+1}}\qquad\text{(summation understood)}
\end{flalign}
\end{subequations}
for the iterated applications of the coproduct. Note that, due to 
coassociativity \eqref{eqn:coassociativity}, it makes sense
to write $\Delta^2 = (\Delta\otimes\id_H)\,\Delta = (\id_H\otimes \Delta)\,\Delta$
for the two-fold application of the coproduct, and similarly
$\Delta^n$ for the $n$-fold applications.
In Sweedler notation, the second and third properties in \eqref{eqn:Hopfaxioms}
read as
\begin{subequations}
\begin{flalign}
\epsilon(h_{\und{1}})\,h_{\und{2}} &\,=\, h \,=\, h_{\und{1}}\,\epsilon(h_{\und{2}})\quad,\\[4pt]
S(h_{\und{1}})\,h_{\und{2}} &\,=\,\epsilon(h)\,\oone\,=\, h_{\und{1}}\,S(h_{\und{2}})\quad, 
\end{flalign}
\end{subequations}
for all $h\in H$, where $\oone=\eta(1)$.
\end{rem}

Given any Hopf algebra $H$, we denote by ${}_H\Mod$ the category of left modules
over its underlying associative unital algebra. An object in 
${}_H\Mod$ is a vector space $V$ together with a linear map
$\triangleright : H\otimes V\to V\,,~h\otimes v\mapsto h\triangleright v$ 
({\em left action}) satisfying
\begin{flalign}
(h\,h^\prime)\ra v\,=\, h\ra(h^\prime\ra v)\quad,\qquad \oone\ra v\,=\,v\quad,
\end{flalign}
for all $h,h^\prime\in H$ and $v\in V$. The morphisms in ${}_H\Mod$
are $H$-equivariant linear maps, i.e.\ linear maps $f : V\to W$
satisfying $f(h\ra v) = h\ra f(v)$, for all $h\in H$ and $v\in V$.
\sk

Using the coproduct and the counit of $H$, one defines a monoidal structure
on the category ${}_H\Mod$. The monoidal product $V\otimes W$
of two objects $V,W\in  {}_H\Mod$ is given by the tensor product of
the underlying vector spaces together with the left tensor product action
\begin{flalign}\label{eqn:tensorleftaction}
h\ra(v\otimes w)\,:=\,(h_{\und{1}}\ra v) \otimes (h_{\und{2}}\ra w) \quad,
\end{flalign}
for all $h\in H$, $v\in V$ and $w\in W$. The monoidal unit is given by
endowing the one-dimensional vector space $\bbK$ with the trivial
left action $h\ra c = \epsilon(h)\,c$, for all $h\in H$ and $c\in \bbK$.
\sk

Using the antipode of $H$, one observes that this monoidal structure
is closed. The internal hom between two objects $V,W\in  {}_H\Mod$
is the vector space $\hom(V,W):=
\Hom_\bbK(V,W)$ of all (not necessarily $H$-equivariant) linear maps from $V$ to $W$
together with the left adjoint action
\begin{flalign}\label{eqn:homleftaction}
h\ra f\,:=\, (h_{\und{1}}\ra \cdot)\circ f\circ (S(h_{\und{2}})\ra \cdot)\quad,
\end{flalign}
for all $h\in H$ and all linear maps $f:V\to W$. The $H$-invariants
of $\hom(V,W)$ are precisely the morphisms from $V$ to $W$ in the category ${}_H\Mod$.
\sk

To further obtain a closed {\em symmetric} monoidal structure on
the category ${}_H\Mod$, we require an additional datum
on the Hopf algebra $H$. 
In the following definition, we denote by
$\Delta^{\mathrm{op}}(h) = h_{\und{2}}\otimes h_{\und{1}}$
the opposite coproduct and, for an element
$R = R^\alpha\otimes R_{\alpha}\in H\otimes H$ (summation over $\alpha$ understood),
we write 
\begin{subequations}
\begin{flalign}
R_{21}\,:=\, R_\alpha\otimes R^\alpha
\end{flalign} 
for the flipped element in $H\otimes H$ and
\begin{flalign}
R_{12}\,:=\, R^\alpha\otimes R_{\alpha}\otimes \oone\quad ,\qquad
R_{13}\,:=\, R^\alpha\otimes \oone\otimes R_{\alpha}\quad ,\qquad
R_{23}\,:=\, \oone\otimes R^\alpha\otimes R_{\alpha}\quad
\end{flalign}
\end{subequations}
for the associated elements in $H\otimes H\otimes H$.
\begin{defi}\label{def:Rmatrix}
A {\em quasi-triangular structure} for a Hopf algebra $H$
is an invertible element $R \in H\otimes H$ ({\em universal $R$-matrix})
satisfying
\begin{subequations}\label{eqn:Rmatrixproperties}
\begin{flalign}
\Delta^{\mathrm{op}}(h) &\,=\, R\,\Delta(h)\,R^{-1}\quad,\\[4pt]
(\id_H\otimes \Delta)(R)&\,=\, R_{13}\,R_{12}\quad,\\[4pt]
(\Delta\otimes\id_H)(R)&\,=\,R_{13} \, R_{23}\quad,
\end{flalign}
\end{subequations}
for all $h\in H$. A {\em triangular structure} 
is a quasi-triangular structure $R \in H\otimes H$
which additionally satisfies $R_{21} = R^{-1}$.
\end{defi}

Suppose now that $R=R^\alpha\otimes R_\alpha\in H\otimes H$
is a quasi-triangular structure for $H$. Then we can define
a braiding for the closed monoidal category
${}_H\Mod$ by setting
\begin{flalign}\label{eqn:tauR}
\tau_R \,:\, V\otimes W~\longrightarrow~ W\otimes V~,~~v\otimes w~\longmapsto~ (R_{\alpha}\ra w)\otimes (R^\alpha\ra v)\quad,
\end{flalign}
for every pair of objects $V,W\in{}_H\Mod$. In the case where
$R$ is triangular, this braiding is symmetric,
i.e.\ ${}_H\Mod$ is a closed symmetric monoidal category.
\begin{ex}\label{ex:HopfAlgebra}
Every finite group $G$ has an associated {\em group Hopf algebra}
$\bbK[G]$. This is the free vector space
spanned by the group elements $g\in G$,
i.e.\ every element $h\in \bbK[G]$ can be written uniquely
as $h=\sum_{g\in G}\, h_{g}\, g$ for some $h_g\in \bbK$.
The product on $\bbK[G]$ is given by the bilinear extension
of the group operation on $G$ and the unit element
is the basis vector associated with the identity element $e\in G$. The coproduct, counit
and antipode are defined by
\begin{flalign}
\Delta(g) \,=\, g\otimes g\quad,\qquad \epsilon(g)\,=\,1 \quad,\qquad S(g) \,=\, g^{-1}\quad,
\end{flalign}
for all $g\in G$, and by linear extension. 
\sk

Depending on the group $G$, 
the Hopf algebra $\bbK[G]$ may admit various (quasi-)triangular structures, 
see below for a concrete example.
The trivial choice, which exists for any group $G$, 
is the element $R_{\mathrm{triv}} = e\otimes e\in \bbK[G]\otimes\bbK[G]$.
It is easy to check that the corresponding closed symmetric 
monoidal category ${}_{\bbK[G]}\Mod$ of $\bbK[G]$-modules 
is the usual closed symmetric monoidal category $\mathrm{Rep}_\bbK(G)$ of 
$\bbK$-linear representations of $G$.
\sk

The following class of explicit examples of (non-trivial) triangular Hopf algebras
features in our study of braided field theories on the fuzzy torus in Section \ref{sec:fuzzytorus}.
Set the ground field to $\bbK=\bbC$ and let $N,n\in \bbZ_{\geq 1}$ be two positive integers. Let us write
\begin{flalign}
\bbZ_N^n\,:=\, \underbrace{\bbZ_N\times\cdots\times\bbZ_N}_{\text{$n$-times}}
\end{flalign}
for the $n$-fold product of the cyclic group $\bbZ_N$ of order $N$.
We denote its elements by
\begin{flalign}
\und{k} \,:=\, (k_1,\dots,k_n)\ \in\ \bbZ_N^n\quad,
\end{flalign}
where each entry $k_i$ will be represented by an integer modulo $N$, and the group operation is
given by addition modulo $N$. For any $N$-th root of unity $q\in \bbC$
and any $n\times n$-matrix $\Theta\in\mathrm{Mat}_{n}(\bbZ)$ 
with integer entries, the element
\begin{flalign}
R \,:=\, \frac{1}{N^n}\,\sum_{\und{s},\und{t}\in\bbZ_N^n}\, q^{\und{s}\Theta\und{t}}~\und{s}\otimes\und{t}
\,:=\, \frac{1}{N^n}\,\sum_{\und{s},\und{t}\in\bbZ_N^n}\, q^{\sum_{i,j=1}^n\,\Theta^{ij}\,s_i\,t_j}~\und{s}\otimes\und{t}\ \in\ 
\bbC[\bbZ_N^n]\otimes\bbC[\bbZ_N^n]
\end{flalign}
defines a quasi-triangular structure for the group Hopf algebra $\bbC[\bbZ_N^n]$.
The relevant properties from Definition \ref{def:Rmatrix} can be easily checked by
using the standard identity
\begin{flalign}\label{eqn:Deltanormalization}
\sum_{\und{t}\in\bbZ_N^n}\, q^{\und{s}\Theta\und{t}}\,=\, N^n\, \delta_{\und{s},\und{0}}\, =\, 
\sum_{\und{t}\in\bbZ_N^n}\, q^{\und{t}\Theta\und{s}} \quad,
\end{flalign}
where $\delta_{\und{s},\und{0}}$ denotes the Kronecker delta-symbol
and $\und{0} := (0,\dots,0)\in \bbZ_N^n$ is the identity element.
In the case where the matrix $\Theta$ is antisymmetric,
$R$ is further a triangular structure for $\bbC[\bbZ_N^n]$.
\end{ex}

\subsection{\label{subsec:braidedBV}Finite-dimensional braided BV formalism}
Let $H$ be a Hopf algebra with triangular structure $R$.
We saw in Section~\ref{subsec:Hopf} that the category
${}_H\Mod$ of left $H$-modules is a closed symmetric monoidal category.
Because ${}_H\Mod$ is an Abelian category,
we can study cochain complexes in ${}_H\Mod$ via 
 standard techniques from homological algebra.
Let us denote by
\begin{flalign}
{}_H\Ch\,:=\, \Ch\big({}_H\Mod\big)
\end{flalign}
the category of cochain complexes of left $H$-modules.
An object in this category
is a $\bbZ$-graded left $H$-module $V$ together with
an $H$-equivariant differential $\dd$ of degree $+1$, 
i.e.~$\dd(h\ra v) = h\ra (\dd v)$, for all $h\in H$ and $v\in V$.
The morphisms are the $H$-equivariant cochain maps.
Analogously to the category $\Ch_\bbK$ that we discussed in
Section~\ref{subsec:cochain}, ${}_H\Ch$ is a closed symmetric monoidal category.
The monoidal product is given by endowing
\eqref{eqn:tensorproductChK} with the left tensor product
$H$-action \eqref{eqn:tensorleftaction}, the monoidal
unit is $\bbK$ endowed with the trivial left $H$-action
and the internal hom is given by endowing \eqref{eqn:mappingChK}
with the left adjoint $H$-action \eqref{eqn:homleftaction}.
The symmetric braiding is given by combining \eqref{eqn:braidingChK}
with \eqref{eqn:tauR}. Explicitly,
\begin{flalign}\label{eqn:braidingHCh}
\tau_R\,:\, V\otimes W~\longrightarrow~ W\otimes V
~,~~v\otimes w~\longmapsto~ (-1)^{\vert v\vert\,\vert w\vert}\,(R_{\alpha}\ra w)\otimes (R^\alpha\ra v)\quad,
\end{flalign}
for all $V,W\in{}_H\Ch$, involves both the Koszul signs 
and the $R$-matrix $R= R^\alpha\otimes R_\alpha \in H\otimes H$.
\sk

The generalization of Definition \ref{def:linearBV} to the present
case then reads as follows.
\begin{defi}\label{def:HeqvlinearBV}
A {\em free braided BV theory} is an object $E=(E,-Q)\in{}_H\Ch$
together with an ${}_H\Ch$-morphism $\langle\,\cdot\,,\,\cdot\,\rangle : E\otimes E\to\bbK[-1]$ 
that is non-degenerate and antisymmetric, 
i.e.\ $\langle\,\cdot\,, \,\cdot\,\rangle =- \langle\,\cdot\,,\,\cdot\,\rangle\circ \tau_R $,
or explicitly 
\begin{flalign}
\langle \varphi,\psi\rangle \,=\,-(-1)^{\vert \varphi\vert\,\vert\psi\vert}~\langle R_\alpha\ra \psi,R^\alpha\ra \varphi\rangle\quad,
\end{flalign}
for all $\varphi,\psi\in E$.
\end{defi}

Before we can generalize Definition \ref{def:P0algebraTrad},
we have to introduce a concept of \emph{braided commutative dg-algebra}.
Because ${}_H\Ch$ is a (closed) symmetric monoidal category,
there is an associated category $\CAlg({}_H\Ch)$ of commutative
algebras in ${}_H\Ch$. An object in this category
is a triple $(A,\mu,\eta)$ consisting of an object 
$A\in{}_H\Ch$ together with two ${}_H\Ch$-morphisms
$\mu: A\otimes A\to A$ and $\eta: \bbK\to A$ satisfying the associativity
and unitality axioms. The $H$-equivariance 
of the product $\mu$ and unit $\eta$ means explicitly that
\begin{flalign}
h\ra(a\,b) \,=\, (h_{\und{1}}\ra a)~(h_{\und{2}}\ra b)\quad,\qquad
h\ra \oone \,=\, \epsilon(h)\,\oone\quad,
\end{flalign}
for all $h\in H$ and $a,b\in A$.
Commutativity in this case means
that $\mu\circ\tau_R = \mu$, or explicitly
\begin{flalign}
a\,b \,=\, (-1)^{\vert a\vert\,\vert b\vert}~(R_\alpha\ra b)\,(R^\alpha\ra a)\quad,
\end{flalign}
for all $a,b\in A$. The main example for us is the 
{\em braided symmetric algebra} $\Sym_R V\in\CAlg({}_H\Ch)$
associated with an object $V\in{}_H\Ch$. In analogy to the usual
case, this $H$-equivariant dg-algebra is generated by
all $v\in V$, modulo the commutation relations involving the $R$-matrix
\begin{flalign}\label{eqn:braidedsymmetricalgebra}
v\,v^\prime \,=\, (-1)^{\vert v\vert\,\vert v^\prime\vert}~(R_\alpha\ra v^\prime)\,(R^\alpha\ra v)\quad,
\end{flalign}
for all $v,v^\prime\in V$. 
\sk

We are now ready to
generalize Definition \ref{def:P0algebraTrad}, which we shall write
in the more explicit format of Remark \ref{rem:P0algebraTrad}.
\begin{defi}\label{def:HeqvP0algebra}
A {\em braided $P_0$-algebra} is a braided
commutative dg-algebra $A\in  \CAlg({}_H\Ch)$ together with
an ${}_H\Ch$-morphism $\{\,\cdot\, , \,\cdot\,\} : A\otimes A\to A[1]$
satisfying the following axioms:
\begin{itemize}
\item[(i)] Compatibility with the differential: For all $a,b\in A$,
\begin{flalign}
-\dd\{a,b\}\,=\,\{\dd a, b\} + (-1)^{\vert a\vert } \,\{a,\dd b\}\quad.
\end{flalign}

\item[(ii)] Braided symmetry: For all $a,b\in A$,
\begin{flalign}
\{ a, b \} \, =\,  (-1)^{\vert a\vert\,\vert b\vert}\, \{R_\alpha\ra b,R^\alpha\ra a\}\quad.
\end{flalign} 

\item[(iii)] Braided Jacobi identity: For all $a,b,c\in A$,
\begin{flalign}
\nn 0\,&=\,(-1)^{\vert a\vert \,\vert c\vert + \vert b\vert + \vert c\vert }~\{a,\{b,c\}\} \\
\nn & \quad \ +(-1)^{\vert b\vert \,\vert a\vert  + \vert c\vert + \vert a\vert }~ \{R_\alpha \ra b,\{R_\beta\ra c,R^\beta\, R^\alpha\ra a\}\} \\
& \quad \ +(-1)^{\vert c\vert \,\vert b\vert + \vert a\vert + \vert b\vert }~ \{R_\beta\, R_{\alpha}\ra c,\{R^\beta\ra a,R^\alpha\ra b\}\}\quad.
\end{flalign}

\item[(iv)] Braided derivation property: For all $a,b,c\in A$,
\begin{flalign}
\{a,b\,c\} \,=\, \{a,b\} \,c + (-1)^{\vert b\vert \,(\vert a\vert +1)}\, (R_\alpha\ra b)~\{R^\alpha\ra a,c\}\quad.
\end{flalign}
\end{itemize}
\end{defi}

To a free braided BV theory 
$(E,-Q,\langle\,\cdot\, , \,\cdot\,\rangle)$ one can assign the 
braided $P_0$-algebra
\begin{flalign}\label{eqn:HeqvObscl}
\Obs^\cl \,:=\, \big(\Sym_R E[1],Q,\{\,\cdot\, , \,\cdot\,\}\big)
\end{flalign}
consisting of the braided symmetric algebra of the dual $E^\ast \cong E[1]$
together with the bracket defined by
\begin{flalign}
\{\varphi,\psi\}\,=\,(\varphi,\psi) \,\oone \quad,
\end{flalign}
for all $\varphi,\psi\in E[1]$, and the properties (ii) and (iv) of Definition \ref{def:HeqvP0algebra}.
The shifted pairing $(\,\cdot\, ,\,\cdot\,) : E[1]\otimes E[1]\to\bbK[1]$ is defined
by $\langle\,\cdot\, ,\cdot\, \rangle$ in complete analogy to \eqref{eqn:dualpairing}.
\sk

As in Section \ref{subsec:BV}, interactions and quantization
are both described by certain types of deformations of the differential
$Q$ in \eqref{eqn:HeqvObscl}. In order to obtain a deformed
cochain complex of left $H$-modules, i.e.\ an object in ${}_H\Ch$,
one should consider {\em $H$-invariant} deformations. 
Let us discuss the different types of deformations in detail.
\sk

To obtain an interacting braided classical BV theory,
we pick a $0$-cochain $I\in (\Sym_R E[1])^0$ that is $H$-invariant,
i.e.\ $h\ra I =\epsilon(h)\,I$ for all $h\in H$, and satisfies
the classical master equation
\begin{flalign}\label{eqn:masterbraidedcl}
Q(\lambda\,I) + \tfrac{1}{2}\,\{\lambda\,I,\lambda\,I\}\,=\,0\quad,
\end{flalign}
where $\lambda$ is an $H$-invariant formal parameter (coupling constant). Then
\begin{flalign}\label{eqn:Obsclintbraided}
\Obs^{\cl,\mathrm{int}}\,:=\, \big(\Sym_R E[1], Q^{\mathrm{int}},\{\,\cdot\, ,\,\cdot\,\}\big)\quad\text{with}\quad
Q^{\mathrm{int}}\,:=\, Q + \{\lambda\,I,\,\cdot\,\}
\end{flalign}
defines a braided $P_0$-algebra that is interpreted as the classical
observables for the interacting braided BV theory with interaction term $I$.
\sk

For quantization, we note that the definition of
the BV Laplacian in \eqref{eqn:BVLaplacian} applies
to our braided case as well and thereby
defines an ${}_H\Ch$-morphism $\Delta_{\BV}: \Sym_R E[1]\to (\Sym_R E[1])[1]$
that squares to $0$. However,
the explicit formula \eqref{eqn:BVexplicit} for the 
ordinary BV Laplacian is modified in the braided
case by suitable actions of the $R$-matrix.
Explicitly, one finds that
\begin{flalign}
\label{eqn:braidedBVLaplacian}\Delta_{\BV}\big(\varphi_1\cdots \varphi_n \big) \,=\, \sum_{i<j}\, & (-1)^{\sum_{k=1}^{i-1}\,\vert \varphi_k\vert + \vert\varphi_j\vert\,  \sum_{k=i+1}^{j-1}\, \vert \varphi_k\vert}~\big( \varphi_i,R_{\alpha_{i+1}}\cdots R_{\alpha_{j-1}}\ra \varphi_j\big)~\\ \nonumber
& \times\,\varphi_1\cdots \varphi_{i-1}\,\widehat{\varphi}_i\,(R^{\alpha_{i+1}}\ra \varphi_{i+1})\cdots 
(R^{\alpha_{j-1}}\ra\varphi_{j-1})\,\widehat{\varphi}_j\,\varphi_{j+1}\cdots\varphi_n\quad,
\end{flalign}
for all $\varphi_1,\dots,\varphi_n\in E[1]$ with $n\geq 2$. Then
\begin{flalign}
\Obs^{\hbar}\,:=\,\big(\Sym_R E[1],Q^{\hbar}\big)\quad\text{with}\quad
Q^{\hbar}\,:=\, Q + \hbar\,\Delta_{\BV}
\end{flalign}
defines a braided $E_0$-algebra that is interpreted as
the quantum observables for the non-interacting braided BV theory.
\sk

Finally, to obtain an interacting braided quantum BV theory,
we consider a $0$-cochain $I\in (\Sym_R E[1])^0$ that is $H$-invariant and satisfies
the quantum master equation
\begin{flalign}\label{eqn:masterbraidedqu}
Q(\lambda\,I) + \hbar\,\Delta_{\BV}(\lambda\,I) + \tfrac{1}{2}\,\{\lambda\,I,\lambda\,I\}\,=\,0\quad.
\end{flalign}
Then
\begin{flalign}
\Obs^{\hbar,\mathrm{int}}\,:=\,\big(\Sym_R E[1],Q^{\hbar,\mathrm{int}}\big)\quad\text{with}\quad
Q^{\hbar,\mathrm{int}}\,:=\, Q + \hbar\,\Delta_{\BV} + \{\lambda\,I,\,\cdot\,\}
\end{flalign}
defines a braided $E_0$-algebra that is interpreted as
the quantum observables for the interacting braided BV theory with interaction term $I$.

\subsection{\label{subsec:braidedLinfty}Braided $L_\infty$-algebras and their cyclic versions}
The construction, reviewed in Section~\ref{subsec:Linfty}, 
of interaction terms satisfying the classical (and also the quantum) 
master equation from cyclic $L_\infty$-algebra structures generalizes
to the case of braided BV theories by using the concept of a braided $L_\infty$-algebra
introduced in~\cite{DCGRS20,DCGRS21}. 
\begin{defi}
A {\em braided $L_\infty$-algebra} is a $\bbZ$-graded left $H$-module $L$
together with a collection $\{ \ell_n : L^{\otimes n}\to L\}_{n\in \bbZ_{\geq 1}}$
of $H$-equivariant graded braided antisymmetric linear maps of degree $\vert \ell_n\vert = 2-n$
that satisfy the braided homotopy Jacobi identities
\begin{flalign}\label{eqn:braidedhomotopyJacobi}
\sum_{k=0}^{n-1}\, (-1)^{k}~
\ell_{k+1}\circ \big( \ell_{n-k}\otimes \id_{L^{\otimes k}}\big)\,\circ\, \sum_{\sigma\in\mathrm{Sh}(n-k;k)}\,
\mathrm{sgn}(\sigma) \,\tau_{R}^\sigma\,=\,0\quad,
\end{flalign}
for all $n\geq 1$, where $\tau_{R}^\sigma : L^{\otimes n}\to L^{\otimes n}$
denotes the action of the permutation $\sigma$ via the symmetric 
braiding $\tau_R$ on the category of graded left $H$-modules.
\end{defi}
\begin{rem}
Let us spell this out in a bit more detail. The graded braided antisymmetry property
of $\ell_n : L^{\otimes n}\to L$ means that
\begin{flalign}
\ell_n\big(v_1,\dots,v_n\big) \,=\, -(-1)^{\vert v_i\vert \,\vert v_{i+1}\vert}~\ell_n\big(v_1,\dots,v_{i-1},R_\alpha\ra v_{i+1},R^\alpha\ra v_i, v_{i+2},\dots, v_n\big)\quad,
\end{flalign}
for all $i=1,\dots, n-1$ and all homogeneous elements $v_1,\dots,v_n\in L$.
The permutation action $\tau_{R}^\sigma : L^{\otimes n}\to L^{\otimes n}$ 
in \eqref{eqn:braidedhomotopyJacobi} includes, in addition to the usual Koszul signs,
appropriate actions of the $R$-matrix as in \eqref{eqn:braidingHCh}. Similarly to 
Remark~\ref{rem:Linfty}, every braided $L_\infty$-algebra has an underlying cochain 
complex $(L,\dd_L:=\ell_1)\in {}_H\Ch$, and $\ell_2:L\otimes L\to L$ is an ${}_H\Ch$-morphism. 
When the only non-vanishing bracket is $\ell_2$, a braided $L_\infty$-algebra 
is an example of a braided Lie algebra in the sense of~\cite{Maj94}.
\end{rem}
\begin{defi}
A {\em cyclic braided $L_\infty$-algebra} is a braided $L_\infty$-algebra $(L,\{\ell_n\})$ 
together with a non-degenerate braided symmetric ${}_H\Ch$-morphism
$\cyc{\,\cdot\,}{\,\cdot\,} : L\otimes L \to\bbK[-3]$ that satisfies
the cyclicity condition
\begin{flalign}
\cyc{v_0}{\ell_n(v_1,\dots,v_n)}\,=\,
(-1)^{n\,(\vert v_0\vert + 1) }
\, \cyc{R_{\alpha_0}\cdots R_{\alpha_{n-1}} \ra v_n}{\ell_n(R^{\alpha_0}\ra v_0,
\dots,R^{\alpha_{n-1}}\ra v_{n-1})}\quad,
\end{flalign}
for all $n\geq 1$ and all homogeneous elements $v_0,v_1,\dots,v_n\in L$.
\end{defi}

Similarly to the ordinary case from Section \ref{subsec:Linfty}, 
every free braided BV theory $(E,-Q,\langle\,\cdot\, , \,\cdot\,\rangle)$
as in Definition \ref{def:HeqvlinearBV} defines an
Abelian cyclic braided  $L_\infty$-algebra given by $E[-1]$, $\ell_1 := \dd_{E[-1]} = Q$
and cyclic structure
\begin{flalign}
\xymatrix@C3.5em{
\cyc{\,\cdot\,}{\,\cdot\,}\,:\,E[-1]\otimes E[-1]\,\cong\,(E\otimes E)[-2]\ar[r]^-{\langle\,\cdot\,,\,\cdot\,\rangle[-2]}~&~\bbK[-1][-2]\,\cong\,\bbK[-3]\quad.
}
\end{flalign}
Introducing an interaction term $I\in(\Sym_R E[1])^0$ that satisfies the classical
master equation \eqref{eqn:masterbraidedcl} is then equivalent
to endowing the Abelian cyclic braided $L_\infty$-algebra $(E[-1],\ell_1,\cyc{\,\cdot\,}{\,\cdot\,})$
with compatible higher brackets $\{\ell_n\}_{n\geq 2}$.
This relationship is given explicitly by
\begin{flalign}\label{eqn:braidedMCaction}
\lambda \, I \,=\, \sum_{n\geq 2}\,\frac{\lambda^{n-1}}{(n+1)!}\, \cyc{\mathsf{a}}{\ell_n^{\mathrm{ext}}(\mathsf{a},\dots,\mathsf{a})}_{\mathrm{ext}}\ \in\ (\Sym_RE[1])^0\quad,
\end{flalign}
where the contracted coordinate functions
\begin{flalign}\label{eqn:braidedCCF}
\mathsf{a}\,:=\, \sum_\alpha\,\varrho^\alpha\otimes\varepsilon_{\alpha}\ \in\ \big((\Sym_R E[1])\otimes E[-1]\big)^1
\end{flalign}
are defined by choosing a basis $\{\varepsilon_\alpha\in E[-1]\}$
with dual basis $\{\varrho^\alpha\in E[-1]^\ast \cong E[2]\}$.
We stress that the extended brackets
\begin{flalign}
\ell_n^{\mathrm{ext}} \,:\,
\big((\Sym_R E[1])\otimes E[-1]\big)^{\otimes n} ~\longrightarrow~ (\Sym_R E[1])\otimes E[-1]
\end{flalign}
and the extended pairing
\begin{flalign}
\cyc{\,\cdot\,}{\,\cdot\,}_{\mathrm{ext}} \,:\,
\big((\Sym_R E[1])\otimes E[-1]\big)\otimes \big((\Sym_R E[1])\otimes E[-1]\big) ~\longrightarrow~ (\Sym_R E[1])[-3]
\end{flalign}
receive, in addition to the obvious Koszul signs (cf.\ \cite[Section 2.3]{JRSW19}),
also appropriate actions of the $R$-matrix. For the pairing we have
\begin{flalign}
\cyc{a\otimes v}{a^\prime\otimes v^\prime}_{\mathrm{ext}}\,=\,
(-1)^{\vert a\vert +\vert a^\prime\vert + \vert v\vert\, \vert a^\prime\vert }~a\,(R_\alpha\ra a^\prime)~\cyc{R^\alpha\ra v}{v^{\prime}} \quad,
\end{flalign}
for all homogeneous $a,a^\prime\in \Sym_R E[1] $ and $v,v^\prime\in E[-1]$,
and similarly for the brackets $\ell_n^{\mathrm{ext}}$.
\begin{rem}
Our claim that \eqref{eqn:braidedMCaction} satisfies
the classical master equation \eqref{eqn:masterbraidedcl}
can be proven by precisely the same calculation as in 
the ordinary case, see e.g.\ \cite[Section~4.3]{JRSW19}. This is
due to the fact that the contracted coordinate functions \eqref{eqn:braidedCCF}
are $H$-invariant elements of $(\Sym_R E[1])\otimes E[-1]$, 
which implies that all appearances of $R$-matrices in the properties
of the extended brackets $\ell_n^{\mathrm{ext}}$ and the extended pairing 
$\cyc{\,\cdot\,}{\,\cdot\,}_{\mathrm{ext}}$ disappear when they are
evaluated on tensor products of the $H$-invariant element
$\mathsf{a}$. By the same argument,
one can use the proofs from the ordinary case~\cite[Section~4.3]{JRSW19} 
to show that \eqref{eqn:braidedMCaction} is
annihilated by the BV Laplacian, i.e.\ $\Delta_{\BV}(\lambda \, I)=0$,
and consequently that it also satisfies the quantum master equation \eqref{eqn:masterbraidedqu}.
\end{rem}

\subsection{Braided homological perturbation theory}
The usual homological perturbation lemma (cf.\ Theorem \ref{theo:HPL})
extends to our braided setting, provided that we use small perturbations
$\delta\in \hom(V,V)^{1}$ that are additionally $H$-invariant.
(Recall that the perturbations corresponding to interactions and quantization
from Section~\ref{subsec:braidedBV} are $H$-invariant.)
Concretely, the details are spelled out as follows.
\begin{defi}\label{def:braidedDefRet}
A {\em braided strong deformation retract} of an object 
$V\in{}_H\Ch$ onto its cohomology $H^\bullet(V)$
is a strong deformation retract
\begin{equation}\label{eqn:Heqvdefretpicture}
\begin{tikzcd}
(H^\bullet(V),0) \ar[r,shift right=1ex,swap,"\iota"] & \ar[l,shift right=1ex,swap,"\pi"] (V,\dd)\ar[loop,out=25,in=-25,distance=28,"\gamma"]
\end{tikzcd}\quad,
\end{equation}
in the sense of Definition \ref{def:defret}, such that 
$\pi$ and $\iota$ are $H$-equivariant, i.e.\ morphisms in ${}_H\Ch$,
and the homotopy $\gamma\in\hom(V,V)^{-1}$ is $H$-invariant, i.e.\
$h\ra\gamma = \epsilon(h)\,\gamma$ for all $h\in H$.
\end{defi}
\begin{cor}\label{cor:HeqvHPL}
Consider a braided strong deformation retract as in \eqref{eqn:Heqvdefretpicture}. 
Let $\delta \in\hom(V,V)^1$ be a small $H$-invariant perturbation, i.e.\ $h\ra \delta=\epsilon(h)\,\delta$
for all $h\in H$.  Then the expressions \eqref{eqn:HPL} define a braided strong deformation retract.
\end{cor}
\begin{proof}
By direct inspection, one observes that the explicit formulas
in \eqref{eqn:HPL} satisfy the necessary $H$-equivariance or $H$-invariance properties.
\end{proof}

Given any free braided BV theory $(E,-Q,\langle\,\cdot\,, \,\cdot\,\rangle)$  
in the sense of Definition \ref{def:HeqvlinearBV} and any
braided strong deformation retract for its dual complex 
\begin{equation}
\begin{tikzcd}
(H^\bullet(E[1]),0) \ar[r,shift right=1ex,swap,"\iota"] & \ar[l,shift right=1ex,swap,"\pi"] (E[1],Q)\ar[loop,out=25,in=-25,distance=28,"\gamma"]
\end{tikzcd}\quad,
\end{equation}
a similar construction as in Section \ref{subsec:HPL} defines
a braided strong deformation retract
\begin{equation}\label{eqn:HeqvdefretSym}
\begin{tikzcd}
\big(\Sym_R\,H^\bullet(E[1]),0\big) \ar[r,shift right=1ex,swap,"\Sym_R\,\iota"] & \ar[l,shift right=1ex,swap,"\Sym_R\,\pi"] \big(\Sym_R\,E[1],Q\big) \qquad  \ar[loop,out=25,in=-25,distance=40,"\Sym_R\,\gamma"]
\end{tikzcd}\quad.
\end{equation}
The correlation functions of braided
non-interacting and also interacting quantum BV theories
can then be determined by applying Corollary \ref{cor:HeqvHPL}
to \eqref{eqn:HeqvdefretSym} and the deformed differentials
from Section~\ref{subsec:braidedBV}.
This construction works in a completely analogous way to the ordinary case
that we reviewed at the end of Section \ref{subsec:HPL}.


\section{\label{sec:fuzzytorus}Braided field theories on the fuzzy torus}
We illustrate the formalism of Section \ref{sec:braidedBV}
in the example of scalar field theories on the fuzzy $2$-torus, reproducing 
from our perspective many facets of Oeckl's braided quantum field theory 
for symmetric braidings~\cite{Oec00}. 
In this section we work over the field $\bbK=\bbC$
of complex numbers.

\paragraph{The fuzzy 2-torus.} To fix our notation and conventions,
let us recall the definition of the fuzzy $2$-torus and its triangular Hopf
algebra symmetry, see e.g. \cite{BG19,BSS17} for further details.
Let us fix a positive integer $N\in\bbZ_{>0}$ and set 
\begin{flalign}\label{eqn:qrootofunity}
q\,:=\e^{2\pi \ii/N}\ \in\ \bbC\quad.
\end{flalign}
The algebra of functions on the fuzzy torus $\bbT^2_N$ is defined
as the noncommutative $\ast$-algebra
\begin{flalign}\label{eqn:fuzzytorusalgebra}
A\,:=\, \bbC[U,V]\,\big/\, \big(U\,U^\ast -\oone\,,~V\,V^\ast-\oone\,,~U\,V- q\, V\,U\,,~U^N-\oone\,,~V^N-\oone\big)
\end{flalign}
generated by two elements $U$ and $V$, modulo the $\ast$-ideal generated by the displayed relations.
One should think of the generators $U$ and $V$ as the two exponential functions corresponding
to the two $1$-cycles of $\bbT^2_N$. Every element 
of $A$ may be written uniquely as $a = \sum_{i,j\in\bbZ_N}\, a_{ij}\, U^i\,V^j$,
where $a_{ij}\in\bbC$ play the role of Fourier coefficients.\footnote{One can 
identify $A\cong{\rm Mat}_N(\bbC)$ by realizing the elements $U$ and $V$ as 
$N\times N$ clock and shift matrices, see e.g.~\cite{LLS01}, but this concrete 
realization is not necessary in our treatment.}
\sk

The fuzzy $2$-torus has a (discrete) translation symmetry
that is given by a left action $\ra : H\otimes A\to A$ 
of the group Hopf algebra $H := \bbC[\bbZ_N^2]$ introduced in Example
\ref{ex:HopfAlgebra}. Explicitly, the basis vectors $\und{k} = (k_1,k_2)\in H$,
with $k_1,k_2\in\bbZ_N$ integers modulo $N$, act on the generators of $A$ as
\begin{flalign}\label{eqn:fuzzytorusHaction}
\und{k}\ra U\,:=\,q^{k_1}\, U\quad,\qquad \und{k}\ra V\,:=\, q^{k_2}\,V\quad.
\end{flalign}
This action is extended to all of $A$ by demanding that $A$ is a left $H$-module algebra,
i.e.\ $\und{k}\ra (a\,b) = (\und{k}_{\und{1}}\ra a)\,(\und{k}_{\und{2}}\ra b) = 
(\und{k}\ra a)\,(\und{k}\ra b)$, for all $a,b\in A$, where we have used the coproduct
of $H = \bbC[\bbZ_N^2]$ from Example \ref{ex:HopfAlgebra}.
\sk

We endow $H$ with the triangular structure (cf.\ Example \ref{ex:HopfAlgebra})
\begin{flalign}\label{eqn:fuzzytorusRmatrix}
R\,:=\,\frac{1}{N^2}\, \sum_{\und{s},\und{t}\in\bbZ_N^2}\, q^{\und{s}\Theta\und{t}}~\und{s}\otimes\und{t}\ \in\ H\otimes H\quad\quad\text{with}\quad
\Theta\,=\,\bigg(\begin{matrix}
0& -1\\
1& 0
\end{matrix}\bigg)\quad,
\end{flalign}
which is motivated from the fact that with this choice $A$ becomes a {\em braided commutative}
left $H$-module algebra, i.e.\ $a\,b = (R_\alpha\ra b)\,(R^\alpha\ra a)$ 
for all $a,b\in A$.
The latter statement can be easily checked by considering,
without loss of generality, two basis elements $a=U^i\,V^j$ and $b = U^k\,V^l$, for some $i,j,k,l\in\bbZ_N$.
Using the commutation relation in \eqref{eqn:fuzzytorusalgebra}, one computes
\begin{subequations}
\begin{flalign}
(U^i\,V^j)\,(U^k\, V^l) \,=\,  q^{i\,l-j\,k} ~ (U^k\, V^l)\, (U^i\,V^j)\quad.
\end{flalign}
Using now the definition of the $R$-matrix \eqref{eqn:fuzzytorusRmatrix} and 
of the left action \eqref{eqn:fuzzytorusHaction}, one computes
\begin{flalign}
\nn \big(R_\alpha\ra (U^k\, V^l)\big) \, \big(R^\alpha\ra(U^i\,V^j) \big)
\,&=\,\frac{1}{N^2}\,\sum_{\und{s},\und{t}\in\bbZ_N^2}\,q^{\und{s}\Theta\und{t}}~
\big(\und{t}\ra (U^k\, V^l)\big) \,\big(\und{s}\ra (U^i\,V^j)\big)\\[4pt]
\nn \,&=\, \frac{1}{N^2}\,\sum_{\und{s},\und{t}\in\bbZ_N^2}\,q^{\und{s}\Theta\und{t}}~
q^{t_1\, k + t_2\, l}~q^{s_1\, i + s_2\, j} ~(U^k\, V^l)\, (U^i\,V^j)\\[4pt]
\,&=\,q^{i\,l-j\,k}~(U^k\, V^l)\, (U^i\,V^j)\quad,
\end{flalign}
\end{subequations}
where the last step follows from \eqref{eqn:Deltanormalization}.
The two expressions coincide, showing that $A$ is braided commutative.
\sk

Integration on the fuzzy torus is defined through the linear map
\begin{flalign}\label{eqn:intmaptorus}
\int \,:\, A~\longrightarrow~\bbC~,~~a\,=\,\sum_{i,j\in\bbZ_N}\, a_{ij}\,U^i\,V^j~\longmapsto~ \int\, a \,:=\, 
a_{00}^{}\quad,
\end{flalign}
which is easily checked to be $H$-equivariant 
and cyclic, i.e.\ $\int\, a\,b = \int\, b\,a$ for all $a,b\in A$.
The scalar Laplacian reads as
\begin{flalign}
\Delta\,:\, A~\longrightarrow~A ~,~~a~\longmapsto ~\Delta(a)\,:=\,
-\frac{1}{\big(q^{1/2} - q^{-1/2}\big)^2}~
\Big(\big[U,\big[U^\ast,a\big]\big] + \big[V,\big[V^\ast,a\big]\big]\Big)\quad,
\end{flalign}
where we have chosen the square root $q^{1/2} := \e^{\pi\ii/N}\in\bbC$ of $q$.
The scalar Laplacian is also $H$-equivariant under the action \eqref{eqn:fuzzytorusHaction}
because the powers of $q$ resulting from the action on $U$ and on $U^\ast$
compensate each other, and similarly for those from $V$ and $V^\ast$.
A basis of eigenfunctions of the scalar Laplacian is given by
\begin{flalign}\label{eqn:planewavebasis}
e_{\und{k}}\,:=\, U^{k_1}\,V^{k_2}\ \in\ A ~~,
\end{flalign}
for all $\und{k}=(k_1,k_2)\in \bbZ_{N}^{2}$, and the corresponding
eigenvalues are given by
\begin{subequations}\label{eqn:torusLaplacian}
\begin{flalign}
\Delta(e_{\und{k}})\,=\, \big([k_1]^2_q + [k_2]^2_q\big)\, e_{\und{k}}\quad,
\end{flalign}
where the $q$-numbers are defined as
\begin{flalign}
[n]_q \,:=\, \frac{q^{n/2} - q^{-n/2}}{q^{1/2}-q^{-1/2}}\quad.
\end{flalign}
\end{subequations}
For later use, let us record the properties 
\begin{subequations}\label{eqn:planewavebasis2}
\begin{flalign}
e_{\und{k}}^\ast \,=\,q^{-k_1\, k_2}~e_{-\und{k}} ~~,\quad
e_{\und{k}}\,e_{\und{l}} \,=\, q^{-l_1\, k_2}~e_{\und{k}+\und{l}}~~,\quad
\int\, e_{\und{k}}^\ast\, e_{\und{l}} \,=\,\delta_{\und{k},\und{l}}\quad,
\end{flalign}
and 
\begin{flalign}
\tau_R(e_{\und{k}}\otimes e_{\und{l}})\,=\, q^{-\und{k}\Theta\und{l}}~e_{\und{l}}\otimes e_{\und{k}}\,=\,
q^{\und{l}\Theta\und{k}}~e_{\und{l}}\otimes e_{\und{k}}\quad,
\end{flalign}
\end{subequations}
for all $\und{k},\und{l}\in \bbZ_{N}^2$, which imply in particular that
the dual of the basis $\{ e_{\und{k}}\} $ under the integration pairing is 
given by $\{e_{\und{k}}^\ast\}$. 

\paragraph{Free braided BV theory.}
We are now ready to describe a non-interacting scalar field theory 
on the fuzzy torus as a free braided BV theory in the sense of 
Definition \ref{def:HeqvlinearBV}.
\begin{defi}\label{def:torusscalar}
The free  braided BV theory associated to a scalar field with mass parameter
$m^2\geq 0$ on the fuzzy torus is given by the ${}_H\Ch$-object
\begin{flalign}
E\, =\, \big(\xymatrix{
\stackrel{(0)}{A} \ar[r]^-{-Q} & \stackrel{(1)}{A}
}\big)\qquad\text{with}\qquad Q\,:=\,\Delta +m^2
\end{flalign}
concentrated in degrees $0$ and $1$, together with the ${}_H\Ch$-pairing
\begin{flalign}\label{eqn:toruspairing}
\langle\,\cdot\,,\,\cdot\,\rangle\,:\,
E\otimes E~\longrightarrow~\bbC[-1]~,~~
\varphi\otimes\psi~\longmapsto~\langle\varphi,\psi\rangle\,:=\, (-1)^{\vert \varphi\vert}\,\int\,
\varphi\,\psi\quad.
\end{flalign}
\end{defi}
\begin{rem}
Note that the pairing is indeed $H$-equivariant because both the product
on $A$ and the integration map \eqref{eqn:intmaptorus} are $H$-equivariant.
For the antisymmetry property of the pairing,
let us first observe that
\begin{flalign}\label{eqn:strictantisymmetric}
\langle \varphi,\psi\rangle
\,=\, (-1)^{\vert \varphi\vert}\,\int\, \varphi\,\psi
\,=\, (-1)^{\vert \varphi\vert}\,\int\,\psi \,\varphi
\,=\, (-1)^{\vert\varphi\vert + \vert\psi\vert}\,\langle \psi,\varphi\rangle
\,=\, -(-1)^{\vert \varphi\vert\,\vert\psi\vert}\,\langle \psi,\varphi\rangle
\end{flalign}
is {\em strictly} antisymmetric, i.e.\ antisymmetric without $R$-matrix actions.
Using the explicit form of the $R$-matrix \eqref{eqn:fuzzytorusRmatrix},
together with the Hopf algebra structure on $H=\bbC[\bbZ_N^2]$ from Example \ref{ex:HopfAlgebra},
one shows that
\begin{flalign}
\nn \langle R_{\alpha}\ra\psi,R^\alpha\ra \varphi\rangle \,&=\,
\frac{1}{N^2}\, \sum_{\und{s},\und{t}\in \bbZ_N^2}\,
q^{\und{s}\Theta\und{t}}~\langle \und{t}\ra \psi, \und{s}\ra \varphi \rangle \\[4pt]
\nn\,&=\, 
\frac{1}{N^2}\, \sum_{\und{s},\und{t}\in \bbZ_N^2}\,
q^{\und{s}\Theta\und{t}}~\langle \psi, (\und{s}-\und{t})\ra \varphi \rangle\\[4pt]
\nn\,&=\,\frac{1}{N^2}\, \sum_{\und{s}^\prime,\und{t}\in \bbZ_N^2}\,
q^{\und{s^\prime}\Theta\und{t}}~\langle \psi, \und{s}^\prime\ra \varphi \rangle\\[4pt]
\,&=\, \sum_{\und{s}^\prime\in\bbZ_N^2}\, \delta_{\und{s}^\prime,\und{0}} \, \langle \psi, \und{s}^\prime\ra \varphi \rangle
\,=\,\langle \psi, \varphi \rangle\quad.
\end{flalign}
In the second step we used $H$-equivariance of the pairing to write
$\langle\,\cdot\, , \,\cdot\,\rangle=
(-\und{t})\ra \langle\,\cdot\, , \,\cdot\,\rangle
=\langle (-\und{t})\ra \,\cdot\, ,\, (-\und{t})\ra \,\cdot\,\rangle$
and in the fourth step we used \eqref{eqn:Deltanormalization}.
Together with \eqref{eqn:strictantisymmetric}, it then follows
that the pairing \eqref{eqn:toruspairing} is {\em both} strictly antisymmetric
and also braided antisymmetric, the latter property being as required by Definition \ref{def:HeqvlinearBV}.
\end{rem}

Following the general approach outlined in Section \ref{subsec:braidedBV}, we can construct
from this input a braided $P_0$-algebra
\begin{flalign}
\Obs^{\cl} \,:=\, \big(\Sym_R E[1],Q,\{\,\cdot\, , \,\cdot\,\}\big)
\end{flalign}
of classical observables of the non-interacting theory. 
Let us recall that $\Sym_R$ denotes the braided symmetric algebra
(defined via \eqref{eqn:braidedsymmetricalgebra}) and note that the complex
\begin{flalign}\label{eqn:torusE1complex}
E[1] \,=\,  \big(\xymatrix{
\stackrel{(-1)}{A} \ar[r]^-{Q} & \stackrel{(0)}{A}
}\big)\ \in\ {}_H\Ch
\end{flalign}
is concentrated in degrees $-1$ and $0$.

\paragraph{Interactions.} Analogously to the scalar field on the fuzzy sphere discussed
in Section \ref{subsec:ScalarSphere}, the braided symmetric algebra
$\Sym_R E[1]$ in the present example is concentrated in non-positive degrees (cf.\ \eqref{eqn:torusE1complex}).
Hence every $H$-invariant $0$-cochain $I\in (\Sym_R E[1])^0$ 
automatically satisfies both the classical and quantum master equations \eqref{eqn:masterbraidedcl} 
and \eqref{eqn:masterbraidedqu} because, for degree reasons, one has
\begin{flalign}
Q(I)\,=\,\{I,I\} \,=\, \Delta_{\BV}(I) \,=\,0\quad.
\end{flalign}

As a concrete example, let us introduce the $m+1$-point interaction
by using the cyclic braided $L_\infty$-algebra formalism from Section~\ref{subsec:braidedLinfty}.
The scalar field from Definition \ref{def:torusscalar} defines
an Abelian cyclic braided $L_\infty$-algebra given by the complex
\begin{flalign}
E[-1]\,=\, \big(\xymatrix{
\stackrel{(1)}{A} \ar[r]^-{Q} & \stackrel{(2)}{A}
}\big)\ \in\ {}_H\Ch
\end{flalign}
and the cyclic structure
\begin{flalign}\label{eqn:toruscyclic}
\cyc{\,\cdot\,}{\,\cdot\,}\,:\, E[-1]\otimes E[-1]~\longrightarrow~\bbC[-3]~,~~
\varphi\otimes\psi~\longmapsto~\cyc{\varphi}{\psi}\,=\, \int\, \varphi\,\psi\quad.
\end{flalign}
Choosing any $m\geq 2$, we can endow this with the compatible $m$-bracket
\begin{flalign}\label{eqn:torusellm}
\ell_m\,:\, E[-1]^{\otimes m}~\longrightarrow~E[-1]~,~~
\varphi_1\otimes\cdots\otimes\varphi_m~\longmapsto~\varphi_1 \cdots\varphi_m
\end{flalign}
given by the multiplication in the left $H$-module algebra $A$.
Note that, for degree reasons, this is only non-vanishing if each $\varphi_i\in E[-1]$
is of degree $1$ in $E[-1]$ and that the braided graded antisymmetry property
of $\ell_m$ is a consequence of the fact that $A$ is braided commutative.
Indeed, we find
\begin{flalign}
\nn \ell_m(\varphi_1,\dots, \varphi_m)\,&=\, 
\varphi_1 \cdots\varphi_i\,\varphi_{i+1}\cdots \varphi_m \\[4pt]
\nn \,&=\, \varphi_1\cdots (R_\alpha\ra \varphi_{i+1})\,(R^\alpha\ra \varphi_i)\cdots\varphi_m\\[4pt]
\,&=\, -(-1)^{\vert \varphi_i\vert \,\vert\varphi_{i+1}\vert}~
\ell_{m}(\varphi_1,\dots,R_\alpha\ra \varphi_{i+1},R^\alpha\ra \varphi_{i},\dots, \varphi_m)\quad,
\end{flalign}
for all $i=1,\dots,m-1$.
\sk

Using the basis $\{e_{\und{k}}\in A\}_{\und{k}\in\bbZ_{N}^2}$ introduced in \eqref{eqn:planewavebasis}, 
together with its properties listed in \eqref{eqn:planewavebasis2}, we obtain the contracted
coordinate functions
\begin{flalign}\label{eqn:torusCCF}
\mathsf{a}\,=\,\sum_{\und{k}\in\bbZ_{N}^2}\, e_{\und{k}}^\ast\otimes e_{\und{k}}
+ \sum_{\und{k}\in\bbZ_{N}^2}\,\widetilde{e}_{\und{k}}^{\,\ast} \otimes \widetilde{e}_{\und{k}}
\ \in\ \big((\Sym_R E[1])\otimes E[-1] \big)^1\quad,
\end{flalign}
where, as in the fuzzy sphere example from Section~\ref{subsec:ScalarSphere}, the individual elements live
in the vector spaces  $e_{\und{k}}\in E[-1]^1$,
$e_{\und{k}}^\ast\in E[1]^0$, $\widetilde{e}_{\und{k}}\in E[-1]^2$
and $\widetilde{e}_{\und{k}}^{\,\ast}\in E[1]^{-1}$.
The interaction term \eqref{eqn:braidedMCaction} corresponding to an $m+1$-point
interaction then reads concretely as
\begin{subequations}\label{eqn:torusinteraction}
\begin{flalign}
\lambda\,I\,&=\, \frac{\lambda^{m-1}}{(m+1)!}\,\cyc{\mathsf{a}}{\ell_m^{\mathrm{ext}}(\mathsf{a},\dots,\mathsf{a})}_{\mathrm{ext}}\\[4pt]
\nn \,&=\,\frac{\lambda^{m-1}}{(m+1)!}\,\sum_{\und{k}_0,\dots,\und{k}_m\in\bbZ_{N}^2}\,
q^{\sum_{i<j}\, \und{k}_i\Theta\und{k}_j}~
e_{\und{k}_0}^\ast\,\cdots\,e_{\und{k}_m}^\ast~\cyc{e_{\und{k}_0}}{\ell_{m}(e_{\und{k}_1},\dots,e_{\und{k}_m})}\ \in\ (\Sym_{R}E[1])^0\quad,
\end{flalign}
where the factors of $q$ arise from the braiding identity in \eqref{eqn:planewavebasis2}. We stress again that the products of the elements $e_{\und{k}_i}^\ast$ in the second
line are taken in the braided symmetric algebra $\Sym_{R}E[1]$ and {\em not} in $A$.
\sk

The constants
\begin{flalign}
I_{\und{k}_0\und{k}_1\cdots\und{k}_m}\,:=\,q^{\sum_{i<j}\, \und{k}_i\Theta\und{k}_j}~
\cyc{e_{\und{k}_0}}{\ell_{m}(e_{\und{k}_1},\dots,e_{\und{k}_m})}\ \in\ \bbC
\end{flalign}
can be worked out explicitly by using \eqref{eqn:toruscyclic}, 
\eqref{eqn:torusellm}, \eqref{eqn:intmaptorus} and \eqref{eqn:planewavebasis2},
from which one finds
\begin{flalign}\label{eqn:torusinteractionexplicit}
I_{\und{k}_0\und{k}_1\cdots\und{k}_m}\,&=\,q^{\sum_{i<j}\, \und{k}_i\Theta\und{k}_j}~
\int\, e_{\und{k}_0}\,e_{\und{k}_1}\cdots e_{\und{k}_m}\\[4pt]
\nn \,&=\, 
q^{\sum_{i<j}\, \und{k}_i\Theta\und{k}_j}~q^{-\sum_{i<j}\, {k_{i}}{}_2 \, {k_{j}}{}_1}~\int\, e_{\und{k}_0 + \und{k}_1+ \cdots + \und{k}_m}
\,=\,q^{-\sum_{i<j}\, {k_i}{}_1\, {k_j}{}_2}~\delta_{\und{k}_0+\und{k}_1+\cdots +\und{k}_m,\und{0}}\quad,
\end{flalign}
\end{subequations}
where the double subscript notation indicates the components $\und{k}_i = ({k_{i}}_1, {k_i}_2)\in\bbZ_N^2$ for $i=0,1,\dots,m$.
The constants $I_{\und{k}_0\und{k}_1\cdots\und{k}_m}$ satisfy the $q$-deformed symmetry property
\begin{subequations}\label{eqn:torusinteractionproperties}
\begin{flalign}
I_{\und{k}_0\und{k}_1\cdots\und{k}_i\und{k}_{i+1}\cdots\und{k}_m}
\,=\,q^{\und{k}_i\Theta\und{k}_{i+1}}~I_{\und{k}_0\und{k}_1\cdots\und{k}_{i+1}\und{k}_{i}\cdots\und{k}_m}
\end{flalign}
for any exchange of neighboring indices, which in particular implies
the {\em strict} cyclicity property
\begin{flalign}
I_{\und{k}_0\und{k}_1\cdots\und{k}_{m}}\,=\, I_{\und{k}_1\cdots \und{k}_m \und{k}_0}
\end{flalign}
\end{subequations}
by further using momentum conservation imposed by the Kronecker delta-symbol
$\delta_{\und{k}_0+\und{k}_1+\cdots +\und{k}_m,\und{0}}$.

\paragraph{Braided strong deformation retract.} In the remainder of this section
we consider only the massive case $m^2>0$. (The massless case is slightly more 
involved, but it can be treated analogously to the fuzzy sphere example in Section \ref{subsec:ScalarSphere}.)
In this case the cohomology of the complex $E[1]$ in \eqref{eqn:torusE1complex}
is trivial, i.e.\ $H^\bullet(E[1])\cong 0$, because the spectrum of the operator
$Q = \Delta + m^2$ is positive (see \eqref{eqn:torusLaplacian} for the spectrum of the Laplacian 
$\Delta$). The strong deformation retract is thus given by
\begin{equation}\label{eqn:torusDefRet}
\begin{tikzcd}
(0,0) \ar[r,shift right=1ex,swap,"\iota=0"] & \ar[l,shift right=1ex,swap,"\pi=0"] (E[1],Q)\ar[loop,out=25,in=-25,distance=28,"\gamma=-G"]
\end{tikzcd}\quad,
\end{equation}
where $G$ is the inverse of $Q=\Delta+m^2$, i.e.\ the Green operator,
and $\gamma = -G$ is defined to act as a degree $-1$ map on $E[1]$. 
Explicitly, the homotopy $\gamma$ acts on the (dual) basis $e_{\und{k}}^\ast \in E[1]^0$ as
\begin{flalign}\label{eqnLtorusGreenoperator}
\gamma(e_{\und{k}}^\ast\,)\,=\, -G(e_{\und{k}}^\ast\,)\,=\, -\frac{1}{[k_1]^2_q + [k_2]^2_q + m^2}~\widetilde{e}_{\und{k}}^{\,\ast}\ \in \ E[1]^{-1}\quad,
\end{flalign}
where we have adopted the same notation for the basis vectors
as in \eqref{eqn:torusCCF}. The strong deformation retract
\eqref{eqn:torusDefRet} clearly satisfies the requisite $H$-equivariance
and $H$-invariance properties to be a braided strong deformation retract in
the sense of Definition \ref{def:braidedDefRet}.

\paragraph{Correlation functions.} We shall now explain in more detail
how correlation functions may be computed and provide some explicit examples.
The braided strong deformation retract \eqref{eqn:torusDefRet} extends to
the braided symmetric algebras. Given any small $H$-invariant
perturbation $\delta$ of the differential $Q$ on $\Sym_R E[1]$, we obtain
via Corollary \ref{cor:HeqvHPL} the deformed braided strong deformation retract
\begin{equation}
\begin{tikzcd}
\big(\Sym_R 0\cong \bbC ,0 \big) \ar[r,shift right=1ex,swap,"\widetilde{\Sym_R\iota}"] & \ar[l,shift right=1ex,swap,"\widetilde{\Sym_R\pi}"] \big(\Sym_R E[1],Q+\delta \big) \qquad \qquad \ \ar[loop,out=25,in=-25,distance=40,"\widetilde{\Sym_R\gamma}"]
\end{tikzcd}\quad,
\end{equation}
where the tilded quantities are computed through the homological perturbation lemma, cf.\
Theorem \ref{theo:HPL}. The correlation functions may be computed by the ${}_H\Ch$-morphism
\begin{flalign}
\widetilde{\Pi}\,:=\,\widetilde{\Sym_R\pi}\,=\,\Pi\,\circ\, \sum_{k=0}^\infty\, (\delta\,\Gamma)^k\quad,
\end{flalign}
where we use again the abbreviations $\Pi := \Sym_{R}\pi$ and $\Gamma := \Sym_R\gamma$.
The relevant perturbations $\delta$ are of the form
\begin{flalign}
\delta\,=\,\hbar\,\Delta_{\BV} +\{\lambda\,I,\,\cdot\,\}\quad,
\end{flalign}
where $\Delta_{\BV}$ is the BV Laplacian \eqref{eqn:braidedBVLaplacian} 
and $\lambda \,I\in (\Sym_R E[1])^0$ denotes the $m+1$-point interaction
term \eqref{eqn:torusinteraction} for some $m\geq 2$.
\sk

In order to evaluate the correlation functions 
$\widetilde{\Pi}(\varphi_1\cdots\varphi_n)$ for test functions
of degree zero, we have to understand how the maps $\Pi$
and $\delta\,\Gamma$ act on elements $\varphi_1\cdots\varphi_n\in \Sym_R E[1]$
with all $\varphi_i\in E[1]^0$ of degree zero.
For $\Pi$ this is very easy and we find
\begin{flalign}
\Pi(\oone)\,=\,1 ~~,\quad \Pi(\varphi_1\cdots\varphi_n)\,=\,0\quad,
\end{flalign}
for all $n\geq 1$. To describe $\delta\,\Gamma = 
\hbar \,\Delta_{\BV}\,\Gamma + \{\lambda\,I,\,\cdot\,\}\,\Gamma$,
it is convenient to consider the two summands individually. 
Using the explicit formula \eqref{eqn:braidedBVLaplacian} 
for the BV Laplacian, one finds for the first term
\begin{multline}\label{eqn:torusDeltaBVhomotopy}
\hbar\,\Delta_{\BV}\,\Gamma\big(\varphi_1\cdots \varphi_n\big)\,=\,
-\frac{2\,\hbar}{n}\,\sum_{i<j}\, \big(\varphi_i,R_{\alpha_{i+1}}\cdots R_{\alpha_{j-1}}\ra G(\varphi_j)\big)~\\
 \ \times~\varphi_1\cdots\widehat{\varphi}_i \,(R^{\alpha_{i+1}}\ra \varphi_{i+1})\cdots(R^{\alpha_{j-1}}\ra \varphi_{j-1})\, \widehat{\varphi}_j\cdots\varphi_n\quad.
\end{multline}
Using now the explicit expression \eqref{eqn:torusinteraction}
for the $m+1$-point interaction term, together with its properties 
\eqref{eqn:torusinteractionproperties}, one finds for the second term
\begin{multline}\label{eqn:torusIhomotopy}
\big\{\lambda\,I,\Gamma(\varphi_1\cdots\varphi_n)\big\} \,=\,\\[4pt]
-\frac{\lambda^{m-1}}{m!\, n}\, \sum_{i=1}^n \ \sum_{\und{k}_0,\dots,\und{k}_m\in\bbZ_{N}^2}\,
\varphi_1\cdots\varphi_{i-1}~
I_{\und{k}_0\und{k}_1\cdots\und{k}_m}\,
\big(e^\ast_{\und{k}_0} , G(\varphi_i)\big)~ e_{\und{k}_1}^\ast\cdots e_{\und{k}_m}^\ast
\,\varphi_{i+1}\cdots \varphi_n\quad.
\end{multline}

Similarly to Section \ref{subsec:ScalarSphere}, the two expressions \eqref{eqn:torusDeltaBVhomotopy} 
and \eqref{eqn:torusIhomotopy} may be visualized graphically.
Depicting again the element $\varphi_1\cdots\varphi_n$ by $n$ vertical lines, 
the map in \eqref{eqn:torusDeltaBVhomotopy} may be depicted as
\begin{flalign}
\hbar\,\Delta_{\BV}\,\Gamma\Big(\,\vcenter{\hbox{\begin{tikzpicture}[scale=0.5]
\draw[thick] (0,0) -- (0,1);
\draw[thick] (0.5,0) -- (0.5,1);
\draw[thick] (1,0) -- (1,1);
\node (dottt) at (1.75,0.5) {$\cdots$};
\draw[thick] (2.5,0) -- (2.5,1);
\draw[thick] (3,0) -- (3,1);
\end{tikzpicture}}}\,\Big)
\,=\, -\frac{2\,\hbar}{n}\,\Big(\,
\vcenter{\hbox{\begin{tikzpicture}[scale=0.5]
\draw[thick] (0,0) -- (0,1);
\draw[thick] (0.5,0) -- (0.5,1);
\draw[thick] (0,1) to[out=90,in=90] (0.5,1);
\draw[thick] (1,0) -- (1,1);
\node (dottt) at (1.75,0.5) {$\cdots$};
\draw[thick] (2.5,0) -- (2.5,1);
\draw[thick] (3,0) -- (3,1);
\end{tikzpicture}}}
~+~
\vcenter{\hbox{\begin{tikzpicture}[scale=0.5]
\draw[thick] (0,0) -- (0,1);
\draw[thick] (0.5,0) -- (0.5,1);
\draw[thick] (1,0) -- (1,1);
\draw[thick] (0,1) to[out=90,in=90] (1,1);
\node (dottt) at (1.75,0.5) {$\cdots$};
\draw[thick] (2.5,0) -- (2.5,1);
\draw[thick] (3,0) -- (3,1);
\end{tikzpicture}}}
~+~\cdots~+~
\vcenter{\hbox{\begin{tikzpicture}[scale=0.5]
\draw[thick] (0,0) -- (0,1);
\draw[thick] (0.5,0) -- (0.5,1);
\draw[thick] (1,0) -- (1,1);
\node (dottt) at (1.75,0.5) {$\cdots$};
\draw[thick] (2.5,0) -- (2.5,1);
\draw[thick] (3,0) -- (3,1);
\draw[thick] (2.5,1) to[out=90,in=90] (3,1);
\end{tikzpicture}}}
\,\Big)\quad,
\end{flalign}
where the cap indicates a contraction of two elements
with respect to $(\,\cdot\, ,G(\cdot))$.
In such pictures it is understood that the right leg
of the contraction is permuted via the symmetric braiding
$\tau_R$ across the intermediate vertical lines, which leads
precisely to the $R$-matrix insertions in \eqref{eqn:torusDeltaBVhomotopy}.
The map in \eqref{eqn:torusIhomotopy} may be depicted as
\begin{flalign}
\Big\{ \lambda\,I\,,\, \Gamma\Big(\,\vcenter{\hbox{\begin{tikzpicture}[scale=0.5]
\draw[thick] (0,0) -- (0,1);
\draw[thick] (0.5,0) -- (0.5,1);
\draw[thick] (1,0) -- (1,1);
\node (dottt) at (1.75,0.5) {$\cdots$};
\draw[thick] (2.5,0) -- (2.5,1);
\draw[thick] (3,0) -- (3,1);
\end{tikzpicture}}}\,\Big)\Big\}
\, =\, -\frac{\lambda^{m-1}}{m!\, n}\,\bigg(
\vcenter{\hbox{\begin{tikzpicture}[scale=0.5]
\draw[thick] (-0.5,0) -- (-0.5,0.5);
\draw[thick] (-0.5,0.5) -- (0,1);
\draw[thick] (-0.5,0.5) -- (-0.25,1);
\draw[thick] (-0.5,0.5) -- (-0.75,1);
\draw[thick] (-0.5,0.5) -- (-1,1);
\node (ov) at (-0.5,1.5)  {\footnotesize{$m\text{ legs}$}};
\node (ol) at (-0.5,-0.75)  {\footnotesize{~~}};
\draw[thick] (0.5,0) -- (0.5,1);
\draw[thick] (1,0) -- (1,1);
\node (dottt) at (1.75,0.6) {$\cdots$};
\draw[thick] (2.5,0) -- (2.5,1);
\draw[thick] (3,0) -- (3,1);
\end{tikzpicture}}}
~+~\cdots~+~
\vcenter{\hbox{\begin{tikzpicture}[scale=0.5]
\draw[thick] (0,0) -- (0,1);
\draw[thick] (0.5,0) -- (0.5,1);
\draw[thick] (1,0) -- (1,1);
\node (dottt) at (1.75,0.5) {$\cdots$};
\draw[thick] (2.5,0) -- (2.5,1);
\draw[thick] (3.5,0) -- (3.5,0.5);
\draw[thick] (3.5,0.5) -- (3,1);
\draw[thick] (3.5,0.5) -- (3.25,1);
\draw[thick] (3.5,0.5) -- (3.75,1);
\draw[thick] (3.5,0.5) -- (4,1);
\node (ov) at (3.5,1.5)  {\footnotesize{$m\text{ legs}$}};
\node (ol) at (3.5,-0.75)  {\footnotesize{~~}};
\end{tikzpicture}}}
\bigg)\quad,
\end{flalign}
where the vertex acts on an element as
$\sum_{\und{k}_0,\dots,\und{k}_m\in\bbZ_N^2}\, 
I_{\und{k}_0\und{k}_1\cdots\und{k}_m}\,
\big(e_{\und{k}_0}^\ast , G(\cdot)\big)~ e_{\und{k}_1}^\ast\cdots e_{\und{k}_m}^\ast$.

\begin{ex}\label{ex:free4pt}
To illustrate the pattern of $R$-matrix insertions,
as a first simple example let us compute the $4$-point function
\begin{flalign}
\widetilde{\Pi}(\varphi_1\,\varphi_2\,\varphi_3\,\varphi_4)
\,=\, \Pi\big((\hbar\,\Delta_{\BV}\,\Gamma)^2(\varphi_1\, \varphi_2\, \varphi_3\, \varphi_4)\big)
\end{flalign}
of a non-interacting scalar field, i.e.\ $I=0$.
Using our graphical description, we compute
\begin{flalign}
\hbar\,\Delta_{\BV}\,\Gamma\Big(\,\vcenter{\hbox{\begin{tikzpicture}[scale=0.5]
\draw[thick] (0,0) -- (0,1);
\draw[thick] (0.5,0) -- (0.5,1);
\draw[thick] (1,0) -- (1,1);
\draw[thick] (1.5,0) -- (1.5,1);
\end{tikzpicture}}}\,\Big)
\,=\, -\frac{\hbar}{2}\,\Big(\,
\vcenter{\hbox{\begin{tikzpicture}[scale=0.5]
\draw[thick] (0,0) -- (0,1);
\draw[thick] (0.5,0) -- (0.5,1);
\draw[thick] (1,0) -- (1,1);
\draw[thick] (1.5,0) -- (1.5,1);
\draw[thick] (0,1) to[out=90,in=90] (0.5,1);
\end{tikzpicture}}}
~+~
\vcenter{\hbox{\begin{tikzpicture}[scale=0.5]
\draw[thick] (0,0) -- (0,1);
\draw[thick] (0.5,0) -- (0.5,1);
\draw[thick] (1,0) -- (1,1);
\draw[thick] (1.5,0) -- (1.5,1);
\draw[thick] (0,1) to[out=90,in=90] (1,1);
\end{tikzpicture}}}
~+~
\vcenter{\hbox{\begin{tikzpicture}[scale=0.5]
\draw[thick] (0,0) -- (0,1);
\draw[thick] (0.5,0) -- (0.5,1);
\draw[thick] (1,0) -- (1,1);
\draw[thick] (1.5,0) -- (1.5,1);
\draw[thick] (0,1) to[out=90,in=90] (1.5,1);
\end{tikzpicture}}}
~+~
\vcenter{\hbox{\begin{tikzpicture}[scale=0.5]
\draw[thick] (0,0) -- (0,1);
\draw[thick] (0.5,0) -- (0.5,1);
\draw[thick] (1,0) -- (1,1);
\draw[thick] (1.5,0) -- (1.5,1);
\draw[thick] (0.5,1) to[out=90,in=90] (1,1);
\end{tikzpicture}}}
~+~
\vcenter{\hbox{\begin{tikzpicture}[scale=0.5]
\draw[thick] (0,0) -- (0,1);
\draw[thick] (0.5,0) -- (0.5,1);
\draw[thick] (1,0) -- (1,1);
\draw[thick] (1.5,0) -- (1.5,1);
\draw[thick] (0.5,1) to[out=90,in=90] (1.5,1);
\end{tikzpicture}}}
~+~
\vcenter{\hbox{\begin{tikzpicture}[scale=0.5]
\draw[thick] (0,0) -- (0,1);
\draw[thick] (0.5,0) -- (0.5,1);
\draw[thick] (1,0) -- (1,1);
\draw[thick] (1.5,0) -- (1.5,1);
\draw[thick] (1,1) to[out=90,in=90] (1.5,1);
\end{tikzpicture}}}
\,\Big)
\end{flalign}
and 
\begin{flalign}
\nn (\hbar\,\Delta_{\BV}\,\Gamma)^2\Big(\,\vcenter{\hbox{\begin{tikzpicture}[scale=0.5]
\draw[thick] (0,0) -- (0,1);
\draw[thick] (0.5,0) -- (0.5,1);
\draw[thick] (1,0) -- (1,1);
\draw[thick] (1.5,0) -- (1.5,1);
\end{tikzpicture}}}\,\Big)
\,&=\,\frac{\hbar^2}{2} \, \Big(\,
2 \ \vcenter{\hbox{\begin{tikzpicture}[scale=0.5]
\draw[thick] (0,0) -- (0,1);
\draw[thick] (0.5,0) -- (0.5,1);
\draw[thick] (1,0) -- (1,1);
\draw[thick] (1.5,0) -- (1.5,1);
\draw[thick] (0,1) to[out=90,in=90] (0.5,1);
\draw[thick] (1,1) to[out=90,in=90] (1.5,1);
\end{tikzpicture}}}
~+~
\vcenter{\hbox{\begin{tikzpicture}[scale=0.5]
\draw[thick] (0,0) -- (0,1);
\draw[thick] (0.5,0) -- (0.5,1.1);
\draw[thick] (0.5,1.4) -- (0.5,1.5);
\draw[thick] (1,0) -- (1,1);
\draw[thick] (1.5,0) -- (1.5,1.5);
\draw[thick] (0,1) to[out=90,in=90] (1,1);
\draw[thick] (0.5,1.5) to[out=90,in=90] (1.5,1.5);
\end{tikzpicture}}}
~+~
\vcenter{\hbox{\begin{tikzpicture}[scale=0.5]
\draw[thick] (0,0) -- (0,1);
\draw[thick] (0.5,0) -- (0.5,1.25);
\draw[thick] (0.5,1.5) -- (0.5,1.7);
\draw[thick] (1,0) -- (1,1.25);
\draw[thick] (1,1.5) -- (1,1.7);
\draw[thick] (1.5,0) -- (1.5,1);
\draw[thick] (0,1) to[out=90,in=90] (1.5,1);
\draw[thick] (0.5,1.7) to[out=90,in=90] (1,1.7);
\end{tikzpicture}}}
~+~
\vcenter{\hbox{\begin{tikzpicture}[scale=0.5]
\draw[thick] (0,0) -- (0,1.2);
\draw[thick] (0.5,0) -- (0.5,1);
\draw[thick] (1,0) -- (1,1);
\draw[thick] (1.5,0) -- (1.5,1.2);
\draw[thick] (0.5,1) to[out=90,in=90] (1,1);
\draw[thick] (0,1.2) to[out=90,in=90] (1.5,1.2);
\end{tikzpicture}}}
~+~
\vcenter{\hbox{\begin{tikzpicture}[scale=0.5]
\draw[thick] (0,0) -- (0,1.5);
\draw[thick] (0.5,0) -- (0.5,1);
\draw[thick] (1,0) -- (1,1.1);
\draw[thick] (1,1.4) -- (1,1.5);
\draw[thick] (1.5,0) -- (1.5,1);
\draw[thick] (0.5,1) to[out=90,in=90] (1.5,1);
\draw[thick] (0,1.5) to[out=90,in=90] (1,1.5);
\end{tikzpicture}}}
\,\Big)\\[4pt]
\,&=\, \hbar^2\,\Big(\,
\vcenter{\hbox{\begin{tikzpicture}[scale=0.5]
\draw[thick] (0,0) -- (0,1);
\draw[thick] (0.5,0) -- (0.5,1);
\draw[thick] (1,0) -- (1,1);
\draw[thick] (1.5,0) -- (1.5,1);
\draw[thick] (0,1) to[out=90,in=90] (0.5,1);
\draw[thick] (1,1) to[out=90,in=90] (1.5,1);
\end{tikzpicture}}}
~+~
\vcenter{\hbox{\begin{tikzpicture}[scale=0.5]
\draw[thick] (0,0) -- (0,1);
\draw[thick] (0.5,0) -- (0.5,1);
\draw[thick] (1,0) -- (1,1);
\draw[thick] (1.5,0) -- (1.5,1);
\draw[thick] (0,1) to[out=90,in=90] (1,1);
\draw[thick] (0.5,1) to[out=90,in=90] (1.5,1);
\end{tikzpicture}}}
~+~
\vcenter{\hbox{\begin{tikzpicture}[scale=0.5]
\draw[thick] (0,0) -- (0,1);
\draw[thick] (0.5,0) -- (0.5,1);
\draw[thick] (1,0) -- (1,1);
\draw[thick] (1.5,0) -- (1.5,1);
\draw[thick] (0,1) to[out=90,in=90] (1.5,1);
\draw[thick] (0.5,1) to[out=90,in=90] (1,1);
\end{tikzpicture}}}\,\Big)\quad.\label{eqn:tmpFree4pt}
\end{flalign}
Let us explain in more detail the simplification in the last step:
Using $H$-equivariance of the pairing $(\,\cdot\,, \,\cdot\,)$
and the standard identity $(S\otimes \id_H)R = R^{-1} = R_{21}$ for a triangular $R$-matrix, 
it follows that
\begin{flalign}
\nn \vcenter{\hbox{\begin{tikzpicture}[scale=0.5]
\draw[thick] (0,0) -- (0,1);
\draw[thick] (0.5,0) -- (0.5,1.1);
\draw[thick] (0.5,1.4) -- (0.5,1.5);
\draw[thick] (1,0) -- (1,1);
\draw[thick] (1.5,0) -- (1.5,1.5);
\draw[thick] (0,1) to[out=90,in=90] (1,1);
\draw[thick] (0.5,1.5) to[out=90,in=90] (1.5,1.5);
\end{tikzpicture}}}
~&=\,\big(\varphi_1,R_\alpha\ra G(\varphi_3)\big)~\big(R^\alpha\ra \varphi_2,G(\varphi_4)\big) \\[4pt]
\nn \, &=\,\big(\varphi_1,R_\alpha\ra G(\varphi_3)\big)~\big(\varphi_2,S(R^\alpha)\ra G(\varphi_4)\big)\\[4pt]
\, &=\,  \big(\varphi_1,R^\alpha\ra G( \varphi_3)\big)~\big(\varphi_2,R_\alpha\ra G(\varphi_4)\big)~=~
\vcenter{\hbox{\begin{tikzpicture}[scale=0.5]
\draw[thick] (0,0) -- (0,1.5);
\draw[thick] (0.5,0) -- (0.5,1);
\draw[thick] (1,0) -- (1,1.1);
\draw[thick] (1,1.4) -- (1,1.5);
\draw[thick] (1.5,0) -- (1.5,1);
\draw[thick] (0.5,1) to[out=90,in=90] (1.5,1);
\draw[thick] (0,1.5) to[out=90,in=90] (1,1.5);
\end{tikzpicture}}}\quad,\label{eqn:exargument1}
\end{flalign}
hence the second and fifth term in the first line of \eqref{eqn:tmpFree4pt} coincide,
yielding the second term in the second line. Furthermore, using again $H$-equivariance
of the pairing $(\,\cdot\, ,\,\cdot\,)$, the third $R$-matrix property in \eqref{eqn:Rmatrixproperties}
and the normalization condition $\epsilon(R^\alpha)\,R_\alpha = 1$, it follows that
\begin{flalign}
\nn \vcenter{\hbox{\begin{tikzpicture}[scale=0.5]
\draw[thick] (0,0) -- (0,1);
\draw[thick] (0.5,0) -- (0.5,1.25);
\draw[thick] (0.5,1.5) -- (0.5,1.7);
\draw[thick] (1,0) -- (1,1.25);
\draw[thick] (1,1.5) -- (1,1.7);
\draw[thick] (1.5,0) -- (1.5,1);
\draw[thick] (0,1) to[out=90,in=90] (1.5,1);
\draw[thick] (0.5,1.7) to[out=90,in=90] (1,1.7);
\end{tikzpicture}}}
~&=\,\big(\varphi_1,R_{\alpha}\,R_{\beta}\ra G(\varphi_4)\big)~\big(R^\alpha\ra \varphi_2,R^\beta\ra G(\varphi_3)\big)\\[4pt]
\nn \,&=\,\big(\varphi_1,R_{\alpha}\ra G(\varphi_4)\big)~\big(R^\alpha_{\und{1}} \ra \varphi_2,R^\alpha_{\und{2}}\ra G(\varphi_3)\big)\\[4pt]
\nn \,&=\,\big(\varphi_1,R_{\alpha}\ra G(\varphi_4)\big)~\epsilon(R^\alpha)\,\big(\varphi_2, G(\varphi_3)\big)\\[4pt]
\,&=\, \big(\varphi_1,G(\varphi_4)\big)~\big(\varphi_2, G(\varphi_3)\big)
~=~\vcenter{\hbox{\begin{tikzpicture}[scale=0.5]
\draw[thick] (0,0) -- (0,1.2);
\draw[thick] (0.5,0) -- (0.5,1);
\draw[thick] (1,0) -- (1,1);
\draw[thick] (1.5,0) -- (1.5,1.2);
\draw[thick] (0.5,1) to[out=90,in=90] (1,1);
\draw[thick] (0,1.2) to[out=90,in=90] (1.5,1.2);
\end{tikzpicture}}}\quad,\label{eqn:exargument2}
\end{flalign}
hence the third and fourth term in the first line of \eqref{eqn:tmpFree4pt} coincide,
yielding the third term in the second line.
\sk

From this we compute the free $4$-point function
\begin{flalign}
\nn \widetilde{\Pi}(\varphi_1\,\varphi_2\,\varphi_3\,\varphi_4)\,&=\,\hbar^2\,\Big(\,
\vcenter{\hbox{\begin{tikzpicture}[scale=0.5]
\draw[thick] (0,0) -- (0,1);
\draw[thick] (0.5,0) -- (0.5,1);
\draw[thick] (1,0) -- (1,1);
\draw[thick] (1.5,0) -- (1.5,1);
\draw[thick] (0,1) to[out=90,in=90] (0.5,1);
\draw[thick] (1,1) to[out=90,in=90] (1.5,1);
\end{tikzpicture}}}
~+~
\vcenter{\hbox{\begin{tikzpicture}[scale=0.5]
\draw[thick] (0,0) -- (0,1);
\draw[thick] (0.5,0) -- (0.5,1);
\draw[thick] (1,0) -- (1,1);
\draw[thick] (1.5,0) -- (1.5,1);
\draw[thick] (0,1) to[out=90,in=90] (1,1);
\draw[thick] (0.5,1) to[out=90,in=90] (1.5,1);
\end{tikzpicture}}}
~+~
\vcenter{\hbox{\begin{tikzpicture}[scale=0.5]
\draw[thick] (0,0) -- (0,1);
\draw[thick] (0.5,0) -- (0.5,1);
\draw[thick] (1,0) -- (1,1);
\draw[thick] (1.5,0) -- (1.5,1);
\draw[thick] (0,1) to[out=90,in=90] (1.5,1);
\draw[thick] (0.5,1) to[out=90,in=90] (1,1);
\end{tikzpicture}}}\,\Big)\\[4pt]
\nn \,&=\,\hbar^2\, \Big(\big(
\varphi_1,G(\varphi_2)\big)~\big(\varphi_3,G(\varphi_4)\big)
~+~
\big(\varphi_1,R_\alpha\ra G(\varphi_3)\big)~\big(R^\alpha\ra \varphi_2,G(\varphi_4)\big)\\
&\qquad~\qquad ~+~\big(\varphi_1,G(\varphi_4)\big)~\big(\varphi_2,G(\varphi_3)\big)
\Big)
\end{flalign}
and observe that there is a single appearance of an $R$-matrix associated with
the line crossing in the second term. Note that,
as in the case of Oeckl's \emph{symmetric} braided quantum field theory~\cite{Oec00},
we do not have to distinguish between over and under crossings because
our $R$-matrix is triangular.
\end{ex}

\begin{ex}\label{ex:torus2pt}
Let us now set $m=3$ and compute the $2$-point function
\begin{flalign}
\widetilde{\Pi}\big(\varphi_1\,\varphi_2\big)\,=\,\sum_{k=0}^\infty\, \Pi\big((\delta\,\Gamma)^k(\varphi_1\,\varphi_2)\big)
\end{flalign}
of $\Phi^4$-theory to the lowest non-trivial order 
in the coupling constant, which we recall is $\mathcal{O}(\lambda^2)$ due to our conventions 
in \eqref{eqn:torusinteraction}. The graphical expansion (for the moment without the simplifications)
is completely analogous to the calculation on the fuzzy sphere from Example \ref{ex:fuzzysphere2pt},
from which we obtain
\begin{flalign}
\nn \widetilde{\Pi}\big(\varphi_1\,\varphi_2\big)\,&=\,-\hbar ~~\vcenter{\hbox{\begin{tikzpicture}[scale=0.5]
\draw[thick] (0,0) -- (0,1);
\draw[thick] (0.5,0) -- (0.5,1);
\draw[thick] (0,1) to[out=90,in=90] (0.5,1);
\end{tikzpicture}}} ~ - ~ \frac{\lambda^2\,\hbar^2}{3!\, 4}\,\Big(\,
2 \ \vcenter{\hbox{\begin{tikzpicture}[scale=0.5]
\draw[thick] (-0.5,0) -- (-0.5,0.5);
\draw[thick] (-0.5,0.5) -- (0,1);
\draw[thick] (-0.5,0.5) -- (-0.5,1);
\draw[thick] (-0.5,0.5) -- (-1,1);
\draw[thick] (0.5,0) -- (0.5,1);
\draw[thick] (-1,1) to[out=90,in=90] (-0.5,1);
\draw[thick] (0,1) to[out=90,in=90] (0.5,1);
\end{tikzpicture}}}
~+~
\vcenter{\hbox{\begin{tikzpicture}[scale=0.5]
\draw[thick] (-0.5,0) -- (-0.5,0.5);
\draw[thick] (-0.5,0.5) -- (0,1);
\draw[thick] (-0.5,0.5) -- (-0.5,1.2);
\draw[thick] (-0.5,1.4) -- (-0.5,1.5);
\draw[thick] (-0.5,0.5) -- (-1,1);
\draw[thick] (0.5,0) -- (0.5,1.5);
\draw[thick] (-1,1) to[out=90,in=90] (0,1);
\draw[thick] (-0.5,1.5) to[out=90,in=90] (0.5,1.5);
\end{tikzpicture}}}
~+~
\vcenter{\hbox{\begin{tikzpicture}[scale=0.5]
\draw[thick] (-0.5,0) -- (-0.5,0.5);
\draw[thick] (-0.5,0.5) -- (0,1.3);
\draw[thick] (0,1.5) -- (0,1.6);
\draw[thick] (-0.5,0.5) -- (-0.5,1.3);
\draw[thick] (-0.5,1.5) -- (-0.5,1.6);
\draw[thick] (-0.5,0.5) -- (-1,1);
\draw[thick] (0.5,0) -- (0.5,1);
\draw[thick] (-1,1) to[out=90,in=90] (0.5,1);
\draw[thick] (-0.5,1.6) to[out=90,in=90] (0,1.6);
\end{tikzpicture}}}
~+~
\vcenter{\hbox{\begin{tikzpicture}[scale=0.5]
\draw[thick] (-0.5,0) -- (-0.5,0.5);
\draw[thick] (-0.5,0.5) -- (0,1);
\draw[thick] (-0.5,0.5) -- (-0.5,1);
\draw[thick] (-0.5,0.5) -- (-1,1);
\draw[thick] (0.5,0) -- (0.5,1);
\draw[thick] (-0.5,1) to[out=90,in=90] (0,1);
\draw[thick] (-1,1) to[out=90,in=90] (0.5,1);
\end{tikzpicture}}}
~+~
\vcenter{\hbox{\begin{tikzpicture}[scale=0.5]
\draw[thick] (-0.5,0) -- (-0.5,0.5);
\draw[thick] (-0.5,0.5) -- (0,1);
\draw[thick] (0,1) -- (0,1.2);
\draw[thick] (0,1.4) -- (0,1.5);
\draw[thick] (-0.5,0.5) -- (-0.5,1);
\draw[thick] (-0.5,0.5) -- (-1,1);
\draw[thick] (-1,1) -- (-1,1.5);
\draw[thick] (0.5,0) -- (0.5,1);
\draw[thick] (-0.5,1) to[out=90,in=90] (0.5,1);
\draw[thick] (-1,1.5) to[out=90,in=90] (0,1.5);
\end{tikzpicture}}}
\\
&\hspace{1.2cm}~+~2 \ 
\vcenter{\hbox{\begin{tikzpicture}[scale=0.5]
\draw[thick] (-1.5,0) -- (-1.5,1);
\draw[thick] (-0.5,0) -- (-0.5,0.5);
\draw[thick] (-0.5,0.5) -- (0,1);
\draw[thick] (-0.5,0.5) -- (-0.5,1);
\draw[thick] (-0.5,0.5) -- (-1,1);
\draw[thick] (-1.5,1) to[out=90,in=90] (-1,1);
\draw[thick] (-0.5,1) to[out=90,in=90] (0,1);
\end{tikzpicture}}}
~+~
\vcenter{\hbox{\begin{tikzpicture}[scale=0.5]
\draw[thick] (-1.5,0) -- (-1.5,1);
\draw[thick] (-0.5,0) -- (-0.5,0.5);
\draw[thick] (-0.5,0.5) -- (0,1);
\draw[thick] (0,1) -- (0,1.5);
\draw[thick] (-0.5,0.5) -- (-0.5,1);
\draw[thick] (-0.5,0.5) -- (-1,1);
\draw[thick] (-1,1) -- (-1,1.2);
\draw[thick] (-1,1.4) -- (-1,1.5);
\draw[thick] (-1.5,1) to[out=90,in=90] (-0.5,1);
\draw[thick] (-1,1.5) to[out=90,in=90] (0,1.5);
\end{tikzpicture}}}
~+~ 
\vcenter{\hbox{\begin{tikzpicture}[scale=0.5]
\draw[thick] (-1.5,0) -- (-1.5,1);
\draw[thick] (-0.5,0) -- (-0.5,0.5);
\draw[thick] (-0.5,0.5) -- (0,1);
\draw[thick] (-0.5,0.5) -- (-0.5,1);
\draw[thick] (-0.5,1) -- (-0.5,1.3);
\draw[thick] (-0.5,1.5) -- (-0.5,1.6);
\draw[thick] (-0.5,0.5) -- (-1,1);
\draw[thick] (-1,1) -- (-1,1.3);
\draw[thick] (-1,1.5) -- (-1,1.6);
\draw[thick] (-1.5,1) to[out=90,in=90] (0,1);
\draw[thick] (-1,1.6) to[out=90,in=90] (-0.5,1.6);
\end{tikzpicture}}}
~+~
\vcenter{\hbox{\begin{tikzpicture}[scale=0.5]
\draw[thick] (-1.5,0) -- (-1.5,1);
\draw[thick] (-0.5,0) -- (-0.5,0.5);
\draw[thick] (-0.5,0.5) -- (0,1);
\draw[thick] (-0.5,0.5) -- (-0.5,1);
\draw[thick] (-0.5,0.5) -- (-1,1);
\draw[thick] (-1,1) to[out=90,in=90] (-0.5,1);
\draw[thick] (-1.5,1) to[out=90,in=90] (0,1);
\end{tikzpicture}}}
~+~
\vcenter{\hbox{\begin{tikzpicture}[scale=0.5]
\draw[thick] (-1.5,0) -- (-1.5,1.5);
\draw[thick] (-0.5,0) -- (-0.5,0.5);
\draw[thick] (-0.5,0.5) -- (0,1);
\draw[thick] (-0.5,0.5) -- (-0.5,1.2);
\draw[thick] (-0.5,1.4) -- (-0.5,1.5);
\draw[thick] (-0.5,0.5) -- (-1,1);
\draw[thick] (-1,1) to[out=90,in=90] (0,1);
\draw[thick] (-1.5,1.5) to[out=90,in=90] (-0.5,1.5);
\end{tikzpicture}}}
\,\Big)+\mathcal{O}(\lambda^4)\\[4pt]
\nn \, &=\,-\hbar ~~\vcenter{\hbox{\begin{tikzpicture}[scale=0.5]
\draw[thick] (0,0) -- (0,1);
\draw[thick] (0.5,0) -- (0.5,1);
\draw[thick] (0,1) to[out=90,in=90] (0.5,1);
\end{tikzpicture}}} ~ - ~
\frac{\lambda^2\,\hbar^2}{3!\,2}\,\Big(\,
\vcenter{\hbox{\begin{tikzpicture}[scale=0.5]
\draw[thick] (-0.5,0) -- (-0.5,0.5);
\draw[thick] (-0.5,0.5) -- (0,1);
\draw[thick] (-0.5,0.5) -- (-0.5,1);
\draw[thick] (-0.5,0.5) -- (-1,1);
\draw[thick] (0.5,0) -- (0.5,1);
\draw[thick] (-1,1) to[out=90,in=90] (-0.5,1);
\draw[thick] (0,1) to[out=90,in=90] (0.5,1);
\end{tikzpicture}}}
~+~
\vcenter{\hbox{\begin{tikzpicture}[scale=0.5]
\draw[thick] (-0.5,0) -- (-0.5,0.5);
\draw[thick] (-0.5,0.5) -- (0,1);
\draw[thick] (-0.5,0.5) -- (-0.5,1.2);
\draw[thick] (-0.5,1.4) -- (-0.5,1.5);
\draw[thick] (-0.5,0.5) -- (-1,1);
\draw[thick] (0.5,0) -- (0.5,1.5);
\draw[thick] (-1,1) to[out=90,in=90] (0,1);
\draw[thick] (-0.5,1.5) to[out=90,in=90] (0.5,1.5);
\end{tikzpicture}}}
~+~
\vcenter{\hbox{\begin{tikzpicture}[scale=0.5]
\draw[thick] (-0.5,0) -- (-0.5,0.5);
\draw[thick] (-0.5,0.5) -- (0,1);
\draw[thick] (-0.5,0.5) -- (-0.5,1);
\draw[thick] (-0.5,0.5) -- (-1,1);
\draw[thick] (0.5,0) -- (0.5,1);
\draw[thick] (-0.5,1) to[out=90,in=90] (0,1);
\draw[thick] (-1,1) to[out=90,in=90] (0.5,1);
\end{tikzpicture}}}
~+~
\vcenter{\hbox{\begin{tikzpicture}[scale=0.5]
\draw[thick] (-1.5,0) -- (-1.5,1);
\draw[thick] (-0.5,0) -- (-0.5,0.5);
\draw[thick] (-0.5,0.5) -- (0,1);
\draw[thick] (-0.5,0.5) -- (-0.5,1);
\draw[thick] (-0.5,0.5) -- (-1,1);
\draw[thick] (-1.5,1) to[out=90,in=90] (-1,1);
\draw[thick] (-0.5,1) to[out=90,in=90] (0,1);
\end{tikzpicture}}}
~+~
\vcenter{\hbox{\begin{tikzpicture}[scale=0.5]
\draw[thick] (-1.5,0) -- (-1.5,1);
\draw[thick] (-0.5,0) -- (-0.5,0.5);
\draw[thick] (-0.5,0.5) -- (0,1);
\draw[thick] (-0.5,0.5) -- (-0.5,1);
\draw[thick] (-0.5,0.5) -- (-1,1);
\draw[thick] (-1,1) to[out=90,in=90] (-0.5,1);
\draw[thick] (-1.5,1) to[out=90,in=90] (0,1);
\end{tikzpicture}}}
~+~ 
\vcenter{\hbox{\begin{tikzpicture}[scale=0.5]
\draw[thick] (-1.5,0) -- (-1.5,1.5);
\draw[thick] (-0.5,0) -- (-0.5,0.5);
\draw[thick] (-0.5,0.5) -- (0,1);
\draw[thick] (-0.5,0.5) -- (-0.5,1.2);
\draw[thick] (-0.5,1.4) -- (-0.5,1.5);
\draw[thick] (-0.5,0.5) -- (-1,1);
\draw[thick] (-1,1) to[out=90,in=90] (0,1);
\draw[thick] (-1.5,1.5) to[out=90,in=90] (-0.5,1.5);
\end{tikzpicture}}}
\,\Big)+\mathcal{O}(\lambda^4)\quad,
\end{flalign}
where the simplifications in the last step use the same arguments as in Example \ref{ex:free4pt},
see in particular \eqref{eqn:exargument1} and \eqref{eqn:exargument2}.
\sk

Using the properties \eqref{eqn:torusinteractionproperties} of the interaction term, 
one can show that all six loop contributions coincide. To illustrate the relevant arguments,
let us show explicitly that the second and the third loop contributions coincide with the first one;
the other terms follow by similar arguments. For the second term, we find
\begin{flalign}
\nn \vcenter{\hbox{\begin{tikzpicture}[scale=0.5]
\draw[thick] (-0.5,0) -- (-0.5,0.5);
\draw[thick] (-0.5,0.5) -- (0,1);
\draw[thick] (-0.5,0.5) -- (-0.5,1.2);
\draw[thick] (-0.5,1.4) -- (-0.5,1.5);
\draw[thick] (-0.5,0.5) -- (-1,1);
\draw[thick] (0.5,0) -- (0.5,1.5);
\draw[thick] (-1,1) to[out=90,in=90] (0,1);
\draw[thick] (-0.5,1.5) to[out=90,in=90] (0.5,1.5);
\end{tikzpicture}}}
~&=\, \sum_{\und{k}_0,\und{k}_1,\und{k}_2,\und{k}_3\in\bbZ_N^2}\, I_{\und{k}_0\und{k}_1\und{k}_2\und{k}_3}
~\big(e_{\und{k}_0}^\ast,G(\varphi_1)\big)~\big(e_{\und{k}_1}^\ast, G(R_\alpha\ra e_{\und{k}_3}^\ast)\big)~
\big(R^\alpha\ra e_{\und{k}_2}^\ast, G(\varphi_2)\big)\\[4pt]
\nn ~&=\,\sum_{\und{k}_0,\und{k}_1,\und{k}_2,\und{k}_3\in\bbZ_N^2}\,q^{-\und{k}_2\Theta\und{k}_3}~ I_{\und{k}_0\und{k}_1\und{k}_2\und{k}_3}
~\big(e_{\und{k}_0}^\ast,G(\varphi_1)\big)~\big(e_{\und{k}_1}^\ast, G(e_{\und{k}_3}^\ast)\big)~
\big(e_{\und{k}_2}^\ast, G(\varphi_2)\big)\\[4pt]
~&=\,\sum_{\und{k}_0,\und{k}_1,\und{k}_2,\und{k}_3\in\bbZ_N^2}\, I_{\und{k}_0\und{k}_1\und{k}_2\und{k}_3}
~\big(e_{\und{k}_0}^\ast,G(\varphi_1)\big)~\big(e_{\und{k}_1}^\ast, G(e_{\und{k}_2}^\ast)\big)~
\big(e_{\und{k}_3}^\ast, G(\varphi_2)\big)\,=~\vcenter{\hbox{\begin{tikzpicture}[scale=0.5]
\draw[thick] (-0.5,0) -- (-0.5,0.5);
\draw[thick] (-0.5,0.5) -- (0,1);
\draw[thick] (-0.5,0.5) -- (-0.5,1);
\draw[thick] (-0.5,0.5) -- (-1,1);
\draw[thick] (0.5,0) -- (0.5,1);
\draw[thick] (-1,1) to[out=90,in=90] (-0.5,1);
\draw[thick] (0,1) to[out=90,in=90] (0.5,1);
\end{tikzpicture}}}\quad,
\end{flalign}
where in the first step we used the properties 
\eqref{eqn:planewavebasis2} of the (dual) basis $e_{\und{k}}^\ast$
and in the second step we used the $q$-deformed symmetry property \eqref{eqn:torusinteractionproperties} 
of the interaction term. 
For the third term, we find
\begin{flalign}
\nn \vcenter{\hbox{\begin{tikzpicture}[scale=0.5]
\draw[thick] (-0.5,0) -- (-0.5,0.5);
\draw[thick] (-0.5,0.5) -- (0,1);
\draw[thick] (-0.5,0.5) -- (-0.5,1);
\draw[thick] (-0.5,0.5) -- (-1,1);
\draw[thick] (0.5,0) -- (0.5,1);
\draw[thick] (-0.5,1) to[out=90,in=90] (0,1);
\draw[thick] (-1,1) to[out=90,in=90] (0.5,1);
\end{tikzpicture}}}
~&=\, \sum_{\und{k}_0,\und{k}_1,\und{k}_2,\und{k}_3\in\bbZ_N^2}\, I_{\und{k}_0\und{k}_1\und{k}_2\und{k}_3}
~\big(e_{\und{k}_0}^\ast,G(\varphi_1)\big)~\big(e_{\und{k}_2}^\ast, G(e_{\und{k}_3}^\ast)\big)~
\big(e_{\und{k}_1}^\ast, G(\varphi_2)\big)\\[4pt]
\nn ~&=\, \sum_{\und{k}_0,\und{k}_1,\und{k}_2,\und{k}_3\in\bbZ_N^2}\, I_{\und{k}_0\und{k}_1\und{k}_2\und{k}_3}
~\big(e_{\und{k}_0}^\ast,G(\varphi_1)\big)~\big({R_{\alpha}}_{\und{1}} \ra 
e_{\und{k}_2}^\ast, G({R_{\alpha}}_{\und{2}}\ra e_{\und{k}_3}^\ast)\big)~
\big(R^\alpha\ra e_{\und{k}_1}^\ast, G(\varphi_2)\big)\\[4pt]
\nn ~&=\, \sum_{\und{k}_0,\und{k}_1,\und{k}_2,\und{k}_3\in\bbZ_N^2}\, I_{\und{k}_0\und{k}_1\und{k}_2\und{k}_3}
~\big(e_{\und{k}_0}^\ast,G(\varphi_1)\big)~\big({R_{\alpha}} \ra 
e_{\und{k}_2}^\ast, G({R_{\beta}}\ra e_{\und{k}_3}^\ast)\big)~
\big(R^\beta\, R^\alpha\ra e_{\und{k}_1}^\ast, G(\varphi_2)\big)\\[4pt]
~&=\,\sum_{\und{k}_0,\und{k}_1,\und{k}_2,\und{k}_3\in\bbZ_N^2}\, I_{\und{k}_0\und{k}_1\und{k}_2\und{k}_3}
~\big(e_{\und{k}_0}^\ast,G(\varphi_1)\big)~\big(e_{\und{k}_1}^\ast, G(e_{\und{k}_2}^\ast)\big)~
\big(e_{\und{k}_3}^\ast, G(\varphi_2)\big)\,=~\vcenter{\hbox{\begin{tikzpicture}[scale=0.5]
\draw[thick] (-0.5,0) -- (-0.5,0.5);
\draw[thick] (-0.5,0.5) -- (0,1);
\draw[thick] (-0.5,0.5) -- (-0.5,1);
\draw[thick] (-0.5,0.5) -- (-1,1);
\draw[thick] (0.5,0) -- (0.5,1);
\draw[thick] (-1,1) to[out=90,in=90] (-0.5,1);
\draw[thick] (0,1) to[out=90,in=90] (0.5,1);
\end{tikzpicture}}}\quad.
\end{flalign}
In the first step we used $H$-equivariance of $(\,\cdot\,,G(\cdot))$
together with the normalization condition $R^\alpha\,\epsilon(R_\alpha)=1$.
The second step follows from the second identity in \eqref{eqn:Rmatrixproperties},
and the third step applies the $q$-deformed symmetry property \eqref{eqn:torusinteractionproperties} 
of the interaction term twice. 
\sk

Altogether we find that the $2$-point function of $\Phi^4$-theory on the fuzzy torus
reads to leading order in the coupling constant as
\begin{flalign}
\nn \widetilde{\Pi}\big(\varphi_1\,\varphi_2\big)\,&=\,-\hbar ~~\vcenter{\hbox{\begin{tikzpicture}[scale=0.5]
\draw[thick] (0,0) -- (0,1);
\draw[thick] (0.5,0) -- (0.5,1);
\draw[thick] (0,1) to[out=90,in=90] (0.5,1);
\end{tikzpicture}}}
~-\, \frac{\lambda^2\,\hbar^2}{2}~~\vcenter{\hbox{\begin{tikzpicture}[scale=0.5]
\draw[thick] (-0.5,0) -- (-0.5,0.5);
\draw[thick] (-0.5,0.5) -- (0,1);
\draw[thick] (-0.5,0.5) -- (-0.5,1);
\draw[thick] (-0.5,0.5) -- (-1,1);
\draw[thick] (0.5,0) -- (0.5,1);
\draw[thick] (-1,1) to[out=90,in=90] (-0.5,1);
\draw[thick] (0,1) to[out=90,in=90] (0.5,1);
\end{tikzpicture}}} ~+~\mathcal{O}(\lambda^4)\\[4pt]
\,&=\, -\hbar\,\big(\varphi_1,G(\varphi_2)\big) - \frac{\lambda^2\,\hbar^2}{2}\,
\sum_{\und{k},\und{l}\in\bbZ_N^2}\, \frac{\big(e_{\und{k}}^\ast,G(\varphi_1)\big)~\big(e_{\und{k}}, G(\varphi_2)\big)}{{[l_1]_q^2 + [l_2]_q^2 +m^2}} \,+\,\mathcal{O}(\lambda^4)\quad,\label{eqn:torus2ptexplicit}
\end{flalign}
where we also used the explicit expression \eqref{eqnLtorusGreenoperator} for the Green operator.
We emphasize that in our interacting $2$-point function \eqref{eqn:torus2ptexplicit}
there is no distinction between planar and non-planar loop corrections, 
in contrast to the traditional (unbraided) 
approaches to noncommutative quantum field theory~\cite{IIKK00,MVRS00}.
This feature is intimately tied to the braided
commutativity property of the fuzzy torus algebra \eqref{eqn:fuzzytorusalgebra},
which in particular implies that the higher $L_\infty$-algebra brackets \eqref{eqn:torusellm} 
are automatically braided (graded anti-)symmetric without the need for graded antisymmetrization
as in the fuzzy sphere case \eqref{eqn:ellmscalar}. The absence of non-planar 
features in loop corrections has similarly been observed by Oeckl in his framework of 
(symmetric) braided quantum field theory~\cite{Oec00}.
\end{ex}

\begin{ex}
We stress that the disappearance of the $q$-factors from the interaction
term  \eqref{eqn:torusinteraction}
in the tadpole diagram \eqref{eqn:torus2ptexplicit} is only due to the shape
of this diagram and not a general feature of our formalism.
As an illustrative example, let us consider again
$\Phi^4$-theory and compute the {\em connected part}
of the $4$-point function to the first non-trivial order in the coupling constant $\lambda^2$.
Using the same arguments as in the previous examples, 
in particular the $q$-deformed symmetry property \eqref{eqn:torusinteractionproperties} 
of the interaction term, one finds
\begin{flalign}
&\widetilde{\Pi}(\varphi_1\,\varphi_2\,\varphi_3\,\varphi_4)_{\text{connected}}
\,=\,\lambda^2 \,\hbar^3\,~
\vcenter{\hbox{\begin{tikzpicture}[scale=0.5]
\draw[thick] (-0.5,0) -- (-0.5,0.5);
\draw[thick] (-0.5,0.5) -- (0,1);
\draw[thick] (-0.5,0.5) -- (-0.5,1);
\draw[thick] (-0.5,0.5) -- (-1,1);
\draw[thick] (0.5,0) -- (0.5,1);
\draw[thick] (1,0) -- (1,1);
\draw[thick] (1.5,0) -- (1.5,1);
\draw[thick] (-1,1) to[out=90,in=90] (1.5,1);
\draw[thick] (-0.5,1) to[out=90,in=90] (1,1);
\draw[thick] (0,1) to[out=90,in=90] (0.5,1);
\end{tikzpicture}}}
~+~\mathcal{O}(\lambda^4)\\[4pt]
\nn &~~=\lambda^2 \,\hbar^3 \!\!\! \sum_{\und{k}_0,\und{k}_1,\und{k}_2,\und{k}_3\in\bbZ_N^2}\!\!\!
I_{\und{k}_0\und{k}_1\und{k}_2\und{k}_3}\,\big(e_{\und{k}_0}^\ast,G(\varphi_1)\big)
~\big(e_{\und{k}_1}^\ast,G(\varphi_4)\big)~
\big(e_{\und{k}_2}^\ast,G(\varphi_3)\big)~
\big(e_{\und{k}_3}^\ast,G(\varphi_2)\big)+\mathcal{O}(\lambda^4)\quad.
\end{flalign}
Recalling now the explicit expression \eqref{eqn:torusinteractionexplicit} 
for the constants $I_{\und{k}_0\und{k}_1\und{k}_2\und{k}_3}$, 
we see that this correlation function includes the expected $q$-factors 
from the interaction term.
\end{ex}


\section{\label{sec:conclusion}Concluding remarks and outlook}
In this paper we have initiated the study of noncommutative quantum field theories
using modern tools from BV quantization~\cite{CG16,Gwi12}. We focused on
the case of fuzzy field theories, which are finite-dimensional models and therefore
do not require regularization and renormalization. We discussed two different
flavors of models, namely ordinary noncommutative field theories (see Sections \ref{sec:prelim}
and \ref{sec:fuzzysphere}) and so-called `braided' noncommutative field theories 
(see Sections \ref{sec:braidedBV} and \ref{sec:fuzzytorus}). Our ordinary noncommutative
field theories are described at the classical level by ordinary $L_\infty$-algebras
as in~\cite{BBKL18} and, as illustrated in the present paper, their BV quantization can 
be carried out using precisely the same methods as in commutative field theory~\cite{CG16,Gwi12}.
We would like to emphasize that this does {\em not} mean that such theories
are insensitive to noncommutative geometry, which enters through the explicit form
of the propagators and the interaction vertices. In particular, we recover from
our formalism well-known noncommutative features such as non-planar contributions
to loop diagrams in noncommutative scalar field theories, see Example \ref{ex:fuzzysphere2pt}.
\sk

The second flavor of models we have studied are the so-called `braided' noncommutative theories,
which are described at the classical level by the `braided $L_\infty$-algebras' proposed 
in~\cite{DCGRS20,DCGRS21}. These are field theories that are defined in the representation
category of a Hopf algebra that is endowed with a non-identity {\em triangular} $R$-matrix. 
The BV quantization techniques of~\cite{CG16,Gwi12} generalize in a rather 
straightforward way to these models (because the representation category
of a triangular Hopf algebra is a symmetric braided monoidal category) and we have spelled
out the details in Section \ref{sec:braidedBV}. When applied to 
a scalar field, our approach coincides with the (symmetric) braided quantum field theory of 
Oeckl~\cite{Oec00}. In particular, we observed in Example \ref{ex:torus2pt} that 
non-planar contributions to loop diagrams are absent in the braided framework.
\sk

The main lesson of our paper is that modern BV quantization
as in~\cite{CG16,Gwi12} provides a collection of systematic and powerful tools 
to study both ordinary and `braided' noncommutative quantum field theories.
In particular, well-known noncommutative features, such as the appearance of non-planar loop
contributions or the absence of those in `braided' field theories, are recovered from this approach.
Compared to more traditional approaches, the main advantage of these more abstract BV quantization techniques
is that they readily apply to noncommutative gauge theories, as we have illustrated with 
an explicit example in Section \ref{subsec:CS}.
\sk

We believe that there are two particularly interesting avenues for future research.
Firstly, it would be interesting to study the regularization and renormalization
of noncommutative quantum field theories on infinite-dimensional algebras,
e.g.\ the Moyal plane, from our point of view of BV quantization. This requires
an adaption of the analytical aspects of Costello and Gwilliam's work~\cite{CG16,Gwi12}
and an understanding of their interplay with noncommutative phenomena such as UV/IR-mixing.
Secondly, it would be interesting to generalize the `braided BV formalism' in Section \ref{sec:braidedBV}
to the truly (i.e.\ non-symmetric) braided case governed by quasi-triangular Hopf algebras.
We expect that this will be considerably more difficult than the symmetric braided
case discussed in this paper, because it is not at all straightforward to define 
truly braided analogs of the relevant algebraic structures, such as $P_0$-algebras and the BV Laplacian.
For instance, already the definition of truly braided analogs of Lie algebras \cite{Maj94} 
is rather non-intuitive and involved, and we expect that this will also be the case
for truly braided $P_0$-algebras. As a minimalistic approach, one could skip
the $P_0$-algebras and try to generalize the explicit form of the BV Laplacian 
\eqref{eqn:braidedBVLaplacian} to the case of a quasi-triangular $R$-matrix.
The resulting formula agrees with Oeckl's braided Wick Theorem \cite{Oec01},
but unfortunately in the truly braided case there are obstructions (due to quasi-triangularity)
to the important square-zero condition $\Delta_{\BV}^2=0$ of the BV Laplacian.
We currently do not know how to resolve these issues and why they seem to play no role
in Oeckl's truly braided approach.


\section*{Acknowledgments}
We thank Robert Laugwitz, Christian S\"amann and Martin
Wolf for helpful discussions.
H.N.\ is supported by a PhD Scholarship from the School of Mathematical 
Sciences of the University of Nottingham. 
A.S.\ gratefully acknowledges the financial support of 
the Royal Society (UK) through a Royal Society University 
Research Fellowship (UF150099), a Research Grant (RG160517) 
and two Enhancement Awards (RGF\textbackslash EA\textbackslash 180270 
and RGF\textbackslash EA\textbackslash 201051). 
The work of R.J.S.\ was supported by
the Consolidated Grant ST/P000363/1 
from the UK Science and Technology Facilities Council.



\end{document}